%% file: tree-tr-paper.tex
\documentclass[11pt]{article}
\usepackage[margin=1in]{geometry}
\include{preamble}
\begin{document}
	\title{Reconstructing Trees from Traces\footnote{An extended abstract of this paper appears in the Proceedings of the 32nd Conference on Learning Theory (COLT), 2019~\cite{DRR19}.}}
	\date{\today}
	
	\author{Sami Davies%
		\thanks{%
			{University of Washington (\url{daviess@uw.edu}); research supported by NSF CAREER grant 1651861 and the David \& Lucile Packard Foundation.}}
		\and Mikl\'os Z. R\'acz%
		\thanks{%
			{Princeton University (\url{mracz@princeton.edu}); research supported in part by NSF grant DMS 1811724.}}
		\and Cyrus Rashtchian%
		\thanks{%
			{Dept. of Computer Science \& Engineering, University of California, San Diego (\url{crashtchian@eng.ucsd.edu})}.
	}}
	
	\maketitle

\begin{abstract}%
We study the problem of learning a node-labeled tree given independent traces from an appropriately defined deletion channel.
This problem, tree trace reconstruction, generalizes string trace reconstruction,
which corresponds to the tree being a path. 
For many classes of trees, including complete trees and spiders, 
we provide algorithms that reconstruct the labels using only a polynomial number of traces. 
This exhibits a stark contrast to known results on string trace reconstruction, 
which require exponentially many traces, 
and where a central open problem is to determine whether a polynomial number of traces suffice. 
Our techniques combine novel combinatorial and complex analytic methods. 
\end{abstract}


\section{Introduction}

Statistical reconstruction problems aim to recover unknown objects given only noisy samples of the data. 
In the \emph{string trace reconstruction} problem, there is an unknown binary string, and we observe noisy samples of this string after it has gone through a deletion channel.
 This deletion channel independently deletes each bit with constant probability $q$ and concatenates the remaining bits.
The channel preserves bit order, so we observe a sampled subsequence known as a \emph{trace}.
The goal is to learn the original string with high probability using as few traces as possible.

The string trace reconstruction problem (with insertions, substitutions, and deletions) 
directly appears in the problem of DNA Data Storage~\cite{church2012next,ceze2019molecular, erlich2017dna, goldman,OAC,yazdi2017portable, yazdi2015dna}. 
It is crucial to minimize the sample complexity, as this directly impacts the cost of retrieving data stored in synthetic DNA.
Since there is an exponential gap between upper and lower bounds for the string trace reconstruction problem, 
it is motivating to study variants.
We introduce a generalization of string trace reconstruction called \emph{tree trace reconstruction}, where
the goal is to learn a node-labelled tree given traces from a deletion channel.
From a technical point of view, tree trace reconstruction may aid in understanding the interplay
of combinatorial and analytic approaches to reconstruction problems and can be a springboard for new ideas. 
From an applications point of view, 
current research on DNA nanotechnology has demonstrated that structures of DNA molecules can be constructed into trees and lattices. 
In fact, recent research has shown how to distinguish different molecular topologies, such as spiders with three arms, from line DNA using nanopores~\cite{KT18}. 
These results may open the door for other tree structures and be useful for applications like DNA data storage.

Let $X$ be a rooted tree with unknown binary labels on its $n$ non-root nodes. 
The goal of {tree trace reconstruction} is to learn the labels of $X$ with high probability, 
using the minimum number of traces, knowing only $q$, the deletion model, and the structure of $X$.

We consider two deletion models. In both models, each non-root node $v$ in $X$ is 
deleted independently with constant probability $q$---the root is never deleted---and deletions are associative. 
The resulting tree is called a trace. 
We assume that $X$ has 
a canonical ordering of its nodes, and 
the children of a node have a left-to-right ordering.
For the Left-Propagation model, we define the {\em left-only} path starting at $v$ as the path that recursively goes from parent to left-most child.

\begin{figure*}[t!]
	\centering
	\begin{subfigure}[b]{0.3\textwidth}
		\centering
		\caption{Original Tree}
		\includegraphics[angle=0,width=\textwidth]{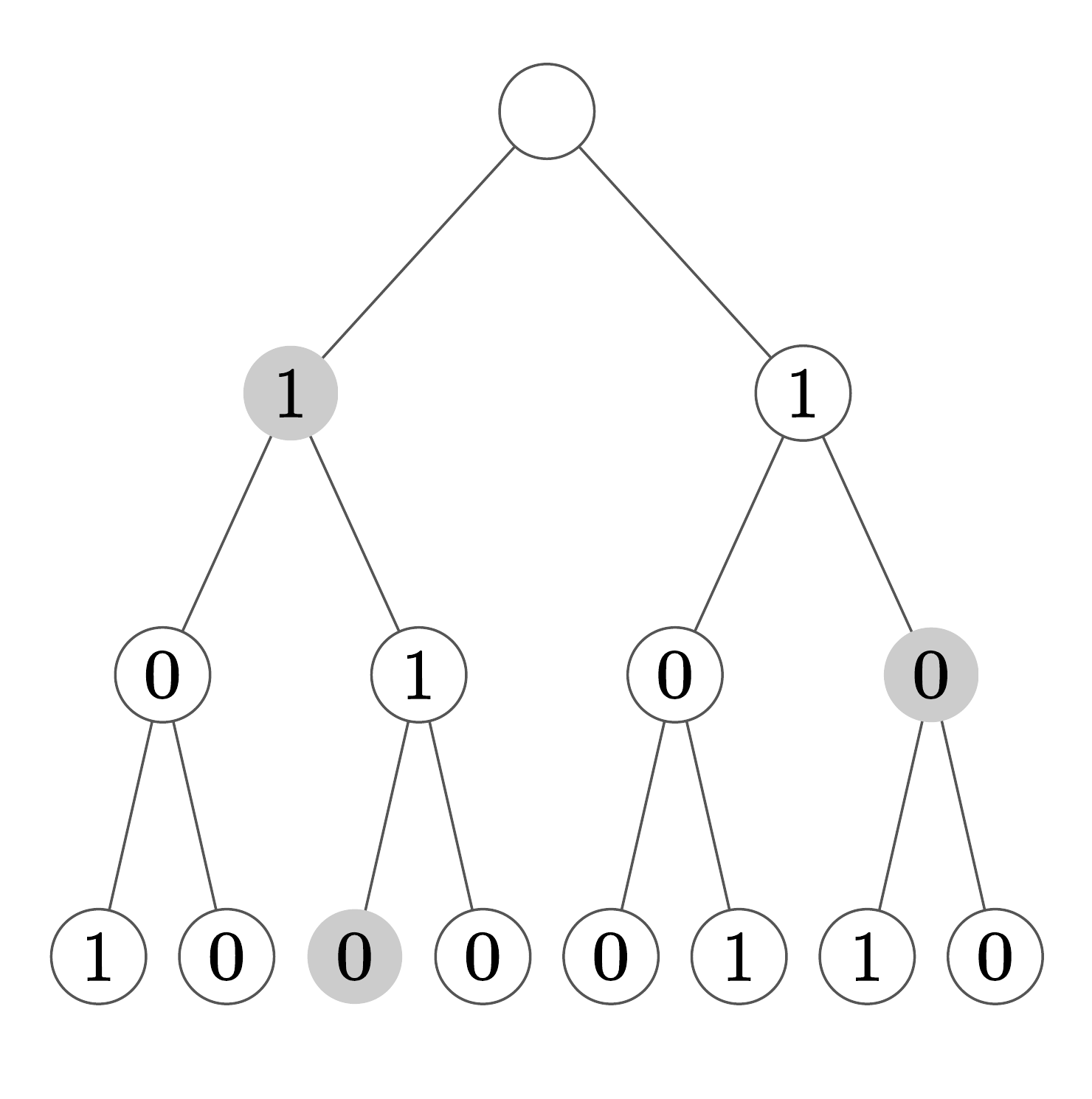}
	\end{subfigure}%
	\begin{subfigure}[b]{0.3\textwidth}
		\centering
		\caption{TED Trace}
		\includegraphics[angle=0,width=\textwidth]{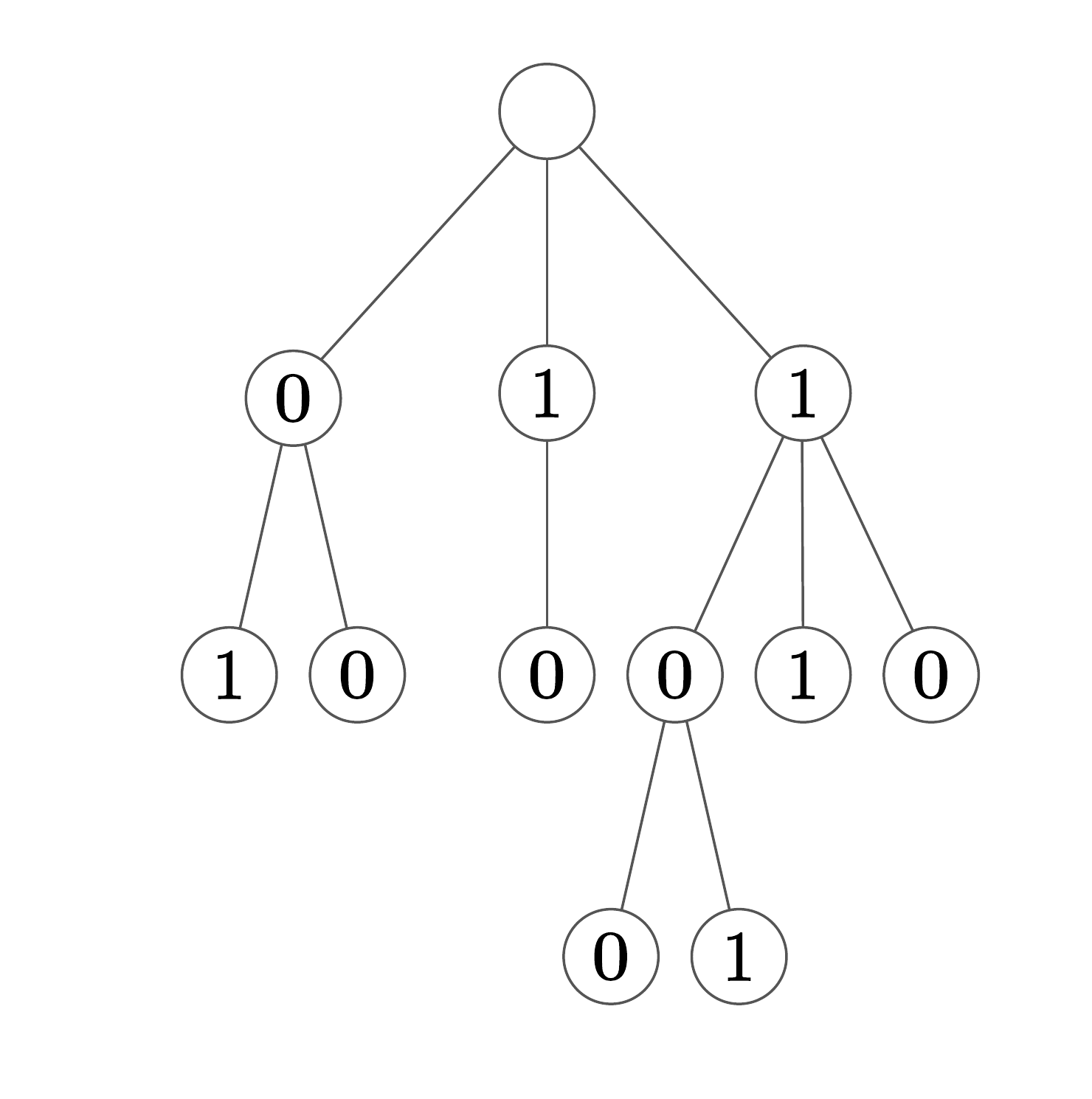}
	\end{subfigure}
	\begin{subfigure}[b]{0.3\textwidth}
		\centering
		\caption{Left-Propagation Trace}
		\includegraphics[angle=0,width=\textwidth]{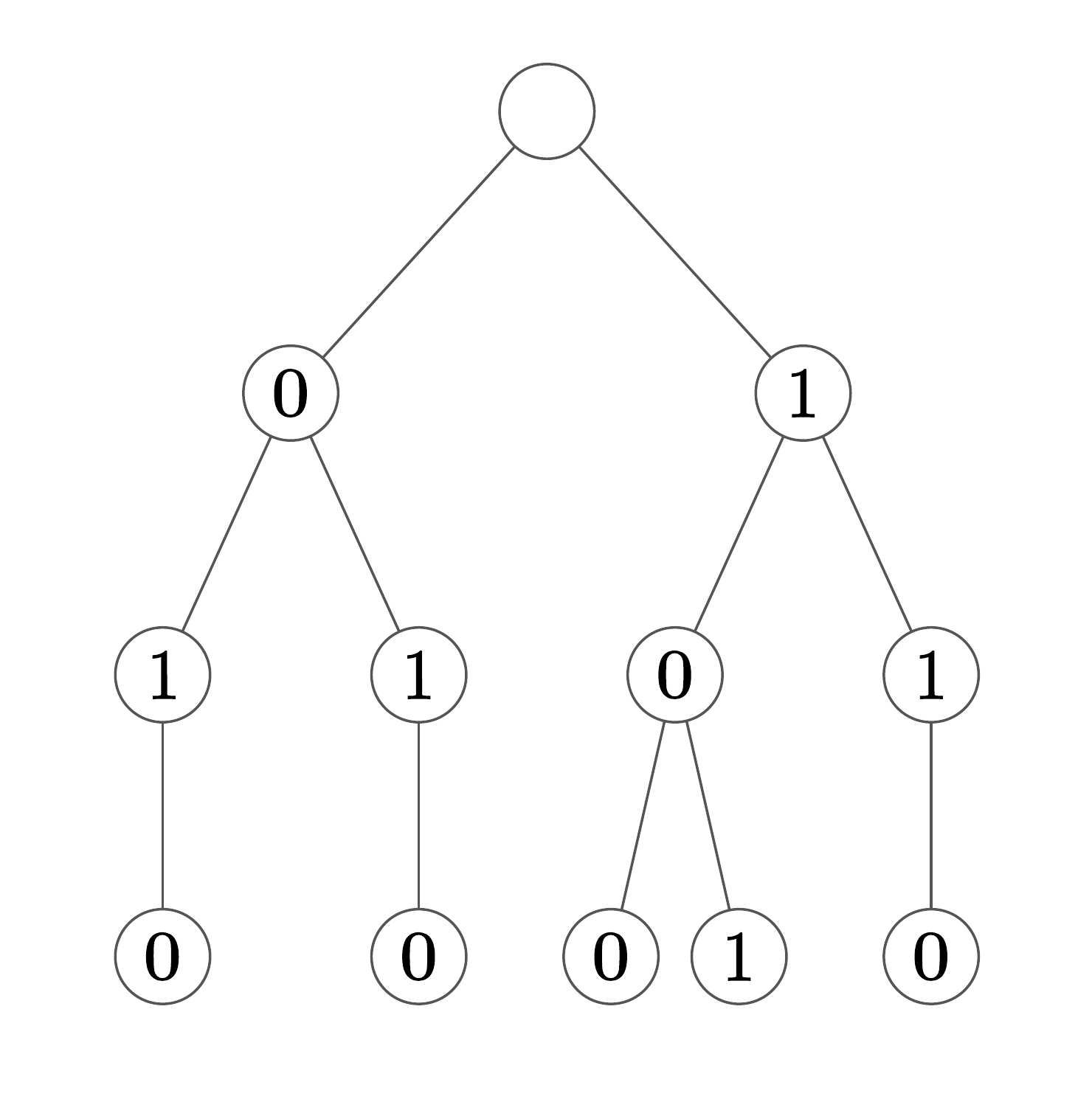}
	\end{subfigure}
	\caption{Deletion models. Gray nodes deleted from original tree (a). Resulting trace in the TED Model (b) and the Left-Propagation Model (c).} 
	\figlab{models}
\end{figure*}

\begin{itemize}
	\item {\bf Tree Edit Distance (TED) model:} When $v$ is deleted, all children of $v$ become children of $v$'s parent. Equivalently, contract the edge between $v$ and its parent, retaining the parent's label. The children of $v$ take $v$'s place as a continuous subsequence in the left-to-right order.
	\item {\bf Left-Propagation model:} When $v$ is deleted, recursively replace every node (together with its label) in the left-only path starting at $v$ with its child in the path. 
	This results in the deletion of the last node of the left-only path,
	with the remaining tree structure unchanged.\footnote{Since the BFS order on $X$ is arbitrary (but fixed), the choice of using the left-only path (as opposed to, say, the right-only one) does not {\em a priori} bias certain nodes.}
\end{itemize}	
\figref{models} depicts traces in both deletion models for a given original tree and set of deleted nodes.   
When $X$ is a path or a star, then both models coincide with the string deletion channel. 
After posting this paper to arXiv, subsequent work has shown that these are the most difficult trees to reconstruct in terms of sample complexity~\cite{maranzatto_thesis}. 
In other words, the sample complexity to reconstruct an arbitrarily labelled tree on $n$ nodes is no more than 
the sample complexity to reconstruct an arbitrarily labelled string on $n$ bits.

A key motivation for the {Tree Edit Distance model}
is that deletions in the TED model correspond exactly to the deletion operation in tree edit distance, which is a well-studied metric for pairs of labeled trees used in applications~\cite{Bille05, ZhangS89}. 
Our main motivation for the {Left-Propagation model} is more theoretical: it preserves different structural properties---for instance, a node's number of children does not increase (see \figref{models})---and poses different challenges than the TED model.

\subsection{Related Work}\seclab{Related}

\subsubsection*{Previous results on string trace reconstruction}\seclab{Previous-string}

Introduced by Batu, Kannan, Khanna, and McGregor~\cite{BatuKannan04-RandomCase}, string trace reconstruction has received a lot of attention, especially recently~\cite{Chase19,DeOdonnellServedio17-WorseCase,de2019optimal, HartungHP18, HL18, HolensteinMPW08,  HoldenPemantlePeres18-RandomCase, tr-revisited, NazarovPeres17-WorseCase, viswanathan}. 
Yet there is still an exponential gap between the known upper and lower bounds for the number of traces needed to reconstruct an arbitrary string with high probability and constant deletion probability: 
it is known that $\exp ( O \left( n^{1/3} \right))$ traces are sufficient~\cite{DeOdonnellServedio17-WorseCase,de2019optimal, NazarovPeres17-WorseCase} and $\widetilde{\Omega}(n^{3/2})$ traces are necessary~\cite{Chase19,HL18}. Determining whether a polynomial number of traces suffice is a challenging open problem in the area. 
A well-studied variant is reconstructing a string with random, average-case labels, instead of arbitrary, worst-case labels~\cite{BatuKannan04-RandomCase, HoldenPemantlePeres18-RandomCase}. This is relevant for applications to DNA data storage~\cite{OAC}. 

In a few of our algorithms, we will reduce 
various subproblems to the string trace reconstruction problem, and hence, we will use existing results as a black box. 
For future reference, we precisely state the previous results now. 
Let $T(n,\delta)$ and $\avgT(n, \delta)$ denote the minimum number of traces needed to reconstruct an~$n$-bit worst-case and average-case string, respectively, with probability at least $1-\delta$, 
where the dependence on the deletion probability $q$ is left implicit.

\begin{theorem}[\cite{DeOdonnellServedio17-WorseCase,de2019optimal,NazarovPeres17-WorseCase}]\thmlab{StringMain}
	The number of traces $T(n,\delta)$ needed to reconstruct a worst-case $n$-bit string with probability $1-\delta$ satisfies
	$T(n, \delta) \leq \ln(\frac{1}{\delta}) \cdot e^{Cn^{1/3}}$, for $C$ depending on $q$.
\end{theorem}

\begin{theorem}[\cite{HoldenPemantlePeres18-RandomCase}]\thmlab{StringRandom}	
	The number of traces $\avgT(n,\delta)$ needed to reconstruct a random $n$-bit string with probability $1-\delta$ satisfies
	$\avgT(n, \delta) \leq \ln(\frac{1}{\delta}) \cdot e^{C\log^{1/3}(n)}$, for $C$ depending on $q$.
\end{theorem}

In terms of lower bounds, it is known that $T(n,\delta) = \wt \Omega(n^{1.5})$ and $\avgT(n, \delta) = \wt \Omega(\log^{5/2}(n))$, 
for any $\delta$ bounded away from one~\cite{Chase19,HL18}. 
The proofs of \thmref{StringMain} rely on a {mean-based algorithm}, one only using the mean of single bits from traces, 
and the bound is optimal for mean-based algorithms~\cite{DeOdonnellServedio17-WorseCase,de2019optimal,NazarovPeres17-WorseCase}. 

\subsubsection*{Other variants of trace reconstruction}
Due to the exponential gap between upper and lower bounds in the string trace reconstruction problem, 
an array of variants have been studied recently. 
Cheraghchi, Gabrys, Milenkovic, and Ribeiro introduce the study of coded trace reconstruction, where the goal 
is to design efficiently encodable codes whose codewords can be efficiently reconstructed with high probability~\cite{CGMR}.
Krishnamurthy, Mazumdar, McGregor, and Pal study trace reconstruction on matrices, 
where rows and columns of a matrix are deleted and a trace is the resulting submatrix. They also 
study string trace reconstruction on sparse strings~\cite{krishnamurthy_et_al:LIPIcs:2019:11189}. 
Ancestral state reconstruction is a generalization of string trace reconstruction, where traces are no longer independent, but instead evolve based on a Markov chain~\cite{andoni2012global}. 

There is also a deterministic version of string trace reconstruction~\cite{Levenshtein:2001}.
Let the {\em $k$-deck} of a string be the multiset of its length $k$ subsequences. The question is to establish how large $k$ must be to uniquely determine an arbitrary string of $n$ bits. Currently, the best known bounds stand at $k = O(\sqrt{n})$ and $k = \exp(\Omega(\sqrt{\log n}))$, due respectively to Krasikov and Roditty~\cite{krasikov1997reconstruction} and Dud\'{i}k and Schulman~\cite{dudik}.
This result has also been used to study population recovery, the problem of learning an unknown distribution of bit strings given noisy samples from the distribution~\cite{ban2019beyond}. 

The term trace complexity has appeared in a network inference context, but the models and definition of a trace are incomparable to ours~\cite{abrahao}.
Other results on deletion channels appear in the survey by Mitzenmacher~\cite{mitzenmacher-survey}.

\subsubsection*{Other graph reconstruction models}
While we are unaware of previous work on reconstructing trees using traces (besides strings), a large variety of other graph-centric reconstruction problems have been considered.

The famous \textit{Reconstruction Conjecture}, due to Kelly~\cite{kelly} and Ulam~\cite{ulam}, posits that every graph $G$ is uniquely determined by its {\em deck}, where the deck of $G$ is the multiset of subgraphs obtained by deleting a single vertex from $G$. Here, the (sub)graphs are unlabeled, and the goal is to determine $G$ up to isomorphism. The Reconstruction Conjecture remains open, although it is known for special cases, such as trees and regular graphs~\cite{kelly, graph-recon-book}. 

Mossel and Ross introduced and studied the shotgun assembly problem on graphs, where they use small vertex-neighborhoods to uniquely identify an unknown graph~\cite{mossel-shotgun}.

\subsection{Our Results}
 
We provide algorithms for two main classes of trees: complete $k$-ary trees and spiders.
In a \emph{complete $k$-ary} tree, every non-leaf node has exactly $k$ children, and all leaves have the same depth.
An {\em $(n,d)$-spider} consists of $n/d$ paths of $d+1$ nodes,
all starting from the same root.
\figref{spider} depicts an example spider,
and it demonstrates that both deletion models lead to the same trace for spiders. 
We focus on these two classes because of their varying amount of structure. 
Spiders behave like a union of disjoint paths, except when some paths have all of their nodes deleted. 
This allows us to extend methods from string trace reconstruction, with a slightly more complicated analysis.
On the other hand, complete $k$-ary trees are so structured that we can
use more combinatorial algorithms, which have proven less successful
for string trace reconstruction so far.
We believe our methods could be used to prove results for larger classes of trees, as well.

In what follows, we use {\em with high probability} to mean with probability at least $1 - O(1/n)$. Also, we let $[t]$ for $t \in \mathbb{N}$ denote the set $\{1,2,\ldots, t\}$.

\subsubsection{TED model for complete \texorpdfstring{$k$}{k}-ary trees}

Let $X$ be a rooted complete $k$-ary tree along with unknown binary labels on its $n$ non-root nodes. Since $k = 1$ and $k = n$ are identical to string trace reconstruction, we focus on $1 < k < n$. We provide two algorithms to reconstruct $X$, depending on whether the degree~$k$ is large or small. 

We state our theorems in terms of $T(k,\delta)$,  
since our reductions use algorithms for string trace reconstruction as a black box 
and the current bounds on $T(k,\delta)$ may improve in the future.

\begin{theorem} \thmlab{ted-large} 
	In the TED model, there exist $c, c'>0$ depending only on $q$ such that if $k \geq c \log^2(n)$, then it is possible to reconstruct a complete $k$-ary tree on $n$ nodes with $\exp(c' \cdot \log_k n) \cdot T(k,1/n^2)$ traces with high probability. 
\end{theorem}
\thmref{StringMain} implies that $T(k,1/n^2) = \exp\left(O\left(k^{1/3}\right)\right)$ if $k \geq c \log^{2}(n)$, so the trace complexity in~\thmref{ted-large} is currently $\exp\left(O\left(\log_k(n) + k^{1/3}\right)\right)$. This is $\mathrm{poly}(n)$ as long as $k = O(\log^3 n)$. 

\begin{theorem} \thmlab{ted-small}
	In the TED model, 
	there exists $C>0$ depending only on $q$ such that	
	$\exp(C k\log_k n)$
	traces suffice to reconstruct a complete $k$-ary tree on $n$ nodes 
	with high probability.
\end{theorem}

In particular, when $k$ is a constant, then the trace complexity of \thmref{ted-small} is $\mathrm{poly}(n)$. \thmref{ted-small} makes no restrictions on $k$, but uses more traces than \thmref{ted-large} for $k \geq c \log^2 n$. It would be desirable to smooth out the dependence on $k$ between our two theorems. In particular, we leave it as an intriguing open question to determine whether $\poly(n)$ traces suffice for all $k \leq \log^3(n)$.

\subsubsection{Left-Propagation model for complete \texorpdfstring{$k$}{}-ary trees}
We provide two reconstruction algorithms for $k$-ary trees in the Left-Propagation model, leading to the following two theorems.

\begin{theorem}\thmlab{left-large}
	In the Left-Propagation model, there exists $c > 0$ depending only on~$q$ such that if $k \geq c \log n$, then $T(d+k, 1/n^2)$ traces suffice to reconstruct a complete $k$-ary tree of depth $d = O(\log_kn)$ with high probability.
\end{theorem}

When $k \geq c \log n$, then $d +k = O(k)$, and we can reconstruct an $n$-node complete $k$-ary tree with $\exp(O(k^{1/3}))$ traces by using \thmref{StringMain}.

We also provide an alternate algorithm that makes no assumptions on $k$. 

\begin{theorem}\thmlab{left-small}
	In the Left-Propagation model, $O(n^{\gamma} \log n)$ traces suffice to reconstruct an $n$-node complete $k$-ary tree with high probability, where $\gamma = \ln\pth{\frac{1}{1-q}}\pth{\frac{c'k}{\ln n} + \frac{1}{\ln k}}$, for a constant $c' > 1$. 
\end{theorem}

\thmref{left-small} implies that $\poly(n)$ traces suffice to reconstruct a $k$-ary tree whenever $k = O(\log n)$ and $q$ is a constant. Moreover, for small enough $q$ and $k$, the algorithm needs only a {sublinear} number of traces (for example, binary trees with $q < 1/2 - \eps$).
From \thmref{StringMain}, the bound in \thmref{left-small} can be more simply thought of as $\exp(C'\cdot(d + k))$; and, in \thmref{left-large} as $\exp(C\cdot(d+k)^{1/3})$.

\subsubsection{Spiders}

Recall that the TED and Left-Propagation deletion models are the same for spiders. 
We provide two reconstruction algorithms, depending on whether the depth $d$ is large or small. 

\begin{theorem}\thmlab{MainSpider} Assume that $d \leq \log_{1/q}n$.
	For $q < 0.7$, there exists $C > 0$ depending only on $q$ such that $\exp (C  \cdot d (n q^d)^{1/3})$ traces suffice to reconstruct an $(n,d)$-spider  with high probability.
\end{theorem} 

To understand the statement of this theorem, consider $d = c \log_{1/q} n$ with $c < 1$. 
A black-box 
reduction to the string case 
results in using $\exp(\widetilde{\Omega}(n^{1-c}))$ traces for reconstruction (see \secref{spiders-from-strings}), 
whereas \thmref{MainSpider} improves this to $\exp( \widetilde{O} ( n^{(1-c)/3} ))$.

\thmref{MainSpider} actually extends to any 
deletion probability $q < 1/\sqrt{2} \approx 0.707$, but this requires taking $d$ to be larger than some constant depending on $q$. 
We discuss further in \remref{ForAllq}
why the regime of $q > 1/\sqrt{2}$ is difficult to handle. Our approach extends previous results based on complex analysis~\cite{DeOdonnellServedio17-WorseCase,de2019optimal,NazarovPeres17-WorseCase}.
As the main technical ingredient, we prove new bounds on certain polynomials whose coefficients are small in modulus.
In particular, we analyze a generating function that might be of independent interest, related to Littlewood polynomials. 

For large depth $d \geq \log_{1/q} n$, full paths of the spider are unlikely to be completely deleted, 
and we derive the following result via a reduction to string trace reconstruction.

\begin{proposition}\proplab{SpiderCor}
For 
$q < 1$ 
and all $n$ large enough, 
an $(n,d)$-spider with $d \geq \log_{1/q} n$ 
can be reconstructed with $2 \cdot T\left(d, \frac{1}{2n^{2}}\right)$ traces with high probability.
\end{proposition}
Using \thmref{StringMain}, the current bound for \propref{SpiderCor} is $2 \cdot T\left(d, \frac{1}{2n^{2}}\right) \leq \exp(O(d^{1/3}))$. 
Comparing \thmref{MainSpider} and \propref{SpiderCor}, 
we see that the bounds in the exponent are 
$d (n q^d)^{1/3}$ 
and $d^{1/3}$, 
for $d \leq \log_{1/q} n$ and $d \geq \log_{1/q} n$, respectively. 
We leave it as an open question to unify these bounds, and in particular, to determine whether the jump is necessary as $d$ crosses $\log_{1/q} n$.

\subsubsection{Average-case labels for trees}
Our results have focused on trees with worst-case, arbitrary labels. Assuming the binary labels are uniformly distributed independent bits leads to significantly improved bounds. For the string case, \thmref{StringRandom} implies that $\avgT(k,1/n^2) = \exp(O(\log^{1/3}k + \log \log n))$ traces suffice to reconstruct a random binary string with high probability. For three of our results, we can use this as a black box and replace the dependence on $T(k,1/n^2)$ with $\avgT(k,1/n^2)$ for average-case labeled trees. 
The average-case trace complexity for $k$-ary trees under the TED model---analogously to \thmref{ted-large}---becomes $\exp(O(\log_{k}(n) + \log^{1/3}k))$ when $k \geq c \log^{2}(n)$. 
For the Left-Propagation model---analogously to \thmref{left-large}---the average-case trace complexity becomes $\exp(O(\log^{1/3} k + \log \log n))$ when $k \geq c \log n$. 
For $(n,d)$-spiders with depth $d \geq \log_{1/q} n$---analogously to \propref{SpiderCor}---the average-case trace complexity becomes $\exp(O(\log^{1/3} d + \log \log n))$. 
Since it is straightforward to use the average-case string result instead of the worst-case result to obtain the results just described,
we restrict our exposition to worst-case labeled $k$-ary trees and spiders.

\subsection{Overview of TED Deletion Algorithms}

Previous work on string trace reconstruction mostly utilizes two classes of algorithms:
mean-based methods, which use single-bit statistics for each position in the trace, 
and alignment-based methods, which attempt to reposition subsequences in the traces to their true positions.

Although mean-based algorithms are currently quantitatively better for string reconstruction, they seem difficult to extend to $k$-ary trees under the TED deletion model. 
Specifically, mean-based methods require a precise understanding of how the bit in position $j'$ of the original tree affects the bit in position $j$ of the trace. 
For strings, there is a global ordering of the nodes which enables this. Unfortunately, for $k$-ary trees with $k \notin \{1,n\}$ under the TED model, nodes may shift to a variety of locations, making it unclear how to characterize bit-wise statistics. 
To circumvent this challenge, we provide two new algorithms, depending on whether or not the degree~$k$ is large ($k \geq c\log^2(n)$). The main idea is to partition the original tree into small subtrees and learn their labels using a number of traces parameterized primarily by $k$ and $\log_k n$, which can be much smaller than $n$.

When $k$ is large enough, we will be able to localize root-to-leaf paths, in the sense that we can identify the location of their non-leaf nodes in the original tree with high probability. By covering the internal nodes of the tree by such paths, we will directly learn the labels for all non-leaf nodes. Then, we observe that the leaves can be naturally partitioned into stars of size $k$, and we can learn their labels by reducing to string trace reconstruction (for strings on $k$ bits). Any improvement to string trace reconstruction will lead to a direct improvement for $k$-ary trees with large degree. 

When $k$ is small, our localization method fails, and we resort to looking at traces which contain even more structure (which requires more traces). We decompose the entire tree into certain subtrees and recover their labels separately.  We define a property which is easily detectable among traces and show that when this property holds, we can extract labels for the subtrees that are correct with probability at least~2/3. Then, we take a majority vote to get the correct labels with high probability.

\subsection{Overview of Left-Propagation Algorithms}

As with the TED model, we combine mean-based and alignment-based strategies, and we provide different algorithms depending on whether the degree is large or small. The two algorithms differ in how they align certain subtrees of traces to positions in the tree. 

When $k$ is large enough ($k \geq c \log n$ for a constant $c > 0$), our first algorithm will use results from string trace reconstruction as a black box. The key idea is that certain subtrees will behave as if they were strings on $O(k)$ bits in the string deletion model. Although this does not happen in all traces, we show that it occurs with high probability. Overall, we partition $X$ into such subtrees, and we reduce to string reconstruction results to recover the labels separately. 

On the other hand, when $k$ is small (such as binary trees with $k=2$), we do not know how to reduce to string reconstruction. Instead, our second algorithm waits until a larger subtree survives in a trace. We show that this makes the alignment essentially trivial, and we can directly recover the labels for certain subtrees.  Quantitatively, the trace complexity of the first algorithm is better, but the reconstruction only succeeds for large enough $k$.

\subsection{Overview of Spider Techniques}

When the paths of a spider are sufficiently long---specifically, if they have depth $d \geq \log_{1/q} n$---then 
with probability close to 1, no path is fully deleted in a given trace. 
This allows us to trivially match paths of the trace spider to paths of the original spider 
and then use string trace reconstruction algorithms on the individual paths, leading to \propref{SpiderCor}. 

When the paths of a spider are shorter ($d < \log_{1/q} n$), many traces have paths fully deleted. 
As illustrated in \figref{spider}, when paths are fully deleted from a spider, 
it is unclear which paths were deleted, 
which forces us to align paths from different traces. 
We bypass direct alignment-based methods and instead use a mean-based algorithm that generalizes the methods introduced in the proof of \thmref{StringMain} by~\cite{DeOdonnellServedio17-WorseCase,de2019optimal,NazarovPeres17-WorseCase}. 
The main difficulty we address is that, in contrast to strings which are one dimensional, spiders are two dimensional: 
one dimension representing which path in the spider a node is in, 
and the other representing where in a path a node~is.

\subsection{Outline} 

The rest of the paper is organized as follows.
Preliminaries are in
\secref{Preliminaries}. 
The proofs of \thmref{ted-large} and \thmref{ted-small} for $k$-ary trees under the TED model appear in \secref{ted}.
The proofs of \thmref{left-large} and \thmref{left-small} 
for the Left-Propagation model appear
in \secref{left}.
The spider reconstruction preliminaries and algorithms for \thmref{MainSpider} and \propref{SpiderCor} are in 
\secref{spiders}.
The three main sections can be read independently, after their preliminaries. 
We conclude in \secref{conclusion}.

\section{Preliminaries}\seclab{Preliminaries}

In what follows, $X$ denotes the (known) underlying tree, along with the (unknown) binary labels on its $n$ non-root nodes.

\paragraph{Standard tree definitions.}
We say that $X$ is {\em rooted} if it has a fixed root node. We assume the root is never deleted (for further explanation see \remref{root}). 
An {\em ancestor} (resp. {\em descendant}) of a node $v$ is a node reachable from $v$ by proceeding repeatedly from child to parent (resp. parent to child). We say $v$ is a {\em leaf} if it has no children, and otherwise $v$ is an {\em internal} node. The {\em length} of a path equals the number of nodes in it.  The {\em depth} of $v$ is the number of edges in the path from the root to~$v$. The {\em height} of $v$ is the number of edges in the longest path between~$v$ and a leaf. The depth of a rooted tree is the height of the root. We say that $X$ is a complete $k$-ary tree of depth $d$ if every internal node has $k$ children and all leaves have depth $d$.

\subsection{\texorpdfstring{$k$}{k}-ary Tree Algorithm Preliminaries}\seclab{tree-prelims}

Let $X$ be a rooted complete $k$-ary tree with depth $d$.
We index the non-root nodes according to the BFS order on $X$ (the root is not indexed; the children of the root are $\{0,1,\ldots, k-1\}$, etc.). We identify nodes of $X$ with their index. 
For $t \in [d]$, let $\calJ_t$ be the nodes at depth~$t$. Define $\calI_1 := \calJ_1 = \{0,1,\ldots, k-1\}$, and for $t \geq 2$, 
\[
\calI_t := \{i \in \calJ_t \mid i\ \mathrm{mod}\ k \neq 0\}.
\]
In words, for $t \geq 2$, $\calI_t$ is the set of nodes at depth $t$ which are not left-most among their siblings. 
Define also $\calI := \bigcup_{t = 1}^{d-1} \calI_{t}$.

We define three unlabeled subtrees of $X$. Let $P_X(i)$ be the path from the root to $i$ in~$X$. Define~$H_X(i)$ as the union of the left-only path starting at $i$, descending to a leaf~$\ell$, and the $k-1$ siblings of $\ell$. Finally, define $G_X(i) := P_X(i) \cup H_X(i)$. See \figref{subtrees} for an example of these subtrees. For clarity, we note that if~$i$ has depth $t$ in $X$ (i.e., $i \in \calJ_t$), then $|P_X(i)| = t+1$ and $|H_X(i)| = d-t+k$ and $|G_X(i)| = d+k$.

\subsubsection*{Canonical subtrees of traces}

We also consider certain subtrees of a trace $Y$. They will be analogous to $P_X(i),\ H_X(i),$ and $G_X(i)$, and they only depend on the position of $i$ in $X$. We will denote them as $P_Y(i),\ H_Y(i),$ and $G_Y(i)$. Intuitively, they are subtrees in $Y$ obtained by looking at nodes that should be in the same position as  the corresponding ones in $X$. However,  the node $i$ does not necessarily belong to these subtrees (e.g., it may have been deleted in $Y$, or another node may be in its place). In what follows, we refer to subtrees as sequences of nodes in the BFS order, since the edge structure will be clear from context (i.e., the subtree is the induced subgraph on the relevant nodes).

\begin{figure*}[t!]
	\centering
	\includegraphics[angle=0,width=0.95\textwidth]{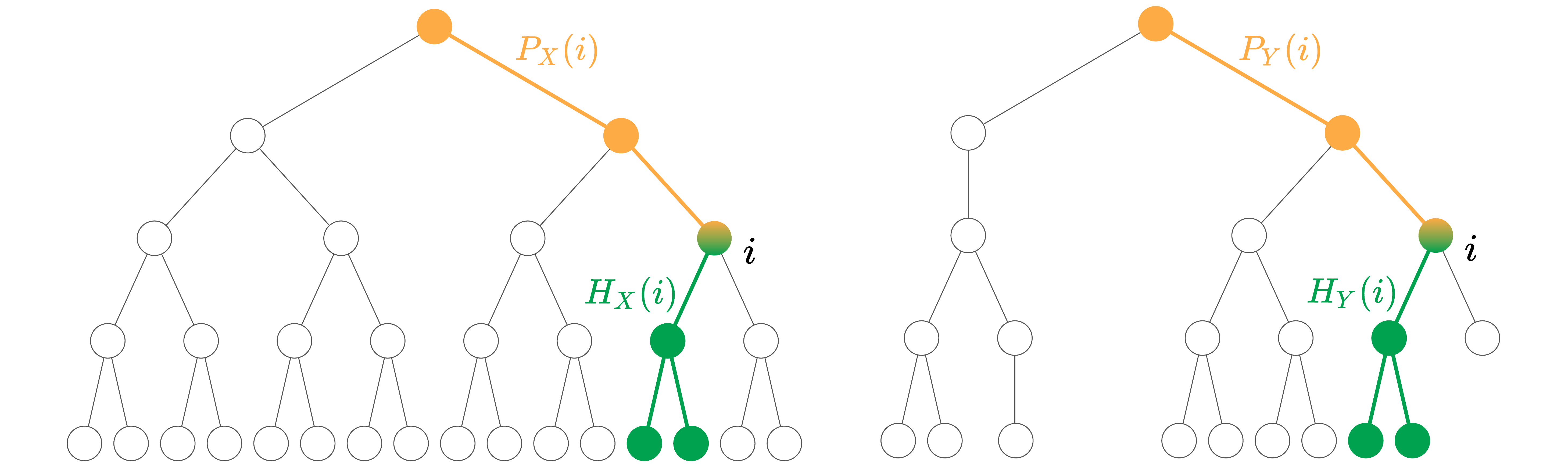}
	\caption{Canonical subtrees for $k$-ary trees, in the original tree (left) and trace (right).} 
	\figlab{subtrees}
\end{figure*}

We now formally define $P_Y(i),\ H_Y(i),$ and $G_Y(i)$, which are also depicted in \figref{subtrees}. Fix $i$, and let $u_0, u_1,\ldots, u_{d-1}$ be the internal nodes in $G_X(i)$, where $u_t$ has depth~$t$, and let $u_d, \ldots, u_{d+k-1}$ be the leaf nodes, ordered left-to-right in the BFS order. Define $\pi_i:\{0,1,\ldots,d-1\} \to \{0,1,\ldots,k-1\}$ so that $\pi_i(t)$ is the position of $u_{t+1}$ in $X$ among its siblings (the children of its parent $u_{t}$). Note that~$\pi_i$ is independent of the labels of $X$.
Let $t_i$ be the depth of $i$ in $X$. We define $P_Y(i)$ as the path  $v_0,v_1,\ldots, v_{t_i}$ in $Y$ obtained from the following process. Set $v_0$ to be the root. Then, for $t \in [t_i]$, let~$v_{t}$ be the node at depth $t$ in $Y$ that is in position $\pi_i(t-1)$ among the $k$ children of $v_{t-1}$, where we abort and set $P_Y(i) = \perp$ if $v_{t-1}$ does not have exactly $k$ children. 
Similarly, let $G_Y(i)$ be the subtree $v_0, v_1,\ldots, v_{d+k-1}$, where $v_t$ is defined as follows. Set $v_0$ to be the root in $Y$. Then, for $t \in [d-1]$, let $v_{t}$ be the node at depth $t$ in $Y$ that is in position $\pi_i(t-1)$ among the $k$ children of $v_{t-1}$, where we abort and set $G_Y(i) = \perp$ if $v_{t-1}$ does not have exactly $k$ children. Finally, set $v_d, \ldots, v_{d+k-1}$ to be the $k$ children of $v_{d-1}$, and again we set $G_Y(i) = \perp$ if $v_{d-1}$ does not have precisely $k$ children. If $G_Y(i) \neq \perp$, then set $H_Y(i) = v_{t_i},\ldots, v_{d+k-1}$, and otherwise, set $H_Y(i) = \perp$. Observe that if  $G_Y(i) \neq \perp$, then we have $G_Y(i) = P_Y(i) \cup H_Y(i)$.

We remark that $G_Y(i), H_Y(i),$ and $P_Y(i)$ depend only on $\pi_i$ and the tree structure of $Y$, and therefore 
they do not use any label information from $X$. 
We also note that whether these subtrees are set to $\perp$ will be significant, 
since this implies certain structural properties of traces.  
If all nodes in $G_X(i)$ survive in a trace $Y$, then we say that $Y$ {\em contains} $G_X(i)$. 
We write $G_Y(i) = G_X(i)$ if the nodes in these subtrees are exactly the same 
(by construction, the edges will also be the same). 
We conclude this section with two remarks that are useful for reconstruction of $X$. 
\begin{remark}\remlab{BitByBit}
If $G_Y(i) = G_X(i)$, then we can reconstruct the labels of $G_X(i)$ 
bit by bit by copying the label to be that from the corresponding bit in $G_Y(i)$. 
The same applies for $H_X(i)$ and $P_X(i)$. 
\end{remark}

\begin{remark}\remlab{SeparateReconstruction}
To reconstruct labels of $X$, one can reconstruct labels of subtrees of $X$, where the subtrees cover all nodes of $X$.
\end{remark}


\section{Reconstructing Trees, TED Deletion Model}\seclab{ted}

In this section we prove our two results for $k$-ary trees in the TED model.

\subsection{Proof of \thmref{ted-large} Concerning Large Degree Trees}
\seclab{overview-ted-large}
Our algorithm utilizes structure that occurs when $k \geq c\log^2(n)$. 
Recall that for a node $i$ in $X$, we think of $i$'s children as being ordered consecutively, left-to-right, 
based on the BFS ordering of $X$.

\begin{definition}
	Let $Y$ be a trace of a tree $X$. We say that~$Y$ is $b$-{\em balanced} if, for every internal node~$i$ in~$X$, at most $b$ consecutive  children of $i$ have been deleted in~$Y$.
\end{definition}

\begin{claim} \clmlab{balanced} If $X$ has $n$ nodes, then a trace $Y$ is $b$-balanced with probability at least $1 - nq^b$. 
\end{claim}
\begin{proof} Any set of $b$ consecutive nodes is deleted with probability $q^b$. Since there are at most $n$ starting nodes for a run of $b$ nodes, a union bound proves the claim.
\end{proof}
Since $T(k,1/n^2) = \exp(O(k^{1/3}))$ by \thmref{StringMain}, the number of traces used in \thmref{ted-large} is $\exp \left( O( \log_{k}(n) + k^{1/3} ) \right)$.
Therefore, setting 
$b := 10 \sqrt{k}/\log(1/q)$ = $\Omega(\log n)$, 
\clmref{balanced} and a union bound show that with high probability \emph{all} traces will be $b$-balanced. 
As we shall see, 
the benefit of this balanced structure manifests itself in the proof of correctness of the reconstruction algorithm used for \thmref{ted-large}. 

Our reconstruction algorithm that proves \thmref{ted-large} consists of two main steps: 
\begin{enumerate}
\item \textbf{(Finding Paths and Grouping Traces)} First, we process all the traces and group them into different sets (which may overlap). This grouping of traces is based on finding root-to-leaf paths that are preserved in a trace and estimating where these paths came from in~$X$. We term this latter algorithm the \textsc{FindPaths} algorithm; see \algref{findpaths_kary_large_TED} below for its pseudocode. 

\item \textbf{(Reconstruction)} We then analyze each subset of traces. Each 
of these 
leads to reconstructing the labels for a particular subset of $X$, consisting of a path from the root to a node at depth $d-1$, together with the $k$ children of this node. Finally, we output the union of all such labels as the estimated labels of $X$. 
In other words, we cover $X$ with the collection of subtrees 
$\left\{ G_{X}(j) : j \in \mathcal{J}_{d-1} \right\}$ 
and estimate the labels of each $G_{X}(j)$ separately. 
\algref{alg:TED_k_large} below states, in pseudocode, the full reconstruction algorithm, which calls \algref{findpaths_kary_large_TED} as a subroutine.
\end{enumerate} 

\subsubsection*{The \textsc{FindPaths} algorithm}

We start by describing the \textsc{FindPaths} algorithm (see \algref{findpaths_kary_large_TED}). 
The input to this algorithm is a trace $Y$, 
while the output of the algorithm will be a subset of $\mathcal{J}_{d-1}$, the nodes of $X$ at depth $d-1$; the conceptual meaning of this subset will be clear once the algorithm is described. 
We recall that we index the non-root nodes according to the BFS order on $X$, 
and we will interchangeably refer to nodes and their BFS index. 

\begin{algorithm}[t!]
	\caption{\textsc{FindPaths} in $k$-ary trees, TED deletion model}
	\alglab{findpaths_kary_large_TED}
	\hspace*{\algorithmicindent} \textbf{Input:} a trace $Y$ sampled from the TED deletion channel.
	\begin{algorithmic}[1] 
		\State Initialize  $\mathcal{S} = \emptyset$.
		\For{$v$ a leaf in $Y$}
			\If{$v$ has depth $d$ in $Y$}
				\State Add the parent of $v$ to $\mathcal{S}$.
			\EndIf
		\EndFor
		\State Initialize  $\widehat{\mathcal{S}} = \emptyset$.
		\For{$v \in \mathcal{S}$}
			\For{$\ell=0$ to $d-2$}
				\State Compute $\widehat{a}_{\ell}$ based on node-to-leaf anchor paths and combining plug-in estimators (see \Eqref{ahatell} for the final formula, the text for further details, and \figref{balanced} for an illustration).
			\EndFor 
			\State Set $\widehat{w} := \widehat{a}_{d-2} \widehat{a}_{d-3} \cdots \widehat{a}_{0}$ (written in base-$k$ expansion).
			\State Add $\widehat{w}$ to $\widehat{\mathcal{S}}$.
		\EndFor
		\State \textbf{Output:} $\widehat{\mathcal{S}}$.
	\end{algorithmic}
\end{algorithm}

The first part of the \textsc{FindPaths} algorithm is to identify root-to-leaf paths that have been preserved (i.e., no vertex in the path has been deleted) in the trace $Y$. This is straightforward, 
since if a root-to-leaf path in~$X$ is preserved, then the corresponding leaf has depth $d$ in $Y$; 
and vice versa, 
every leaf in $Y$ that has depth $d$ corresponds to a root-to-leaf path in $X$ that was preserved. Once all surviving root-to-leaf paths have been identified, we collect in the set $\mathcal{S}$ all the nodes of~$Y$ that are on a surviving root-to-leaf path and have depth $d-1$ (see lines 1--6 of \algref{findpaths_kary_large_TED}).

We know that each node $v \in \mathcal{S}$ must have come from a node $w \in \mathcal{J}_{d-1}$ (i.e., a node in $X$ of depth $d-1$). 
The second and final part of the \textsc{FindPaths} algorithm consists of estimating, for each node $v \in \mathcal{S}$, which original node $w \in \mathcal{J}_{d-1}$ it came from; this estimate is denoted by $\widehat{w}$.  
Note that the left-to-right ordering of the nodes in $\mathcal{S}$ and the original nodes in $\mathcal{J}_{d-1}$ which they come from are the same, so the algorithm needs only to output the set 
$\widehat{\mathcal{S}} := \left\{ \widehat{w} : v \in \mathcal{S} \right\}$ 
(since the mapping between $\mathcal{S}$ and $\widehat{\mathcal{S}}$ follows the left-to-right ordering).\footnote{Regarding notation: note that $\widehat{\mathcal{S}}$ is \emph{not} an estimate of $\mathcal{S}$, but rather an estimate of the pre-image of $\mathcal{S}$ before~$X$ is passed through the deletion channel to obtain $Y$. We hope that the reader accepts this abuse of notational convention.}

Given $v \in \mathcal{S}$, 
to compute the estimate $\widehat{w}$, 
we first observe that any node $w \in \mathcal{J}_{d-1}$ can be written in its base-$k$ expansion, 
\[
w = a_{d-2} a_{d-3} \cdots a_{0},
\] 
where $a_{\ell} \in \left\{ 0, 1, \ldots, k - 1 \right\}$ for $\ell \in \left\{ 0, 1, \ldots, d - 2 \right\}$. 
Thus in order to compute an estimate $\widehat{w}$, 
it suffices to compute an estimate $\widehat{a}_{\ell}$ of $a_{\ell}$ for every $\ell \in \left\{ 0, 1, \ldots, d - 2 \right\}$ 
and then set 
\[
\widehat{w} := \widehat{a}_{d-2} \widehat{a}_{d-3} \cdots \widehat{a}_{0}.
\] 
The following is an equivalent and more pictorial way of thinking about this. 
Let $u_{0}, u_{1}, \ldots, u_{d-1}$ denote the nodes in $X$ on the path from the root to $w \in \mathcal{J}_{d-1}$, with $u_{t}$ having depth $t$ 
(in particular, $u_{0}$ is the root and $u_{d-1} = w$). 
Then, for $\ell \in \left\{ 0, 1, \ldots, d-2 \right\}$, 
the quantity $a_{\ell}$ 
is the position (from the left, with indexing starting at $0$) of $u_{d- 1 - \ell}$ among the $k$ children of $u_{d-2-\ell}$. 
Now suppose that $v \in \mathcal{S}$ came from node $w \in \mathcal{J}_{d-1}$, 
and let $u'_{0}, u'_{1}, \ldots, u'_{d-1}$ denote the nodes in $Y$ on the path from the root to $v$, with $u'_{t}$ having depth $t$ (in particular, $u'_{0}$ is the root and $u'_{d-1} = v$). 
Thus estimating $a_{0}, a_{1}, \ldots, a_{d-2}$ 
corresponds to estimating, for each node $u'_{t}$, 
where its pre-image in $X$ ranks in the left-to-right ordering of itself and its $k-1$ siblings. 

We now explain how to compute the estimate $\widehat{a}_{0}$ given $v \in \mathcal{S}$; computing $\widehat{a}_{\ell}$ for general $\ell$ is similar but involves slightly more notation, so we defer this for now. 
Let $z_{0}, z_{1}, \ldots, z_{m}$ denote $v$ and its siblings in $Y$, ordered from left to right, and let $\mathcal{Z} := \left\{ z_{0}, z_{1}, \ldots, z_{m} \right\}$. 
Let $z^{*}_{0}, z^{*}_{1}, \ldots, z^{*}_{k'}$ denote the nodes among $\mathcal{Z}$ that have a child in $Y$, ordered from left to right, and let $\mathcal{Z}^{*} := \left\{ z^{*}_{0}, z^{*}_{1}, \ldots, z^{*}_{k'} \right\}$. Note that $v \in \mathcal{Z}^{*}$ by definition; define $k^{*}$ to be the index such that $v = z^{*}_{k^{*}}$. 
Also, by construction, the pre-images of all nodes in $\mathcal{Z}^{*}$ were siblings in $X$, so we must have that $k' \leq k-1$. 
Note that there are two ways that a node can be in $\mathcal{Z} \setminus \mathcal{Z}^{*}$: 
\begin{itemize}
\item A sibling $w'$ of $w$ in $X$ is \emph{not} deleted in $Y$, but \emph{all} of the children of $w'$ are deleted in $Y$. Then the image of $w'$ in $Y$ is in $\mathcal{Z} \setminus \mathcal{Z}^{*}$. Note that this is a highly unlikely event, since $k$ is large. 

\item A sibling $w'$ of $w$ in $X$ is deleted in $Y$, but not all of the children of $w'$ are deleted in $Y$. Then the images of the non-deleted children of $w'$ in $Y$ are in $\mathcal{Z} \setminus \mathcal{Z}^{*}$. 
Note that if such a vertex $w'$ is deleted, then in expectation there will be $(1-q)k$ non-deleted children. 
\end{itemize}
Since the first bullet point above is highly unlikely and the second bullet point describes the typical behavior of a trace, 
this motivates the following estimation procedure. 
For $i \in \left[ k' \right]$, let 
$\alpha_{i}$ denote the number of nodes in $\mathcal{Z}$ that are between 
$z^{*}_{i-1}$ and $z^{*}_{i}$; furthermore, let $\alpha_{0}$ denote the number of nodes in $\mathcal{Z}$ that are before $z^{*}_{0}$. 
Now for every $i \in \left\{ 0, 1, \ldots, k' \right\}$ let 
$\widehat{\alpha}_{i}$ denote the unique integer satisfying 
\[
\widehat{\alpha}_{i} - 1/2 
\leq 
\frac{\alpha_{i}}{(1-q)k} 
<
\widehat{\alpha}_{i} + 1/2. 
\]
Finally, we set 
\begin{equation}\eqlab{ahat0}
\widehat{a}_{0} := k^{*} + \sum_{i=0}^{k^{*}} \widehat{\alpha}_{i}.
\end{equation}
(If this results in an estimate that is greater than $k-1$, then instead set $\widehat{a}_{0} := k-1$.)

Now we turn to estimating $\widehat{a}_{\ell}$ for general $\ell \in \left\{ 0, 1, \ldots, d - 2 \right\}$, given $v \in S$. 
Recall that $u'_{0}, u'_{1}, \ldots, u'_{d-1}$ denote the nodes in $Y$ on the path from the root to $v$, with $u'_{t}$ having depth $t$. 
Let $z_{0}, z_{1}, \ldots, z_{m}$ denote $u'_{d-1-\ell}$ and its siblings in $Y$, ordered from left to right, and let 
$\mathcal{Z} := \left\{ z_{0}, z_{1}, \ldots, z_{m} \right\}$ 
(we reuse notation from above). 
Let $z^{*}_{0}, z^{*}_{1}, \ldots, z^{*}_{k'}$ denote the nodes among $\mathcal{Z}$ that have height $\ell + 1$ in $Y$, ordered from left to right, and let $\mathcal{Z}^{*} := \left\{ z^{*}_{0}, z^{*}_{1}, \ldots, z^{*}_{k'} \right\}$. Note that $u'_{d-1-\ell} \in \mathcal{Z}^{*}$ by definition; define $k^{*}$ to be the index such that $u'_{d-1-\ell} = z^{*}_{k^{*}}$. 
Also, by construction, the pre-images of all nodes in $\mathcal{Z}^{*}$ were siblings in $X$, so we must have that $k' \leq k-1$. 
For $i \in \left[ k' \right]$, let 
$\alpha_{i}$ denote the number of nodes in $Y$ that are either 
(a) in $\mathcal{Z}$ between 
$z^{*}_{i-1}$ and $z^{*}_{i}$, 
or (b) are descendants in $Y$ of such a node; see \figref{balanced} for an illustration. 
Furthermore, let $\alpha_{0}$ denote the number of nodes in $Y$ that are either 
(a) in $\mathcal{Z}$ before $z^{*}_{0}$, 
or (b) are descendants in $Y$ of such a node. 
Now for every $i \in \left\{ 0, 1, \ldots, k' \right\}$ let 
$\widehat{\alpha}_{i}$ denote the unique integer satisfying 
\begin{equation}\eqlab{alpha_i_est}
\widehat{\alpha}_{i} - 1/2 
\leq 
\frac{\alpha_{i}}{(1-q) \sum_{h = 1}^{\ell+1} k^{h}} 
<
\widehat{\alpha}_{i} + 1/2. 
\end{equation}
Finally, we again set 
\begin{equation}\eqlab{ahatell}
\widehat{a}_{\ell} := k^{*} + \sum_{i=0}^{k^{*}} \widehat{\alpha}_{i}.
\end{equation}
The estimate in \Eqref{ahatell} is thus a generalization of the special case of $\widehat{a}_{0}$ in \Eqref{ahat0}. 
(If this results in an estimate that is greater than $k-1$, then instead set $\widehat{a}_{\ell} := k-1$.)

\begin{figure*}[t!]
	\centering
	\includegraphics[angle=0,width=0.55\textwidth]{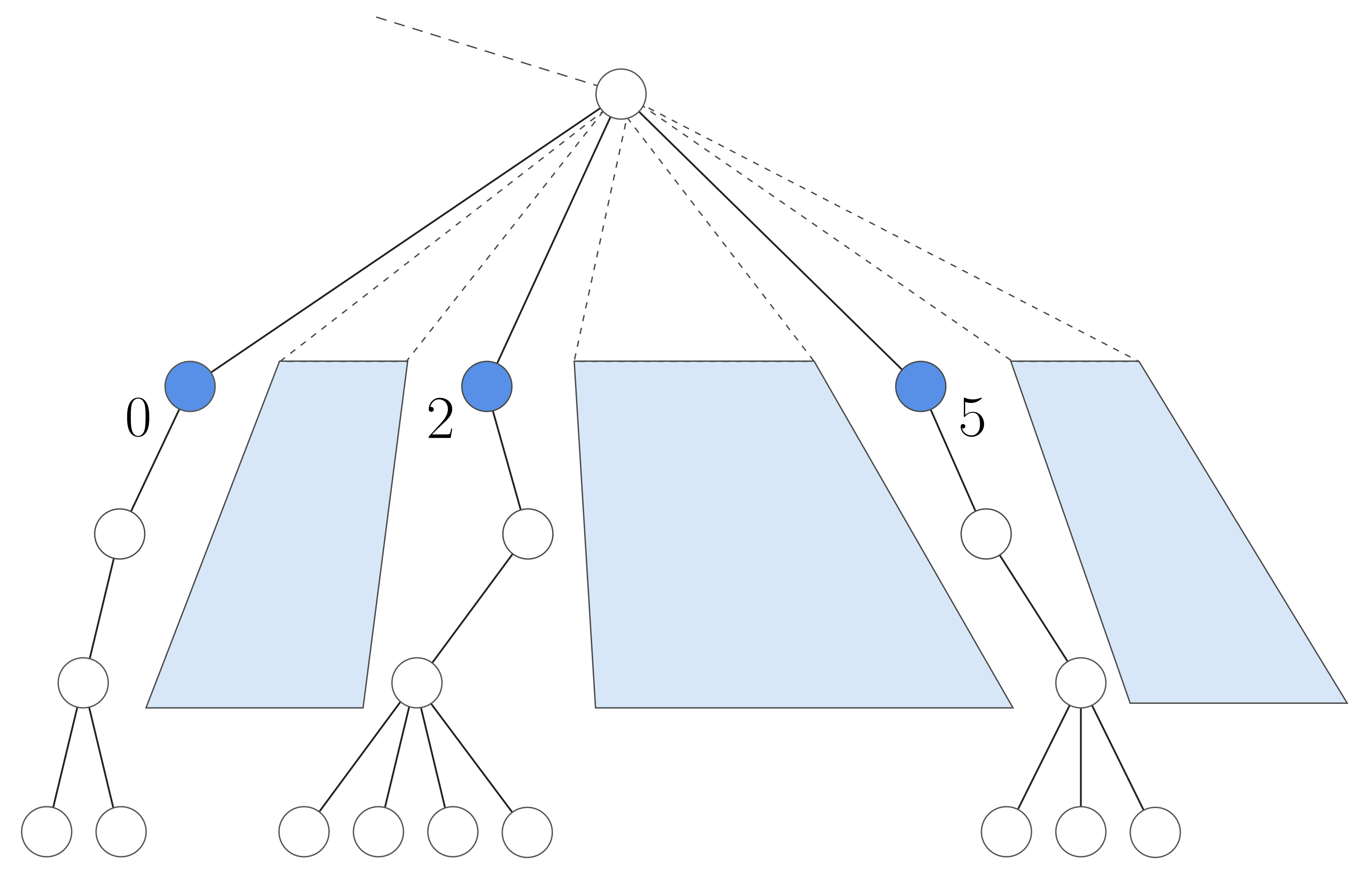}
	\caption{\textbf{Path estimation in $k$-ary trees.} We estimate, level by level, the pre-image of each root-to-leaf path in a trace. At each level, we estimate the number of nodes deleted at that level using the number of nodes in the (light blue) trapezoids in the figure (this uses that $k$ is large enough to apply concentration bounds). This then allows us to determine the positions of the surviving (dark blue) nodes, which have paths to a leaf (e.g., positions $0\ \textunderscore\ 2\  \textunderscore \ \textunderscore\  5\  \textunderscore$ above).} 
	\figlab{balanced}
\end{figure*}

This fully completes the description of the \textsc{FindPaths} algorithm. In the following lemma we analyze the performance of the \textsc{FindPaths} algorithm 
and show that its output is correct with high probability.

\begin{lemma}\lemlab{findpaths}
There exist constants $c$ and $c'$, that depend only on $q$, such that the following holds. 
Let $k \geq c \log^{2} (n)$, let $X$ be a $k$-ary tree with arbitrary binary labels, and let $Y$ be a trace sampled from the TED deletion channel. 
The \textsc{FindPaths} algorithm is fully successful---that is, for \emph{all} nodes $v \in \mathcal{S}$, the estimate $\widehat{w}$ is correct---with probability at least $1 - \exp ( - c' \sqrt{k} )$. 
\end{lemma}
\begin{proof}
Throughout this proof we denote the complement of an event $\mathcal{E}$ by $\mathcal{E}^{c}$. 
Set 
$b := 10 \sqrt{k}/\log(1/q)$ = $\Omega(\log n)$ 
and let $\mathcal{B}_{b}$ denote the event that $Y$ is $b$-balanced. 
By \clmref{balanced} we have that 
\begin{equation}\eqlab{balanced_bound}
\mathbb{P} \left( \mathcal{B}_{b}^{c} \right) 
\leq \exp \left( - C' \sqrt{k} \right)
\end{equation}
for some constant $C'$.

Let $u$ be an internal node of $X$ and let $h_{u}$ be the height of $u$.  
Since $u$ is an internal node, we have that $1 \leq h_{u} \leq d$. 
The number of descendants of $u$ in $X$ is 
$\sum_{\ell = 1}^{h_{u}} k^{\ell}$. 
Let $R_{u}$ denote the number of descendants of $u$ in $X$ that survive in $Y$. 

Now fix $m \leq b$ and let $u_{1}, \ldots, u_{m}$ denote $m$ consecutive siblings in $X$ with height $h \in [d]$. Let $\mathcal{E}_{u_{1}, \ldots, u_{m}}$ denote the event that 
\[
\left| \sum_{i=1}^{m} R_{u_{i}} - m \left( 1 - q \right) \sum_{\ell = 1}^{h} k^{\ell} \right| 
\leq \frac{1}{3} \left( 1 - q \right) \sum_{\ell = 1}^{h} k^{\ell}.
\] 
By a standard Chernoff bound  
we have that there exist constants $c'', c''' > 0$ such that 
\begin{equation}\eqlab{descendants_Chernoff}
\mathbb{P} \left( \mathcal{E}_{u_{1}, \ldots, u_{m}}^{c} \right) 
\leq \exp \left( - c'' (1-q) \sum_{\ell = 1}^{h} k^{\ell} / m \right) 
\leq \exp \left( - c'' (1-q) k / m \right)
\leq \exp \left( - c''' \sqrt{k} \right),
\end{equation}
where in the last inequality we used that 
$m \leq b = 10 \sqrt{k} / \log(1/q)$. 

Finally, define the event 
\[
\mathcal{E} := \mathcal{B}_{b} \cap \bigcap_{m=1}^{b} \bigcap_{u_{1}, \ldots, u_{m}} \mathcal{E}_{u_{1}, \ldots, u_{m}},
\]
where the intersection is over all possible $m$ consecutive siblings $u_{1}, \ldots, u_{m}$ in $X$. 
Putting together \Eqref{balanced_bound}, \Eqref{descendants_Chernoff}, and a union bound, we have that 
\[
\mathbb{P} \left( \mathcal{E}^{c} \right) \leq \exp \left( - c' \sqrt{k} \right)
\]
for some constant $c'$. 
On the other hand, on the event $\mathcal{E}$, 
the estimates $\widehat{\alpha}_{i}$ in \Eqref{alpha_i_est} 
are correct for all $v \in \mathcal{S}$, all $\ell \in \left\{ 0, 1, \ldots, d - 2 \right\}$, and all $i \in \left\{ 0, 1, \ldots, k' \right\}$. 
This implies that for every $v \in \mathcal{S}$ and every $\ell \in \left\{ 0, 1, \ldots, d - 2 \right\}$, the estimate $\widehat{a}_{\ell}$ in \Eqref{ahatell} is correct. 
Therefore for every $v \in \mathcal{S}$ the estimate $\widehat{w}$ is also correct. 
\end{proof}

\subsubsection*{The reconstruction algorithm: estimating the labels of $G_{X}(j)$ for each $j \in \mathcal{J}_{d-1}$}

Now that we have described and analyzed the \textsc{FindPaths} algorithm (\algref{findpaths_kary_large_TED}), we turn our attention to the full reconstruction algorithm (see \algref{alg:TED_k_large}).

\begin{algorithm}[t!]
	\caption{Reconstructing $k$-ary trees, $k \geq c\log^2(n)$, TED deletion model}
	\alglab{alg:TED_k_large}
	\hspace*{\algorithmicindent} Set $T = \exp( c' \log_{k} (n)) \cdot T(k,1/n^{2})$ (for a large enough constant $c'$). \\
	\hspace*{\algorithmicindent} \textbf{Input:} traces $Y_{1}, \ldots, Y_{T}$ sampled independently from the TED deletion channel.
	\begin{algorithmic}[1] 
		\For{$j \in \mathcal{J}_{d-1}$} 
			\State Initialize $\mathcal{A}_{j} = \emptyset$. 
		\EndFor
		\For{$t = 1$ to $T$}
			\State Run \algref{findpaths_kary_large_TED} with input $Y_{t}$; let $\widehat{\mathcal{S}}_{t}$ denote the output. 
			\For{$j \in \widehat{\mathcal{S}}_{t}$} 
				\State Add $Y_{t}$ to $\mathcal{A}_{j}$.
			\EndFor
		\EndFor
		\For{$j \in \mathcal{J}_{d-1}$} 
		estimate the labels of $G_{X}(j)$ as follows:
		\If{$\mathcal{A}_{j} = \emptyset$} \Comment{This happens with vanishing probability.}
		\State Terminate the algorithm and produce no output.
		\EndIf
		\Statex \emph{To estimate the labels of $P_{X}(j)$:} 
		\State Choose an arbitrary trace $Y \in \mathcal{A}_{j}$; 
		\State Let $v$ denote the node in $Y$ which caused $Y$ to be included in $\mathcal{A}_{j}$; 
		\State Estimate labels of $P_{X}(j)$ by copying bits from the path in $Y$ that goes from the root to~$v$.
		\Statex \emph{To estimate the labels of the children of $j$:} 
		\State Initialize $\mathcal{T} = \emptyset$. 
		\For{$Y \in \mathcal{A}_j$}
		\State Let $v$ denote the node in $Y$ which caused $Y$ to be included in $\mathcal{A}_{j}$; 
		\State Form a string $Z$ by reading, from left to right, the bits of the children of $v$ in $Y$; 
		\State Add $Z$ to $\mathcal{T}$. 
		\EndFor
		\State Use a string trace reconstruction algorithm to estimate the labels of the children of $j$ from~$\mathcal{T}$.
		\EndFor
		\State \textbf{Output:} Take a union, over all $j \in \mathcal{J}_{d-1}$, of the estimated labels of $G_X(j)$, to estimate the labels of $X$ (as in \remref{SeparateReconstruction}).
	\end{algorithmic}
\end{algorithm}

Lines 1--9 of \algref{alg:TED_k_large} describe the first step of the reconstruction algorithm, where we process all the traces and group them into different sets. Formally, we define a set $\mathcal{A}_{j}$ for every $j \in \mathcal{J}_{d-1}$, which we initialize with $\mathcal{A}_{j} = \emptyset$. 
Then for every trace $Y_{t}$ in our input, we run \algref{findpaths_kary_large_TED} with input $Y_{t}$, and we let $\widehat{\mathcal{S}}_{t}$ denote the output. 
We then add $Y_{t}$ to $\mathcal{A}_{j}$ for every $j \in \widehat{\mathcal{S}}_{t}$. 

We now turn to the main step of the reconstruction algorithm, which is described in lines 10--25 of \algref{alg:TED_k_large}. 
For every $j \in \mathcal{J}_{d-1}$, 
we use the traces in $\mathcal{A}_{j}$ to estimate the labels of $G_{X}(j)$, and finally we take a union of these estimates to estimate the labels of $X$. 
The estimation of the labels of $G_{X}(j)$ is done in two parts: 
(1) the estimation of the labels of $P_{X}(j)$, 
and (2) the estimation of the labels of the children of $j$; 
see \figref{Aj-example} for an illustration. 

To estimate the labels of $P_{X}(j)$, we take an arbitrary trace $Y \in \mathcal{A}_{j}$; if $\mathcal{A}_{j} = \emptyset$, then the algorithm terminates without output. 
Let $v$ denote the node in $Y$ which caused $Y$ to be included in $\mathcal{A}_{j}$. 
Assuming that $Y$ was included in $\mathcal{A}_{j}$ for the correct reason, that is, the pre-image of $v$ is indeed $j$, then the labels of $P_{X}(j)$ are identical to the bits on the path in $Y$ that goes from the root to $v$; see \figref{Aj-example} for an illustration. 
Therefore we estimate the labels of $P_{X}(j)$ by copying the bits from the path in $Y$ that goes from the root to $v$.

Finally, we estimate the labels of the children of $j$; it turns out that this reduces to string trace reconstruction. 
Given a trace $Y \in \mathcal{A}_{j}$, let $v$ denote the node in $Y$ which caused $Y$ to be included in $\mathcal{A}_{j}$. 
Assuming that $Y$ was included in $\mathcal{A}_{j}$ for the correct reason, that is, the pre-image of $v$ is indeed $j$, then the children of $v$ in $Y$ are a random subset of the children of $j$ in $X$; see \figref{Aj-example} for an illustration. 
Thus if we restrict our attention to the bits on the children of $j$ in $X$, 
the children of $v$ in the trace $Y$ are as if the original bits were passed through the string deletion channel; see \figref{Aj-example} again for an illustration. 
This motivates collecting a string trace from each $Y \in \mathcal{A}_{j}$, by looking at the children of the appropriate vertex $v$; we let $\mathcal{T}$ denote this collection of string traces. 
Finally, we use a string trace reconstruction algorithm to reconstruct the bits on the children of $j$ in $X$ from $\mathcal{T}$.

\begin{figure*}[t!]
	\centering
	\includegraphics[angle=0,width=0.7\textwidth]{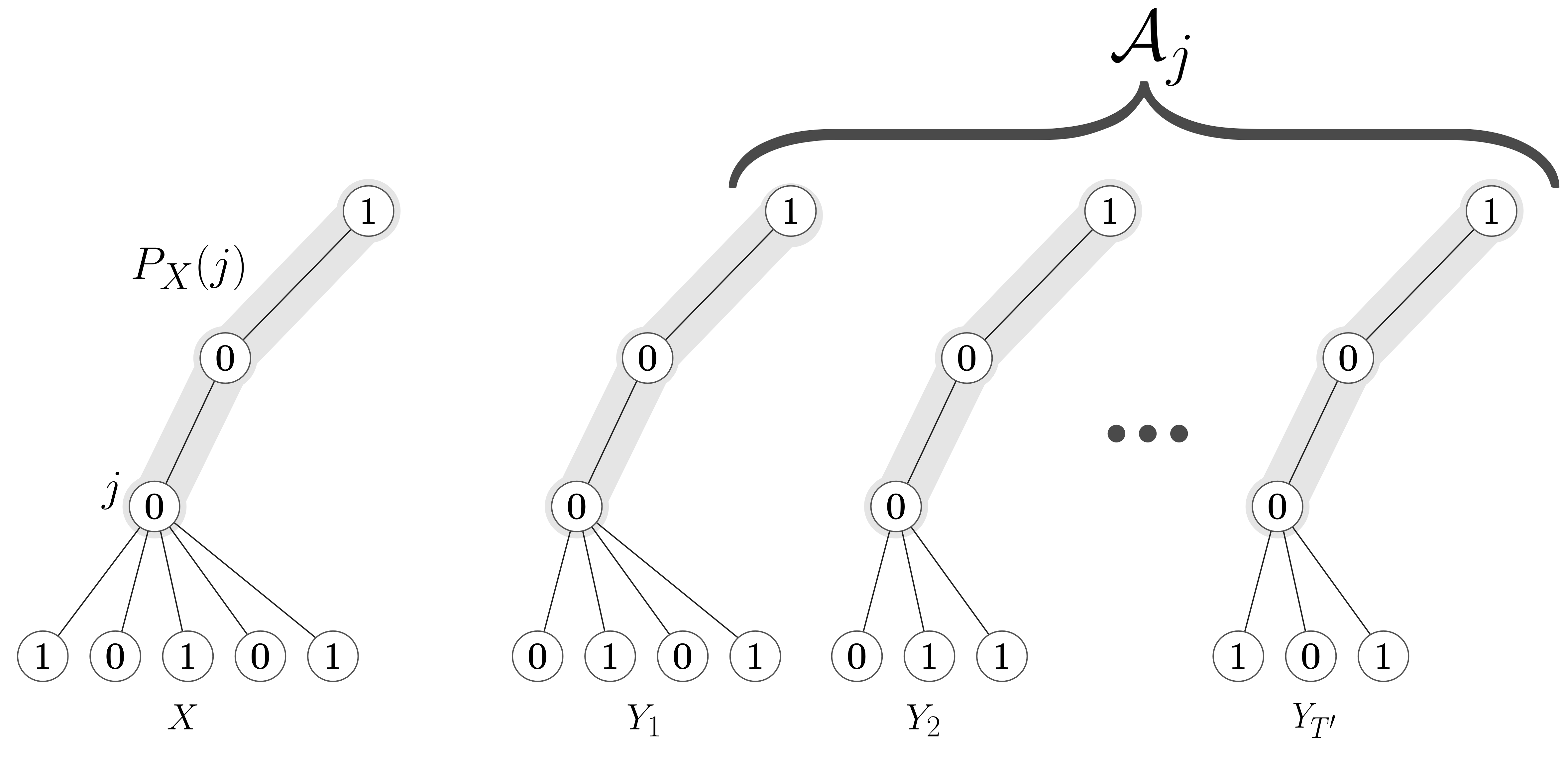}
	\caption{An example set of traces $\mathcal{A}_j$ appearing in \algref{alg:TED_k_large}. The two salient points are that, with high probability: (1)  the bits on the highlighted paths are all the same as the original bits in $P_{X}(j)$, and (2) the restriction to leaves corresponds to string trace reconstruction.} 
	\figlab{Aj-example}
\end{figure*}

Now that we have fully described the reconstruction algorithm, we are ready to prove that it correctly reconstructs the labels of $X$ with high probability. 
The following lemma is an important step towards this. 

\begin{lemma}\lemlab{alg_analysis}
There exist finite positive constants $c'$ and $c''$ such that the following holds. 
Let $T = \exp( c' \log_{k} (n)) \cdot T(k,1/n^{2})$ 
and let $Y_{1}, \ldots, Y_{T}$ be i.i.d. traces from the TED deletion channel. 
With probability at least $1 - \exp( - c'' \sqrt{k})$ the following hold: 
\begin{enumerate}
\item The \textsc{FindPaths} algorithm is fully correct for all traces $Y_{1}, \ldots, Y_{T}$. 
\item For every $j \in \mathcal{J}_{d-1}$ we have that $\left| \mathcal{A}_{j} \right| \geq T(k,1/n^{2})$.
\end{enumerate}
\end{lemma}
\begin{proof}
The first claim follows from \lemref{findpaths} and a union bound, using the fact that 
$T = \exp \left( O( \log_{k}(n) + k^{1/3} ) \right)$. 
For each $j \in \mathcal{J}_{d-1}$, the path $P_{X}(j)$ consists of $d-1$ non-root nodes and hence it survives in a trace with probability $(1-q)^{d-1}$. 
Since $T \geq 2 (1-q)^{-(d-1)} T(k,1/n^{2})$, the second claim follows from a standard Chernoff bound. 
\end{proof}

\subsubsection*{Finishing the proof of \thmref{ted-large}}

\begin{proof}[Proof of \thmref{ted-large}] 
The reconstruction algorithm is described in \algref{alg:TED_k_large}, with a subroutine described in \algref{findpaths_kary_large_TED}. 
Let $\mathcal{E}$ be the event that 
(1) the \textsc{FindPaths} algorithm is fully correct for all traces $Y_{1}, \ldots, Y_{T}$, 
and 
(2) for every $j \in \mathcal{J}_{d-1}$ we have that $\left| \mathcal{A}_{j} \right| \geq T(k,1/n^{2})$. 
By \lemref{alg_analysis} we have that $\mathbb{P} \left( \mathcal{E} \right) \geq 1 - \exp( - c'' \sqrt{k})$. 

Conditioned on the event $\mathcal{E}$, the reconstruction algorithm correctly reconstructs the labels of $P_{X}(j)$ for every $j \in \mathcal{J}_{d-1}$ (see lines 14--17 of \algref{alg:TED_k_large}); in other words, the reconstruction algorithm correctly reconstructs the labels of all internal nodes of $X$. 

We next turn to the leaves of $X$. Conditioned on the event~$\mathcal{E}$ we have that 
$\left| \mathcal{A}_{j} \right| \geq T(k,1/n^{2})$, so using string trace reconstruction we can correctly reconstruct the labels of all children of $j$ with probability at least $1-1/n^{2}$. Since there are at most $n$ nodes in $\mathcal{J}_{d-1}$, a union bound shows that, conditioned on the event~$\mathcal{E}$, we can correctly reconstruct the labels of all leaves of $X$ with probability at least $1-1/n$. 

Overall, the error probability in reconstructing the labels of $X$ is at most $\exp( - c'' \sqrt{k}) + 2/n$. 
\end{proof}

\subsection{Proof of \thmref{ted-small} Concerning Arbitrary Degree Trees}
Recall the definition for $\cal{I}$ defined in the first paragraph of \secref{tree-prelims}: 
$\calI := \bigcup_{t = 1}^{d-1} \calI_{t}$, where 
 $\calI_1 := \calJ_1 = \{0,1,\ldots, k-1\}$ and for $t \geq 2$,
$\calI_t := \{i \in \calJ_t \mid i\ \mathrm{mod}\ k \neq 0\}$, 
and $\calJ_t$ is the set of the nodes at depth~$t$.
We use traces that have a strong underlying structure, which we call $s$-stable; see \figref{stable} for an illustration.

\begin{definition}\deflab{stable}  
	A trace $Y$ is {\em $s$-stable for $i \in \calI$} if $G_Y(i) \neq \perp$, and for every internal node $v$ in $G_Y(i)$ with height $h \leq s$ in $Y$, each of the $k$ children of $v$ has height exactly $h-1$ in $Y$.
\end{definition}

\begin{figure*}[h!]
	\centering
	\includegraphics[angle=0,width=0.4\textwidth]{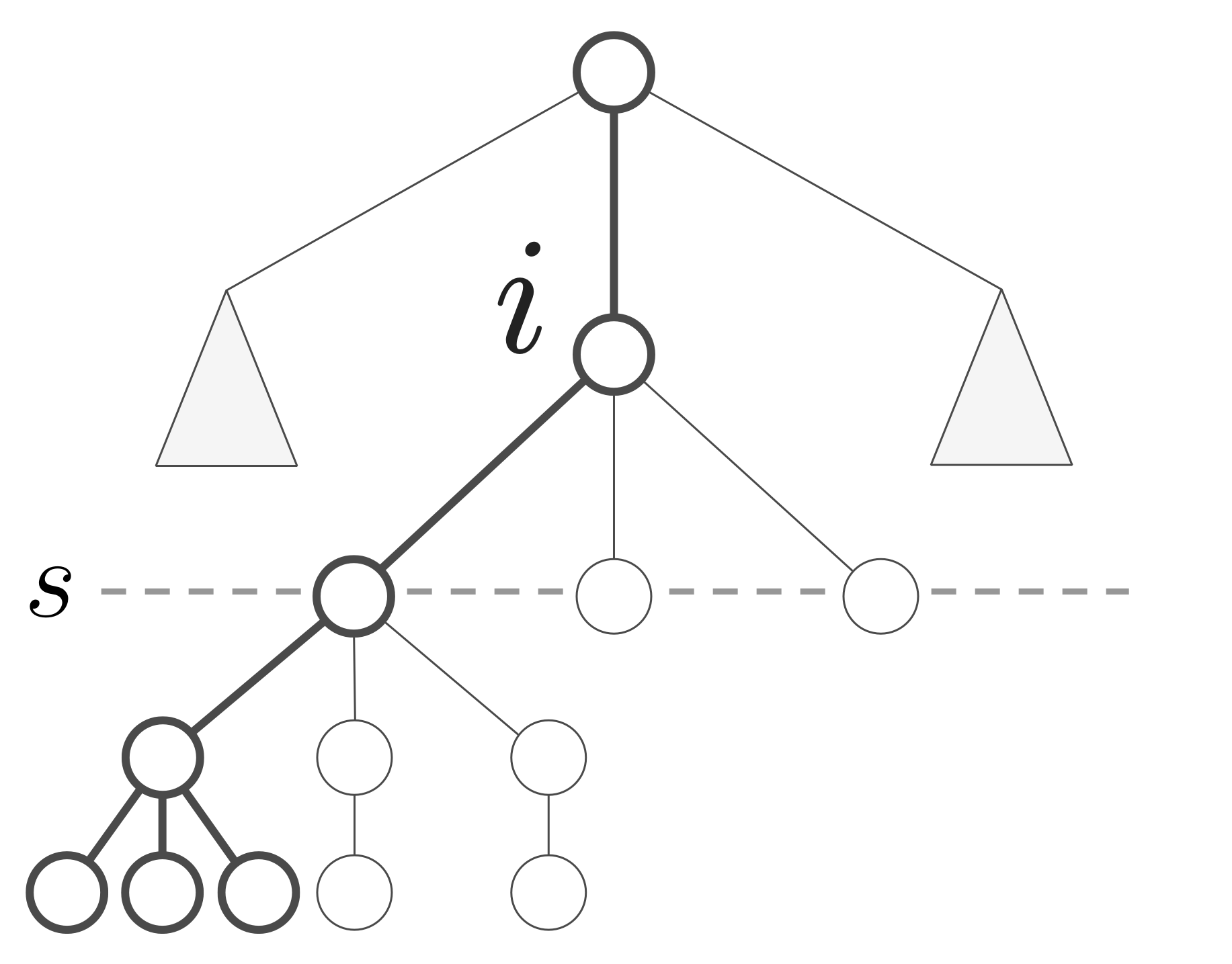}
	\caption{An example of a trace that is $s$-stable for $i$ (where here $s=3$).} 
	\figlab{stable}
\end{figure*}

\algref{kary_arbitrary_TED} below states, in pseudocode, our reconstruction algorithm for proving \thmref{ted-small}. 
At a high level, we will recover the labels for $G_X(i)$ separately for each $i \in \calI$, which is sufficient because these subtrees cover all of the non-root nodes in $X$. 
\begin{algorithm}
	\caption{Reconstructing $k$-ary trees, arbitrary $k$, TED deletion model}
	\alglab{kary_arbitrary_TED}
	\hspace*{\algorithmicindent} Set $s = \lceil  \log_k \log_{1/q} (3dk)\rceil $ and 
		$T = C \log(n) \cdot (1-q)^{-(dk+s^2k)}$ (for a large enough $C$). \\
	\hspace*{\algorithmicindent} \textbf{Input:} traces $Y_{1}, \ldots, Y_{T}$ sampled independently from the TED deletion channel.
	\begin{algorithmic}[1] 
		\State Set  $\mathcal{A} = \left\{ Y_{1}, \ldots, Y_{T} \right\}$.
		\For{$i \in \mathcal{I}$}
		\State Initialize $\mathcal{A}_i = \emptyset$.  
		\For{$t = 1$ to $T$}
		\If{$Y_t$ is $s$-stable for $i$}
		add $Y_t$ to $\mathcal{A}_i$.
		\EndIf	
		\EndFor
		\For{node $b$ in $G_X(i)$}
		\State Let the learned label of $b$ be the majority vote over all $Y \in  \mathcal{A}_i$ of the labels on node $b$ in $G_Y(i)$, as in \lemref{stable-prob}.
		\EndFor
		\EndFor
		\State \textbf{Output:} Union the learned labels of $G_X(i)$ over all $i \in \mathcal{I}$ to reconstruct labels of $X$, as described in \remref{SeparateReconstruction}.
	\end{algorithmic}
\end{algorithm}

The challenge is that, in the TED deletion model,  $G_X(i)$ may shift to an incorrect position, even when $G_Y(i) \neq \perp$. This happens, for example, when the parent of $i$ has children deleted in such a way that $i$ moves to the left or right, but $i$ still has $k-1$ siblings (some of which are new); see \figref{GXplus-with-traces} for an illustration.
The intuition for overcoming this issue is as follows. Let $u$ be a node in $G_X(i)$ with child $u'$ that is not a leaf (so $u$ and $u'$ both originally have $k$ children). If $u$ and all of its $k$ children survive in a trace, then we will be in good shape. However, consider the situation when $u$ survives and $u'$ is deleted. In the TED model, we expect $(1-q)k$ children of $u'$ to move up to become children of $u$. Since this occurs for every deleted child of $u$, we expect $u$ to now have many more than $k$ children. 

The bad case is when $u$ has exactly $k$ children in a trace after some of its original children are deleted; see \figref{GXplus-with-traces} for an illustration. This only happens when subtrees rooted at children of $u$ are completely deleted. If such a subtree is large (that is, $u$ is higher up in the tree), then this is extremely unlikely.
To deal with the nodes $u$ closer to the leaves, we use the $s$-stable property to force the relevant subtrees to survive. 

An obvious way for $Y$ to be $s$-stable is for it to contain $G_X(i)$ and enough relevant descendants of nodes in $G_X(i)$. Let $G^+_X(i)$ be the union of $G_X(i)$ and the $k$ children of every internal node in~$G_X(i)$; see \figref{GXplus-example} for an illustration. Then $Y$ will be $s$-stable if it contains $G_X^+(i)$ and at least one path to a leaf (in $X$) from every node in $G^+_X(i)$ with height at most $s$. In \lemref{stable}, we even argue that this happens with high enough probability to achieve the bound in the theorem.

Unfortunately, we cannot directly check whether $Y$ contains the exact nodes in $G^+_X(i)$. We can check if $Y$ is $s$-stable for $i$ by examining the nodes of $G_Y(i)$ and their descendants in~$Y$.  But if $Y$ is $s$-stable, then it is still not necessarily the case that $G_Y(i) = G_X(i)$, since the nodes in $G_X(i)$ may have shifted in $Y$ or been deleted.

To get around this complication, we rely on the $s$-stable property of a trace. We argue in \lemref{stable-prob} that if $s$ is large enough and a trace $Y$ is $s$-stable for $i$, then with probability at least~2/3, we have $G_Y(i) = G_X(i)$. 
We take a majority vote of $G_Y(i)$ over $O(\log n)$ traces $Y$ to recover $G_X(i)$ with high probability. Since the subtrees $G_X(i)$ for $i \in \calI$ cover $X$, we will be done.

\begin{figure*}[t!]
	\centering
	\includegraphics[angle=0,width=0.95\textwidth]{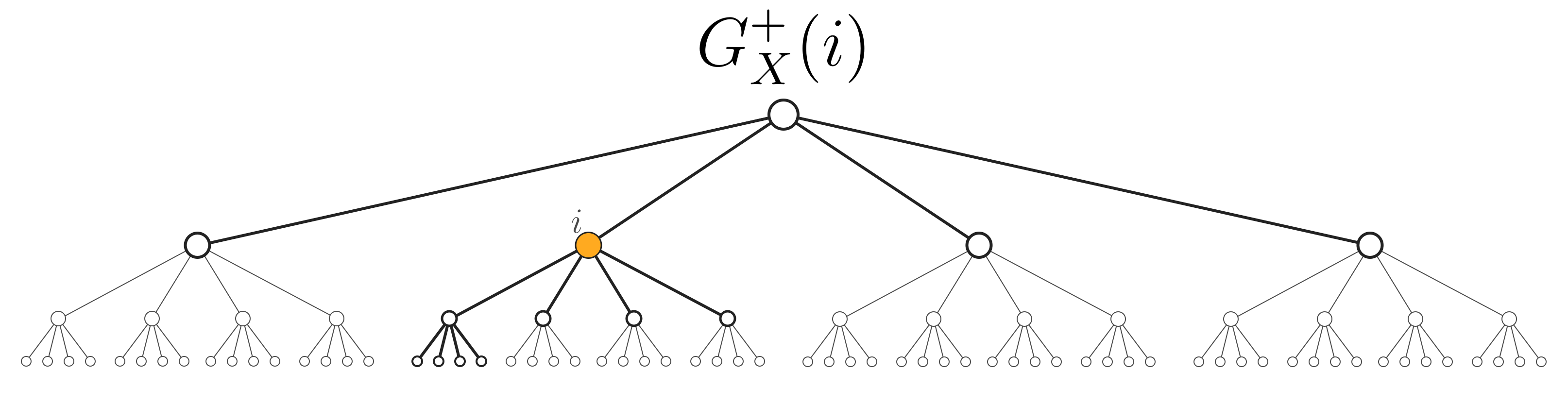}
	\caption{Example of $G^+_X(i)$, where the node $i$ is orange, and the full subtree is bold.} 
	\figlab{GXplus-example}
\end{figure*}

\begin{figure*}[t!]
	\centering
	\includegraphics[angle=0,width=0.95\textwidth]{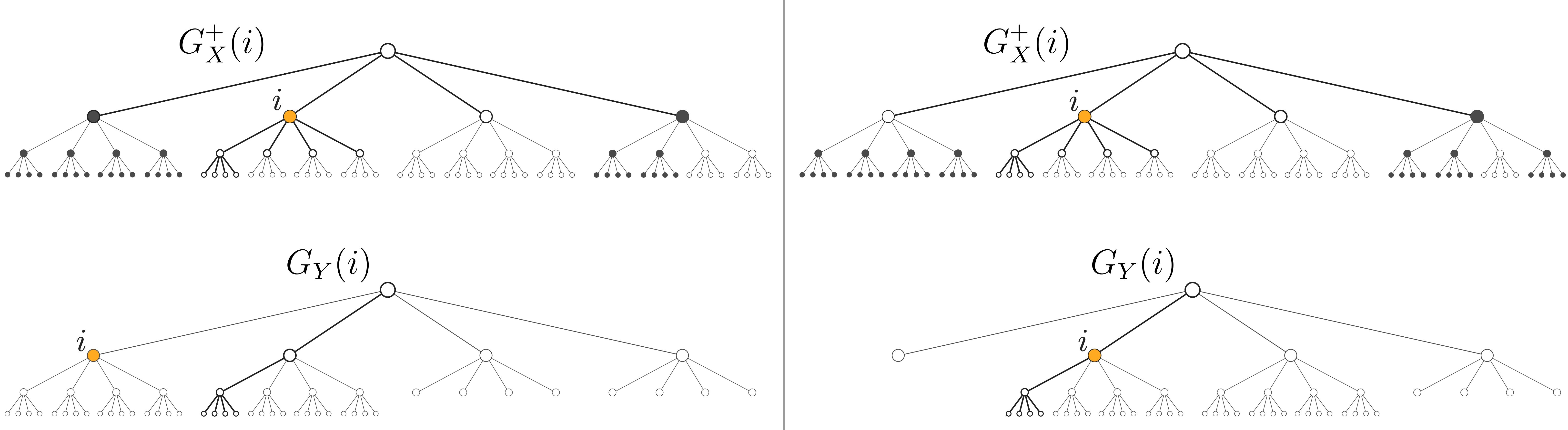}
	\caption{\textbf{Two traces, one ``bad'' and one ``good''.} In the top two trees, gray nodes in $X$ are deleted to produce the corresponding traces below. The trace on the left has the subtree rooted at~$i$ in the incorrect location (it moved over to the left). The trace on the right has the subtree in the correct location.} 
	\figlab{GXplus-with-traces}
\end{figure*}

\subsubsection*{Analyzing and using stable traces}

\newcommand{\sbound}{\left \lceil \log_k \log_{1/q}(3dk)\right \rceil}
In what follows, we fix $s = \sbound$. We first show that a trace is $s$-stable with good enough probability.

\begin{lemma}\lemlab{stable} For $i \in \calI$, a trace is $s$-stable for $i$ with probability at least $(1-q)^{dk + s^2k}$.
\end{lemma}
\begin{proof}
	Being $s$-stable has two conditions. First, we need $G_Y(i) \neq \perp$. Let $G^+_X(i)$ be the union of $G_X(i)$ and the $k$ children of every internal node in $G_X(i)$, where $|G^+_X(i)| = dk+1$. We will prove that if $Y$ contains $G^+_X(i)$, then $G_Y(i) \neq \perp$, because in fact, $G_Y(i) = G_X(i)$. Since the root is never deleted, all nodes in $G^+_X(i)$ survive in a trace with probability $(1-q)^{dk}$, and so $G_Y(i) = G_X(i)$ with at least this probability.
	
	Assume that $Y$ contains $G^+_X(i)$. Let $G_X(i) = u_0,\ldots, u_{d+k-1}$, and consider building $G_Y(i) = v_0,\ldots, v_{d+k-1}$ using $\pi_i$. We argue recursively: For $t \in [d-1]$, we assume that $v_{t'} = u_{t'}$ for all $t' < t$, and we prove that $v_{t} = u_{t}$ as well. The base case $t' = 0$ holds because the root $v_0 = u_0$ is never deleted. Then, since $Y$ contains $G^+_X(i)$, we know that $v_{t'} = u_{t'}$ has exactly $k$ children in $Y$, which are the children of $u_{t'}$ in $X$. Moreover, the left-to-right order of these $k$ children is preserved in the deletion model. Therefore, the child of $v_{t'}$ in position $\pi_i(t')$ must indeed be $u_{t'+1}$ for all $t' < t$. 
	This establishes $v_t = u_t$ for all $t \in \{0,1,\ldots, d-1\}$.
	For the leaves of $G_X(i)$, when $v_{d-1} = u_{d-1}$, and $v_{d-1}$ has $k$ children in $Y$, then we must also have $v_{d},\ldots, v_{d+k-1} = u_{d},\ldots, u_{d+k-1}$.
	
	For the second condition of $s$-stability, consider an internal node $u_t$ in $G_X(i)$ with height $h = d-t$ satisfying $1 \leq h \leq s$. Let $u'_0,\ldots, u'_{k-1}$ be the children of $u_t$ in~$X$. Because $u'_j$ has height $h-1$ in $X$, there is some path with $h$ nodes from $u'_{j}$ to a leaf in $X$. Consider one such path for each $j = 0,\ldots, k-1$ such that $j \neq \pi_i(t)$. Since there are $k-1$ choices for $j$, let $P_t$ be the union of these $k-1$ paths, where $|P_t| = h(k-1) \leq s(k-1)$. The survival of $P_t$ guarantees that $u'_j$ has the correct height for $Y$ to be $s$-stable. Since $|\bigcup_{t=d-s}^{d-1} P_t | \leq s^2(k-1)$, and each node survives independently with probability $(1-q)$, we have that $P_{d-s},\ldots, P_{d-1}$ survive with probability at least $(1-q)^{s^2(k-1)}$.

	Combining these two conditions, $Y$ is $s$-stable with probability at least $(1-q)^{dk + s^2k}$.
\end{proof}

We now formalize the intuition that if all nodes in $G_Y(i)$ have $k$ children, and the parents are high enough in the tree, then the children are probably correct. The reason is that subtrees rooted at their children are unlikely to be completely deleted. This is the only bad case, since otherwise, we expect deleted nodes to cause their parents to have many more than $k$ children. Finally, since the trace is $s$-stable, the nodes near the leaves will be correct as well.

\begin{lemma}\lemlab{stable-prob} For $i \in \calI$, if $Y$ is a random $s$-stable trace for $i$, then $G_Y(i) = G_X(i)$ with probability at least $2/3$.
\end{lemma}
\begin{proof}	
	Since $Y$ is $s$-stable, $G_Y(i) \neq \perp$. 
	Let $G_Y(i) = v_0,\ldots, v_{d+k-1}$ and $G_X(i) = u_0,\ldots, u_{d+k-1}$, where $v_t$ and $u_t$ have depth~$t \in \{0,1,\ldots, d-1\}$, and $v_{d-1}$ and $u_{d-1}$ have children $v_{d}, \ldots, v_{d+k-1}$ and $u_{d}, \ldots, u_{d+k-1}$, respectively. 
	Our strategy is to define an event $\calE$ that happens with probability at least $2/3$ and implies that $v_t = u_t$ for $t \leq d+k-1$. Consider $t \in [d]$, and let $u'_0,\ldots, u'_{k-1}$ be the children of $u_{t-1}$ in $X$. Define $\calE_t$ to be the event that, for every $j \in \{0,1,\ldots, k-1\}$, at least one node in the subtree rooted at $u'_j$ survives in $Y$. Then, define $\calE_{\leq m}  = \bigcap_{t = 1}^{m} \calE_t$ and set $\calE = \calE_{\leq d}$.

	We first argue that when $\calE_{\leq m}$ holds, then $v_t = u_t$ for all $t \leq m$.  Because the root has not been deleted, we have $v_0 = u_0$. Then, for $t \in [m]$, we assume that $v_{t'} = u_{t'}$ for $t' < t$, and we prove that $v_t = u_t$.    
 
	Because $Y$ is $s$-stable, $v_{t-1}$ has $k$ children in $Y$. Denote them $v'_0,\ldots, v'_{k-1}$. We need to show that~$u_{t}$ is in position $\pi_i(t-1)$ among them, so that $v_t = v'_{\pi_i(t-1)} = u_t$. Since~$\calE_{t}$ holds, there is some surviving node in $Y$ from the subtree rooted at each original child of $u_{t-1}$ in $X$. Moreover, since $u_{t-1} = v_{t-1}$, this accounts for at least $k$ children of $v_{t-1}$ in $Y$. Because there are exactly $k$ children of $v_{t-1}$, it must be the case that $v'_{\pi_i(t-1)}$ is originally from the subtree rooted at $u_{t}$ in $X$. In particular, $v'_{\pi_i(t-1)} = u_t$ if and only if $u_t$ survives in $Y$. 
	 
	We claim that if $u_t$ were deleted, then it would contradict $Y$ being $s$-stable, since we would have $G_Y(i) = \perp$ instead. Indeed, the deletion of $u_t$ would cause $v'_{\pi_i(t-1)}$ to have height less than $d-t$ in $Y$. This would imply that at some depth $d'$ with $t < d' < d$, the node $v_{d'}$ in $G_Y(i)$ would be a leaf, leading to $G_Y(i) = \perp$. We conclude that~$u_t$ survives in $Y$, and so that $v_t = v'_{\pi_i(t-1)} = u_t$, as desired.

	We have  shown that $\calE$ guarantees that $v_{t} = u_{t}$ for all $t \leq d-1$. In particular, $v_{d-1} = u_{d-1}$, and the~$k$ children of $v_{d-1}$ in $Y$ must be the children of $u_{d-1}$ in $X$. This finishes the argument that $\calE$ implies that $v_t = u_t$ for all $t \leq d+k-1$, that is, $G_Y(i) = G_X(i)$.
	
	Now, we prove that $\calE$ happens with probability at least $2/3$ in an $s$-stable trace. We prove this in two steps. First, we argue that~$\calE_{\leq d-s}$ occurs with probability at least $2/3$. Then, we show that $\calE_{\leq d-s}$ implies $\calE$. 
	Consider the node $u_{t-1}$ in $G_X(i)$ for $t \in [d-s]$, and let $u'_0,\ldots, u'_{k-1}$ be the~$k$ children of~$u_{t-1}$ in $X$. Since the height of $u'_j$ is at least $s$, the subtree rooted at $u'_j$ in $X$ contains at least $\sum_{\ell=0}^s k^\ell \geq k^s$ nodes. The probability that all of these nodes are deleted is at most~$q^{k^s}$. Because $s = \sbound$, this is at most $1/(3dk)$. Taking a union bound over the $k$ children implies that $\calE_t$ occurs with probability at least $1 - 1/(3d)$, and taking a union bound over $t \in [d-s]$ implies that $\calE_{\leq d-s}$ holds with probability at least $2/3$.
	
	The final step is to prove that $\calE$ happens with probability one, in an $s$-stable trace, assuming that~$\calE_{\leq d-s}$ holds. More precisely, we will show that $\calE_{\leq d-s+\ell}$ implies $\calE_{d-s+\ell+1}$ for $\ell = 0,1\ldots, s-1$. We have already argued that $\calE_{\leq d-s +\ell}$ guarantees that  $v_{d-s +\ell} = u_{d-s +\ell}$. We claim that the~$k$ children $v'_0,\ldots, v'_{k-1}$ of $v_{d-s +\ell}$ are the original children of $u_{d-s+\ell}$ in $X$ (and this clearly implies $\calE_{d-s+\ell+1}$). Since $Y$ is $s$-stable, there is a path with $s -\ell +1$ nodes from $v'_j$ to a leaf in $Y$. If $v'_j$ were not an original child of $u_{d-s+\ell}$, then all such paths would have at most $s-\ell$ nodes. This implies no children of $u_{d-s+\ell} = v_{d-s+\ell}$ have been deleted in~$Y$, and their existence witnesses the survival of the subtrees needed for $\calE_{d-s+\ell+1}$. Since this holds for $\ell = 0,1\ldots, s$, we conclude that $\calE = \calE_{\leq d}$ follows from $\calE_{\leq d-s}$ in an $s$-stable trace, and $\Pr[G_Y(i) = G_X(i)] \geq \Pr[\calE] = \Pr[\calE_{\leq d-s}] \geq 2/3$.
\end{proof}

\subsubsection*{Completing the proof of \thmref{ted-small}}

\begin{proof}[Proof of \thmref{ted-small}]
	
	Let $\A$ be a set of $T = C\log(n)  / (1-q)^{dk + s^2k}$ traces with $C$ a large enough constant. By \lemref{stable}, each trace in $\A$ is $s$-stable for $i$ with probability $(1-q)^{dk + s^2 k}$. Therefore, by setting $C$ large enough and taking a union bound over $i \in \calI$, we can ensure that with probability at least $1-1/n^2$, for every $i \in \mathcal{I}$ there is a subset $\A_i \subseteq \A$ of $s$-stable traces for~$i$ with $|\A_i| \geq C' \log n$, for a constant $C'$ to be set later.
	
	By  \lemref{stable-prob}, each trace $Y  \in \A_i$ has the property that $G_Y(i) = G_X(i)$ with probability at least~$2/3$. Let $f_i(Y) \in \bit^{d+k-1}$ be the labels of  $G_Y(i)$ in $Y$. In expectation over $Y  \in \A_i$, we have that at least a 2/3 fraction of $Y$ satisfy $f_i(Y) = f_i(X)$. Therefore, since $|\A_i| \geq C' \log n$ for a large enough constant $C'$, we have by a standard Chernoff bound that the majority value of $f_i(Y)$ over $Y  \in \A_i$ is equal to $f_i(X)$, with probability at least $1-1/n^2$. For each $i \in \calI$, our reconstruction algorithm uses this majority vote to deduce the labels for $G_X(i)$. Taking a union bound over~$i \in \calI$, where $|\calI| \leq n$, we correctly label all nodes with probability at least $1-1/n$.
	
	It remains to show that $T = \exp(O(dk))$, where $d = O(\log_k n)$. Recall that  we have set $s = \sbound$. If $k \geq d$, then $s \leq c \log \log k / \log k$ for a constant $c$, since $q$ is a constant, and so $s = O(1)$. If $k \leq d$, then $s \leq c \log \log d$, and in particular, $s^2 < c''d$ for some constant $c''$. Therefore, for any $k$, we have $T \leq C \log n \cdot \exp(c'dk)$ for some constant $c'> 1$ depending only on~$q$, and since $\log \log n < dk$, we conclude that that $T = \exp(O(dk))$. 
\end{proof}


\section{Reconstructing Trees, Left-Propagation Model}\seclab{left}

In this section we present our two algorithms for $k$-ary trees in the Left-Propagation deletion model.

\begin{figure*}[t!]
	\centering
	\includegraphics[angle=0,width=0.6\textwidth]{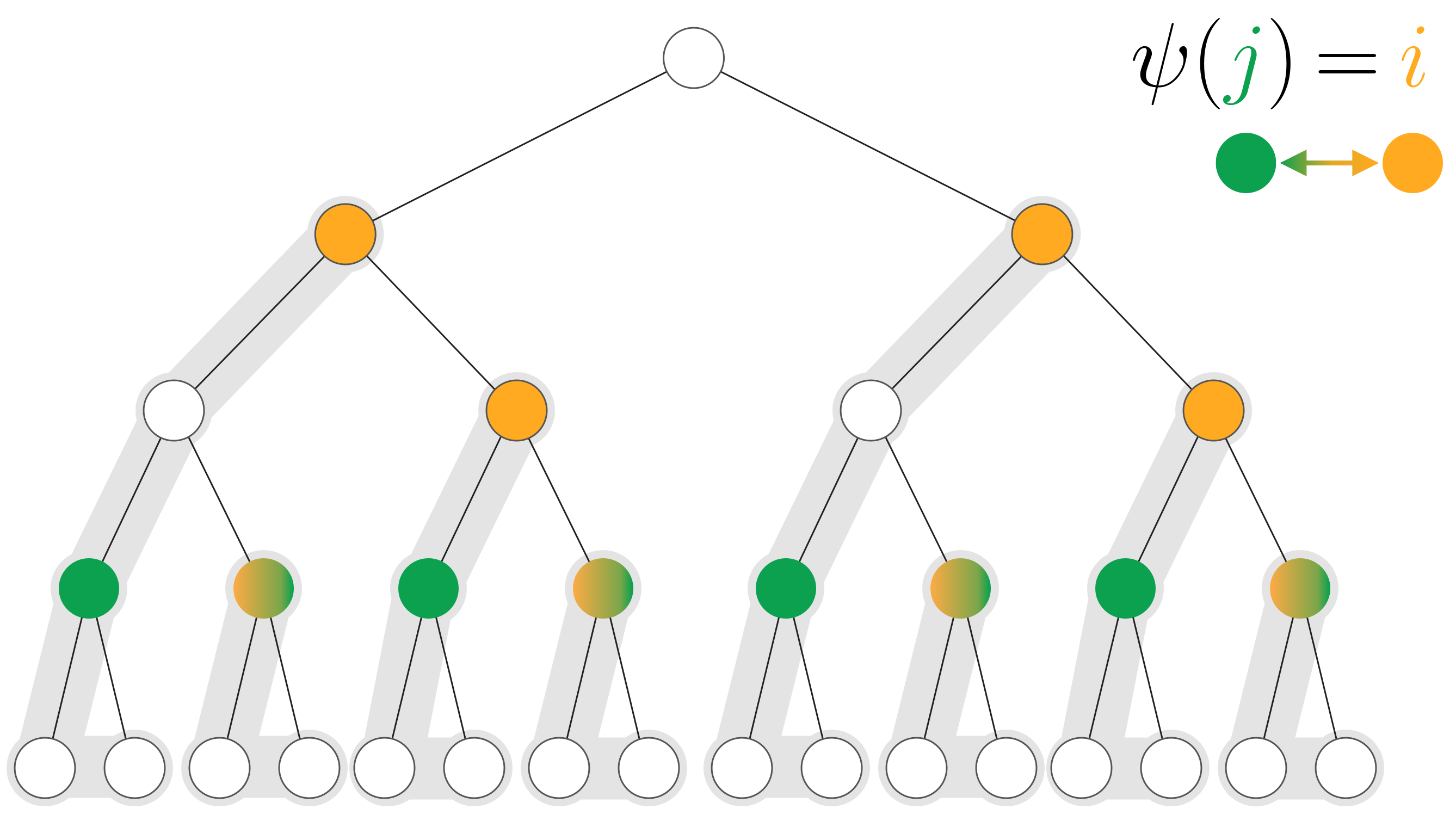}
	\caption{The partition of the tree $X$ into the subtrees $\left\{ H_X(i) : i \in \mathcal{I} \right\}$ is illustrated in gray. The bijection $\psi : \mathcal{J}_{d-1} \to \mathcal{I}$ maps, for each gray subtree, a single green node to a single orange node in that subtree. Nodes that are both green and orange map to themselves.} 
	\figlab{left-prop-partition}
\end{figure*}

\subsection{Proof of \thmref{left-large} Concerning Large Degree Trees}\seclab{outline-left-large}
Recall the definitions from the first paragraph in \secref{tree-prelims}, in particular that of $\calI$.
Recall also that the sets $H_X(i)$ for $i \in \calI$ partition the non-root nodes of $X$, and each $H_X(i)$ contains exactly one node from $\calJ_{d-1}$; see \figref{left-prop-partition} for an illustration.
	 Define the bijection $\psi: \calJ_{d-1}  \to \calI$ as $\psi(j) = i$ for the distinct $i$ such that $j \in H_X(i)$; see \figref{left-prop-partition} again for an illustration. 
	For each fixed $j \in \calJ_{d-1}$ and each trace $Y$, we will extract a bit-string $s_Y(j)$
	and use it to reconstruct the labels for $H_X(\psi(j))$. 
	We only define $s_{Y}(j)$ whenever $P_{Y}(j) \neq \perp$, but this suffices for our purposes (as discussed later). 
	To define $s_Y(j)$, we need some notation.
	Let $v_0,v_1,\ldots, v_{d-1}$ be the nodes in $P_Y(j)$, where $v_t$ has depth $t$ in $Y$. Then, let $v_d,\ldots,v_{d+k'}$ be the children of $v_{d-1}$ in $Y$, where $k' \leq k -1$. Let $i = \psi(j)$ and let $t_i$ be the depth of $i$ in $X$. Finally, define $s_Y(j)$ as a bit-string of length $d+k'-t_i+1$ consisting of the labels in $Y$ of the nodes $v_{t_i},\ldots, v_{d+k'}$; see \figref{left-prop-traces-large-degree} for an illustration.

\begin{figure*}[h!]
	\centering
	\includegraphics[angle=0,width=0.8\textwidth]{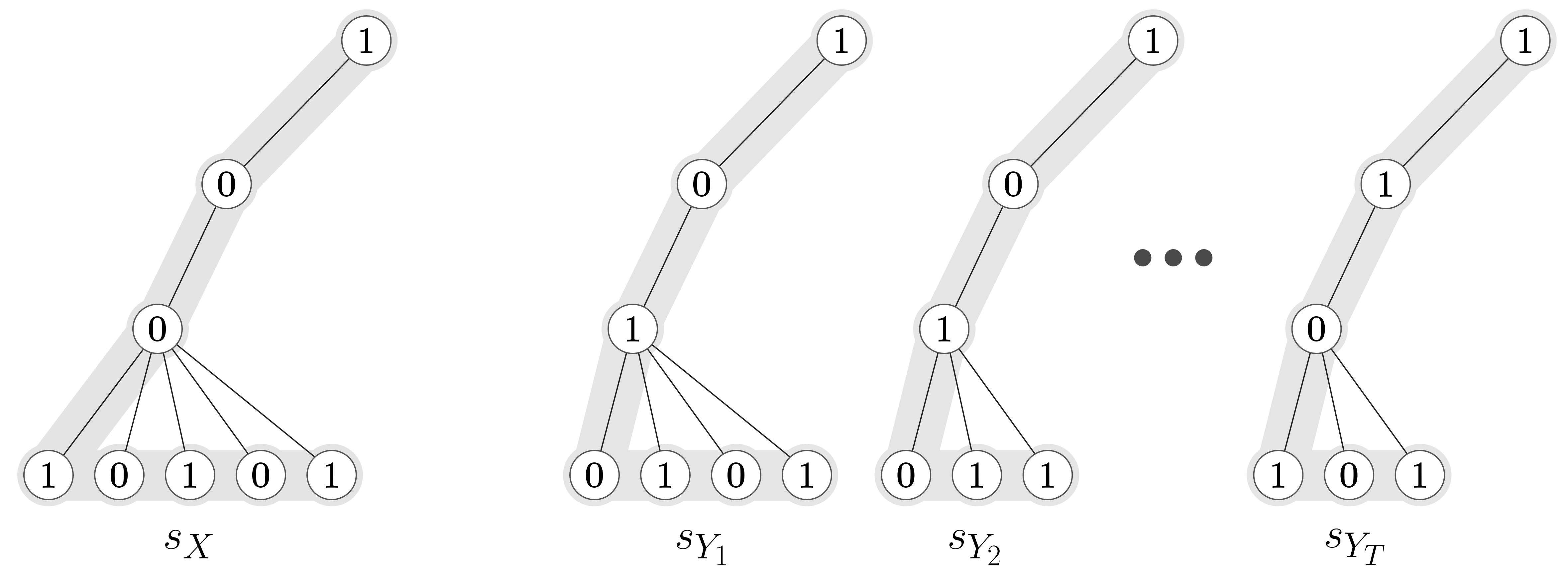}
	\caption{Extracting the strings $s_Y(j)$ from the subtrees that contain the path $P_Y(j)$.} 
	\figlab{left-prop-traces-large-degree}
\end{figure*}

\algref{kary_large_LP} below states, in pseudocode, our reconstruction algorithm for proving \thmref{left-large}. 
We note that this reconstruction algorithm is tailored to the Left-Propagation deletion model. 
\begin{algorithm}
	\caption{Reconstructing $k$-ary trees, $k \geq c \log(n)$, Left-Propagation deletion model}
	\alglab{kary_large_LP}
	\hspace*{\algorithmicindent} Set $T = T(d+k,1/n^2)$. \\
	\hspace*{\algorithmicindent} \textbf{Input:} traces $Y_{1}, \ldots, Y_{T}$ sampled independently from the TED deletion channel.
	\begin{algorithmic}[1] 
		\State Set $\mathcal{A} = \left\{ Y_{1}, \ldots, Y_{T} \right\}$. 
		\If{there exists a trace $Y \in \mathcal{A}$ and a node $j \in \mathcal{J}_{d-1}$ such that $P_{Y}(j) = \perp$ \Comment{This happens with vanishing probability.}\\}
		Terminate the algorithm and produce no output. 
		\Else{ \Comment{This happens with high probability.}
			\For{$j \in \mathcal{J}_{d-1}$}
			\State Reconstruct labels of $H_X(\psi(j))$ from $\left\{s_{Y}(j) \right\}_{Y \in \mathcal{A}}$, via string trace recon.\ (\thmref{StringMain}).
			\EndFor}
		\State \textbf{Output:} Union the learned labels of $H_X(i)$ over all $i \in \calI$ to reconstruct labels of $X$, as described in \remref{SeparateReconstruction}.
		\EndIf
	\end{algorithmic}
\end{algorithm}

For complete $k$-ary trees with sufficiently large $k \geq c \log n$,
a trace $Y$ has $P_Y(j) \neq \perp$ 
for all nodes $j \in \calJ_{d-1}$ with high probability. 
When $P_Y(j) \neq \perp$ and $j \in H_X(i)$,
we can extract a subset of $H_X(i)$ that behaves as if it went through the string deletion channel
(i.e., as if it were a path on $|H_X(i)|$ nodes). 
Therefore, using traces with $P_Y(j) \neq \perp$ for all $j \in \calJ_{d-1}$,
we reduce to string trace reconstruction (see \figref{left-prop-traces-large-degree} for an illustration), and we reconstruct the labels for each $H_X(i)$ separately. This suffices because the subtrees $H_X(i)$ for $i \in \calI$ partition the non-root nodes of $X$ (see \figref{left-prop-partition} for an illustration).

We first argue that $P_Y(j) \neq \perp$ with high probability when $k \geq c \log n$. 

\begin{lemma}\lemlab{left-good-trace}
	Let $k \geq c \log n$. In the Left-Propagation model, a random trace $Y$ has $P_Y(j) \neq \perp$ for every $j \in \calJ_{d-1}$ with probability at least $1 - \exp(-c'k)$.
\end{lemma}
\begin{proof}
	The property that $Y$ has $P_Y(j) \neq \perp$ for every $j \in \calJ_{d-1}$ is equivalent to $Y$ containing a complete $k$-ary subtree of depth $d-1$ with the same root as $X$. Consider any node $j \in \calJ_{d-1}$, and recall that the subtree $G_X(j)$ has $d+k-1$ non-root nodes. Each non-root node in $G_X(j)$ survives independently in $Y$ with probability $(1-q)$. Let $\calE'$ be the event that at least~$d$ nodes from $G_X(j)$ survive in $Y$ for every $j \in \calJ_{d-1}$. Because $k \geq c\log n$ and $d \leq  \log_k n$ and $|\calJ_{d-1}| \leq n$, a standard Chernoff and union bound implies that $\calE'$ holds with probability $1 - \exp(-c'k)$ for a constant $c' > 0$ depending on $q$. When $\calE'$ holds, for all $j \in \calJ_{d-1}$, every node in $P_Y(j)$ has exactly $k$ children in $Y$ and $P_Y(j) \neq \perp$ . 
\end{proof}

\begin{proof}[Proof of \thmref{left-large}] 
	Let $T = T(d+k, 1/n^2)$ be the number of traces needed to learn $d+k$ bits with probability $1-1/n^2$ in the string model with deletion probability $q$. 
	We will reconstruct $X$ with probability $1 - O(1/n)$ using~$T$ traces from the Left-Propagation model. 
	
	By \lemref{left-good-trace}, a trace $Y$ has $P_Y(j) \neq \perp$ for every $j \in \calJ_{d-1}$ with probability $1-\exp(-c'k)$. 
	By \thmref{StringMain} we have that $T = \exp( O( (d+k)^{1/3} ) ) = \exp( O( k^{1/3} ) )$, where the second inequality is due to $d \leq k$. 
	Thus by a union bound it follows that, with probability at least $1 - \exp( - c'' k )$ for some constant $c'' > 0$, we have $P_{Y}(j) \neq \perp$ for \emph{all} traces $Y$ and nodes $j \in \mathcal{J}_{d-1}$. 
	So from now on we assume that $P_Y(j) \neq \perp$ for every trace $Y$ and every $j \in \mathcal{J}_{d-1}$.

	Decompose $X$ into subtrees $H_X(\psi(j))$ for $j \in \calJ_{d-1}$; see \figref{left-prop-partition} for an illustration. For each of the $T$ traces $Y$, extract the bit-string $s_Y(j)$; see \figref{left-prop-traces-large-degree} for an illustration. Consider these as $T$ traces from the string deletion model on $|H_X(\psi(j))| < d+k$ bits. More precisely, let $s_X(j)$ be the labels in $X$ for the nodes in $H_X(\psi(j))$. We claim that $s_Y(j)$ is a valid trace for the string deletion model with unknown string $s_X(j)$. In the Left-Propagation model, when $P_Y(j) \neq \perp$, the nodes considered in $Y$ for $s_Y(j)$ form a subsequence of the corresponding nodes in $X$. Therefore, since each node is deleted with probability $q$, the bits in $s_Y(j)$ will be a trace of the string $s_X(j)$. 
	Though we only consider traces with at least $d$ bits remaining, the probability that at least one trace of $T$ has
	less than $d$ bits occurs with probability at most $\exp(-O(k))$.
	So by slightly increasing the factor $C'$ in $T$, we can use
	\thmref{StringMain} to see $T$ traces suffice to reconstruct $s_X(j)$ with probability $1-1/n^2$. 
	Moreover, $s_X(j)$ are the labels for $H_X(\psi(j))$. Taking a union bound over $|\calI| \leq n$, we can reconstruct $H_X(i)$ for all $i \in \calI$ with probability at least $1-1/n$.
\end{proof}

\subsection{Proof of \thmref{left-small} Concerning Arbitrary Degree Trees}\seclab{proof-left-small}

\algref{kary_arb_LP} below states, in pseudocode, our reconstruction algorithm for proving \thmref{left-small}. 
\begin{algorithm}
	\caption{Reconstructing $k$-ary trees, arbitrary $k$, Left Propagation deletion model}
	\alglab{kary_arb_LP}
	\hspace*{\algorithmicindent} Set $T = C (1-q)^{-(d+c'k)} \log(n)$ (for large enough $c'$ and $C$). \\
	\hspace*{\algorithmicindent} \textbf{Input:} traces $Y_{1}, \ldots, Y_{T}$ sampled independently from the TED deletion channel.
	\begin{algorithmic}[1] 
		\For{every $i \in \mathcal{I}$}
		\State Initialize $\mathcal{A}_i = \emptyset$. 
		\For{$t=1$ to $T$}  
		\If{$Y_t$ has $G_{Y_{t}}(i) \neq \perp$}
		add $Y_t$ to $\mathcal{A}_i$.
		\EndIf
		\EndFor
		\EndFor
		\If{there exists $i \in \mathcal{I}$ such that $\mathcal{A}_{i} = \emptyset$ \Comment{This happens with vanishing probability.}\\}
		Terminate the algorithm and produce no output.
		\Else{\Comment{This happens with high probability.}
		\For{every $i \in \mathcal{I}$} 
		\State Choose an arbitrary $Y \in \mathcal{A}_i$.
		\State Reconstruct labels in $H_X(i)$ bit by bit as those of $H_Y(i)$, using \remref{BitByBit}.
		\EndFor}
		\State \textbf{Output:} Union the learned labels of $H_X(i)$ over all $i \in \mathcal{I}$ to reconstruct labels of $X$, as described in \remref{SeparateReconstruction}.
		\EndIf
	\end{algorithmic}
\end{algorithm}

As in the proof of \thmref{left-large}, we reconstruct 
 $X$ by reconstructing the subtrees $H_X(i)$ for $i \in \calI$, 
which partition the non-root nodes of $X$. Instead of reducing to string reconstruction, we  use traces with $G_Y(i) \neq \perp$ to directly obtain labels for $H_X(i)$. 
We only need to take enough traces to balance out the fact that a trace with $G_Y(i) \neq \perp$ 
for $i \in \calI$ occurs with probability $\exp(-O(d+k))$. 

\subsubsection*{Recovering the labels for subtrees}

We first show that if a trace satisfies $G_Y(i) \neq \perp$, then we can reconstruct the labels of $H_X(i)$; see \figref{left-prop-traces-large-degree-caterpillar} for an illustration.

\begin{figure*}[t!]
	\centering
	\includegraphics[angle=0,width=0.75\textwidth]{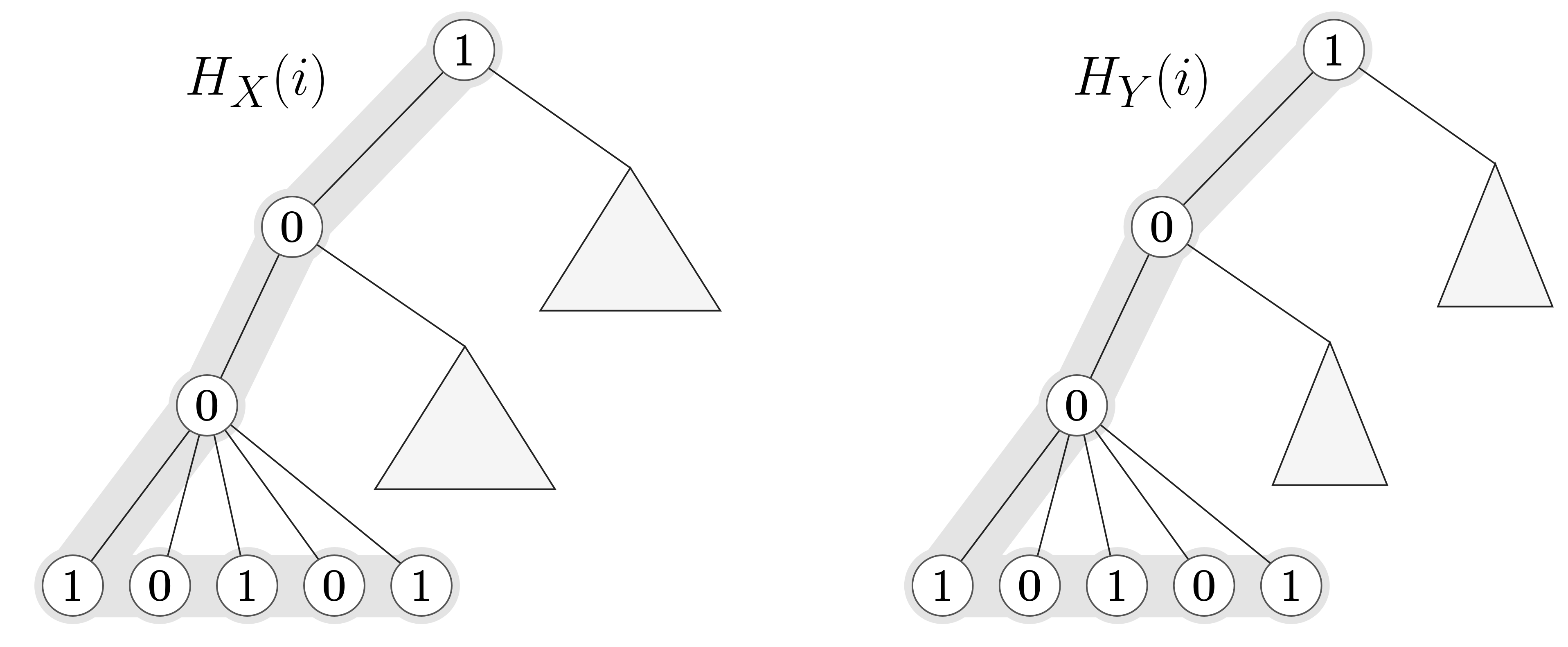}
	\caption{Extracting correct labels for $s_X(j)$ from a single trace containing a caterpillar,  $G_Y(i) \neq \perp$.} 
	\figlab{left-prop-traces-large-degree-caterpillar}
\end{figure*}

\begin{lemma}\lemlab{label-left} In the Left-Propagation model, if $G_Y(i) \neq \perp$, then $H_Y(i) = H_X(i)$ and the labels for these subtrees are identical in $Y$ and $X$.
\end{lemma}
\begin{proof} 
Let $s_Y(i) = v_{\ell},\ldots, v_{b}$ be the labels on the left only path from $i$ to the leaf plus its siblings in trace $Y$. 
These are the labels on $H_Y(i)$.
Similarly, let $s_X(i) = u_{\ell'},\ldots, u_{b'}$ be the labels on the left-only path from $i$ to the leaf plus its siblings in $X$. 
These are the labels on $H_X(i)$. 
Due to the behavior of the Left-Propagation model, the string $s_Y(i)$ is a trace of the string $s_X(i)$. 
When $G_Y(i) \neq \perp$, no nodes of $H_X(i)$ could have been deleted to obtain the trace $H_Y(i)$,
otherwise the number of leaves in $G_Y(i)$ would be too small and $G_Y(i)$ would not be defined.
As $s_Y(i)$ is a trace of $s_X(i)$ and $|H_Y(i)| = |H_X(i)|$, this implies $H_Y(i) = H_X(i)$. 
\end{proof}

\begin{lemma}\lemlab{cater-left} 
Fix $i \in \mathcal{I}$ and let $Y$ be a trace from the Left-Propagation deletion model. 
There exists an absolute constant $c' > 1$ such that with probability at least $(1-q)^{d + c'k}$ we have that $G_Y(i) \neq \perp$. 
\end{lemma}
\begin{proof} There are $d+k$ nodes in $G_X(i)$ and they all survive in $Y$ with probability $(1-q)^{d+k}$. From now on, we assume this holds. Let $G_X(i) = u_0,\ldots, u_{d+k-1}$, where $u_t$ has depth~$t \in \{0,1,\ldots, d\}$, and $u_{d-1}$ has children $u_{d}, \ldots, u_{d+k-1}$. Consider $t \in [d]$, and let $u'_0,\ldots, u'_{k-1}$ be the $k$ children of~$u_{t-1}$ in $X$. Define $\calE_t$ to be the event that, for every $j \in \{0,1,\ldots, k-1\}$, at least one node in the subtree rooted at $u'_j$ survives in $Y$. We observe that, in the Left-Propagation model, if both $\bigcap_{t = 1}^{d} \calE_t$ holds and $G_X(i)$ survives, then we have $G_Y(i) \neq \perp$. 

Because $G_X(i)$ surviving implies that the $k$ children of $u_{d-1}$ survive, we already know that $\calE_d$ holds. For $t \in [d-1]$, each child $u'_j$ of $u_{t-1}$ has height $h = d-t+1$ in $X$. In particular, the subtree rooted at $u'_j$ in $X$ contains at least $k^{d-t+1}$ nodes. If $u'_j \in G_X(i)$, then we have assumed it survives, otherwise there are $k-1$ other subtrees. Since the subtrees considered are independent, at least one node survives from each of them (for all $t \in [d-1]$) with probability at least
\[\prod_{h = 2}^{d} \left (1-q^{k^{h}}\right )^{k-1} = (1-q)^{ck},\] for some constant $c$. Putting everything together, $G_X(i)$ survives and $\bigcap_{t = 1}^{d} \calE_t$ holds, and therefore $G_Y(i) \neq \perp$, with probability at least
$(1-q)^{d + k} \cdot (1-q)^{ck} = (1-q)^{d+(c+1)k}$.
\end{proof}

\subsubsection*{Completing the proof of \thmref{left-small}}\seclab{outline-left-small}

\begin{proof}[Proof of \thmref{left-small}.] 
	By \lemref{cater-left}, 
	the probability that none of $T$ traces satisfy $G_Y(i) \neq \perp$ for some $i \in \calI$ is at most 
	\[|\calI|\pth{1 - (1-q)^{d + c'k}}^T \leq n\exp(-T(1-q)^{d + c'k}).\]
	To ensure that at least one trace has $G_Y(i) \neq \perp$ for every $i \in \calI$ with high probability, we take 
	$T' = O(T \log n)$ traces, where $T$ satisfies
	\[T = (1-q)^{-(d + c'k)} \leq (1-q)^{-\ln(n)/\ln(k) - c'k} = n^{\ln(1/(1-q))(c'k/\ln(n) + 1/\ln(k))}.\]
	By \lemref{label-left}, any trace with $G_Y(i) \neq \perp$ induces the correct labeling of $H_X(i)$ by using the labels for the nodes in $H_Y(i)$. In other words, with high probability, a set of $T' = O(T \log n)$ traces yields a correct labeling of $H_X(i)$ for all $i \in \calI$. Since the subtrees $H_X(i)$ for $i \in \calI$ form a partition of $X$, we can recover all labels in $X$.  
\end{proof}


\section{Reconstructing Spiders}\seclab{spiders}

In this section, we describe how to reconstruct spiders and prove \thmref{MainSpider} and \propref{SpiderCor}.
 We start with preliminaries in \secref{spider-prelims}.
An outline of the proof for \thmref{MainSpider} is followed by the full proof in \secref{proof_mainspider}.
The proof assumes a lemma requiring complex analysis that is deferred to \secref{LemmaContour}.
\propref{SpiderCor} is proven in \secref{spiders-from-strings}.
The remaining proofs of lemmas stated in this section are detailed in \secref{spider_proofs}.

\subsection{Spider Algorithm Preliminaries}\seclab{spider-prelims}

When a labeled $(n,d)$-spider, $X$, goes through the deletion channel, we assume that its trace, $Y$, is an $(n,d)$-spider by inserting length $d$ paths of $0$s after the remaining paths and nodes labeled $0$ to the end of paths. After this, traces have $n/d$ paths of length~$d$ (excluding the root). 

We define a left-to-right ordered DFS index for $(n,d)$-spiders, illustrated in \figref{spider}. 
The labels increase along the length of the paths from the root and increase left to right among the paths. 
Specifically, if node $v$ is in the \nth{i} path from the left and has depth $j$, 
then its label is $(i-1)d + j-1$. 
These labels will be used to define appropriate generating functions. 
As discussed in \remref{root}, 
we need not consider the root as part of the generating function.

\subsection{Proof of \thmref{MainSpider} Concerning $(n,d)$-spiders with Small $d$}\seclab{proof_mainspider}

In the regime where spiders have short paths ($d \leq \log_{1/q} n$), we use mean-based algorithms that generalize the methods of~\cite{DeOdonnellServedio17-WorseCase,de2019optimal,NazarovPeres17-WorseCase}. Using the DFS indexing of nodes, 
let $X$ be an $(n,d)$-spider with labels $\left\{ a_{j} \right\}_{j = 0}^{n-1}$ 
and let $Y$ be a trace of $X$, with the labels of $Y$ denoted by $\left\{ b_{j} \right\}_{j=0}^{n-1}$. 
Consider now the random generating function 
\[
 \sum_{j=0}^{n-1} b_{j} w^{j}
\]
for $w \in \mathbb{C}$. 
Due to the special structure of spiders, the expected value of this random generating function can be computed (see \lemref{GeneratingFunction} below), and while it is more complicated than the corresponding formula for strings, it is still tractable. 
This is useful since by averaging samples we can approximate this expected value.

We then show that for every pair of labeled $(n,d)$-spiders, $X^{1}$ and $X^{2}$, with different binary labels, 
we can carefully choose $w \in \mathbb{C}$ 
so that the corresponding values of the respective generating functions differ in expectation at some index $j = j(X^{1}, X^{2})$. 
In choosing between candidate spiders $X^{1}$ and $X^{2}$, 
the algorithm deems the better match of the pair to be the spider for
which the expected value of the generating function at $w \in \mathbb{C}$ is closer to the mean of the traces at $j$. 
If any spider is a better match compared to every other spider, 
it is said to be the best match, and the algorithm outputs that spider.

For the quantitative estimates, the key technical challenge is to lower bound the modulus of the (expected) generating function on a carefully chosen arc of the unit disc in the complex plane. Our analysis, based on harmonic measure, is inspired by~\cite{Littlewood}, as well as the recent work of~\cite{HartungHP18}.

\begin{figure}[t]
	\center
	\includegraphics[angle=0,width=0.55\textwidth]{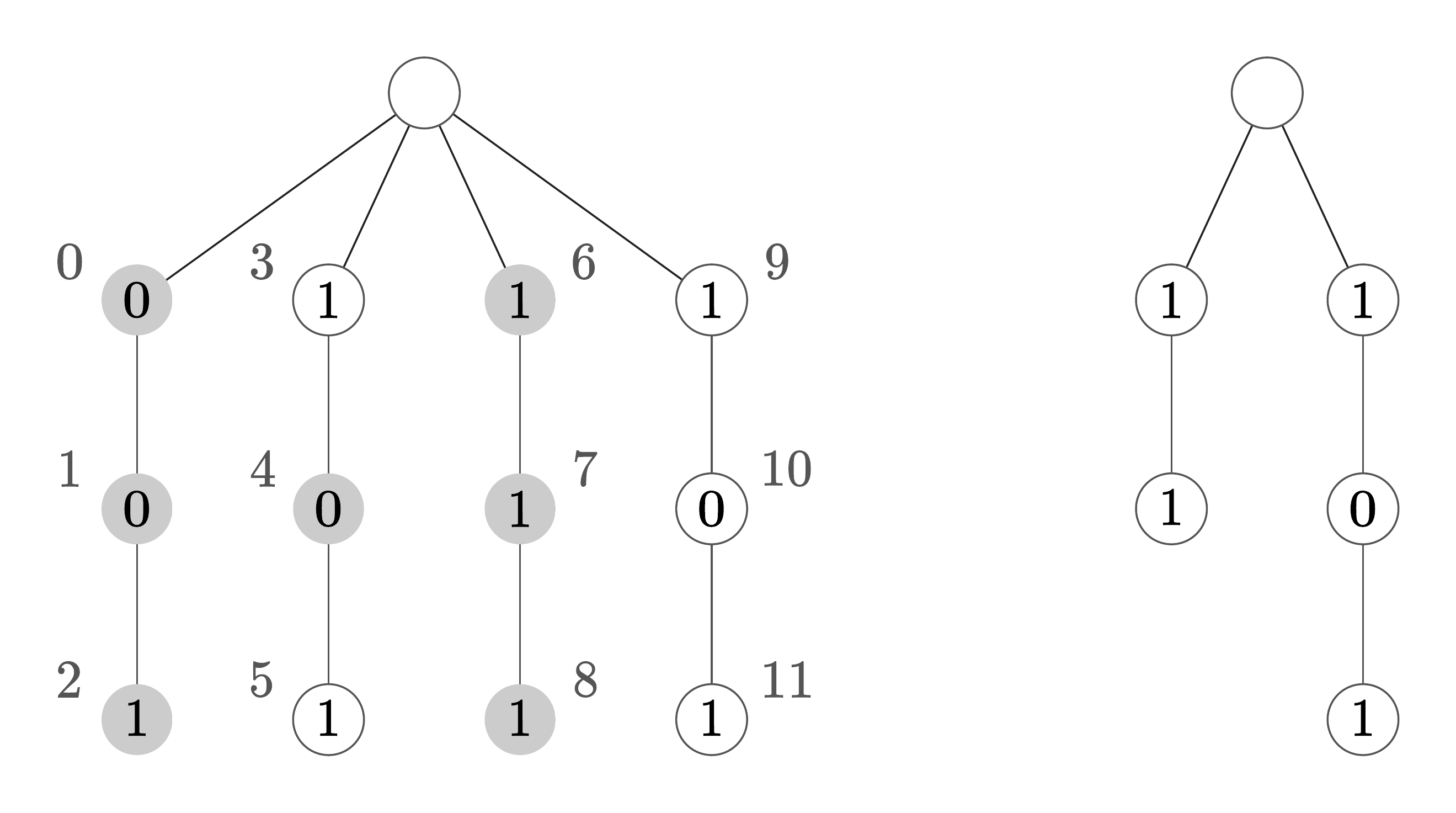}
	\caption{DFS indexing and example trace (in both deletion models) for a $(12,3)$-spider.}
	\figlab{spider}
\end{figure}

When $d$ is constant, the reconstruction problem on $(n,d)$-spiders can be reduced to string trace reconstruction 
(see \propref{RowReduction}). Hence, we will assume that $d$ is greater than a specific constant ($d \geq 20$ suffices). 
We begin by computing the expected value of the generating function for an $(n,d)$-spider which has gone through a deletion channel with parameter $q$. We denote this expected generating function by $A(w)$, where $w\in \mathbb{C}$.

\begin{lemma}\lemlab{GeneratingFunction}
Let $a=\left\{a_{i}\right\}_{i=0}^{n-1}$ be the labels of an $(n,d)$-spider with labels $a_i \in \mathbb{R}$ and let $b = \left\{ b_{j} \right\}_{j=0}^{n-1}$ be the labels of its trace from the deletion channel with deletion probability $q$. Then 
\begin{align*}
A(w):=\mathbb{E} \left (  \sum\limits_{j =0}^{n-1} b_{j} w^{j}\right)  &= (1-q)\sum\limits_{\ell=0}^{n-1}  a_{\ell} (q+(1-q)w)^{\ell \Mod d}~(q^d+(1-q^d)w^d)^{ \lfloor  \frac{\ell}{d}\rfloor },
\end{align*}
where the expectation is over the random labels $b$.
\end{lemma}
While $A(w)$ is written as only a function of $w$, it implicitly depends on the labels $a$ of the original spider. 
The proof of \lemref{GeneratingFunction} is in \secref{spider_proofs}, as it follows from a standard manipulation of equations. 
We use this generating function to distinguish between two candidate $(n,d)$-spiders $X^{1}$ and $X^{2}$, 
which have labels $a^{1} = \{ a^{1}_{j} \}_{j=0}^{n-1}$ and $a^{2} = \{ a^{2}_{j} \}_{j=0}^{n-1}$ 
which are different (that is, there exists $j \in \{ 0, 1, \ldots, n-1\}$ such that $a^{1}_{j} \neq a^{2}_{j}$). 
Let $Y^{1}$ and $Y^{2}$ denote random traces with labels 
$b^{1} = \{ b^{1}_{j} \}_{j=0}^{n-1}$ and $b^{2} = \{ b^{2}_{j} \}_{j=0}^{n-1}$
that arise from passing $X^{1}$ and $X^{2}$ through the deletion channel with deletion probability $q$.

Define $a:= a^{1} - a^{2}$ and 
let $A \left( w \right)$ be the expected generating function with input $a$. 
From \lemref{GeneratingFunction} we have that 
\begin{equation}\eqlab{genfn0}
 \sum_{j=0}^{n-1} \left( \mathbb{E} \left[ b^{1}_{j} \right] - \mathbb{E} \left[ b^{2}_{j} \right] \right) w^{j} 
 = A(w).
\end{equation}
Let
$\ell^{*} := \argmin_{\ell \geq 0} \left\{ a_{\ell} \neq 0 \right\}$
(note that $\ell^{*} \leq n-1$ by construction)
and define 
\[
 \widetilde{A}(w) := (1-q) \sum_{\ell = \ell^{*}}^{n-1} a_{\ell} (q+(1-q)w)^{\ell \Mod d}~(q^d+(1-q^d)w^d)^{ \lfloor  \frac{\ell}{d}\rfloor - \lfloor  \frac{\ell^*}{d}\rfloor}.
\]
Observe that 
$A(w) =(q^d + (1-q^d)w^d)^{\lfloor \frac{\ell^*}{d} \rfloor} \cdot \widetilde{A}(w)$; 
accordingly, we call $\widetilde{A}(w)$ the factored generating function. 
Taking absolute values in \Eqref{genfn0} we obtain that 
\begin{equation}\eqlab{genfunc}
\sum\limits_{j=0}^{n-1}   \left| \mathbb{E} \left[ b^{1}_{j} \right] - \mathbb{E} \left[ b^{2}_{j} \right] \right| \left| w \right|^j 
\geq |A(w)|
=  (1-q) \left| (1-q^d)w^d+q^d \right|^{\lfloor \frac{\ell^{*}}{d}\rfloor} \left| \widetilde{A}(w) \right|. 
\end{equation} 
Ultimately, we aim to bound from below 
$\max_{j} \left| \mathbb{E} \left[ b^{1}_{j} \right] - \mathbb{E} \left[ b^{2}_{j} \right] \right|$ 
by choosing $w \in \mathbb{C}$ appropriately. 
To do this, we balance $|w|$ on the left hand side of \Eqref{genfunc} 
with $|(1-q^d)w^d+q^d|$ on the right; 
it would be best if $|w|$ were small while $|(1-q^d)w^d+q^d|$ were large, 
and so a compromise is to let $w$ vary along an arc of the unit disc $\mathbb{D}$. 
In particular, let 
$\gamma_L := \{e^{i \theta} : - \pi/L \leq \theta \leq \pi / L\}$, 
where we assume that $L\geq20$, a choice which will become clear later. 
The following lemma bounds $\left|(1-q^d)w^d+q^d\right|$ from below while $w \in \gamma_L$. 
The proof is a standard calculation, deferred to \secref{spider_proofs}. 
\begin{lemma}\lemlab{BoundVariable}
For $w \in \gamma_L$ 
we have that 
\begin{align*} 
\left| (1-q^d)w^d+q^d \right| \geq \exp \left(-2 \pi^2 \cdot  q^d(1-q^d) d^2/L^2 \right).
\end{align*}
\end{lemma}
Additionally, we bound $\sup_{\gamma_L}|\widetilde{A}(w)|$ from below using \lemref{RidLogFactor}, 
whose proof is in \secref{LemmaContour}. 
\begin{lemma}\lemlab{RidLogFactor}
Let $0<q<0.7$ be a constant. There exists $\zeta \in \gamma_{L}$, as well as a constant $C>0$ depending only on $q$, such that 
$|\widetilde{A}(\zeta)| \geq \exp \left(-C \cdot d L \right)$. 
\end{lemma}
We are now ready to prove \thmref{MainSpider}. 
\begin{proof}[Proof of \thmref{MainSpider}]

Let $\zeta \in \gamma_L$ be the point guaranteed by \lemref{RidLogFactor}. 
Substituting $\zeta$ into \Eqref{genfunc}, we use \lemref{RidLogFactor} 
and the fact that $\left| \zeta \right| = 1$ to see that
\begin{align*}
	\sum\limits_{j=0}^{n-1}   \left| \mathbb{E} \left[ b^{1}_{j} \right] - \mathbb{E} \left[ b^{2}_{j} \right] \right|  
	\geq |A(\zeta)| 
	&=  (1-q) \left| (1-q^d)\zeta^d+q^d \right|^{\lfloor \frac{\ell^{*}}{d}\rfloor} \left| \widetilde{A}(\zeta) \right| \\
	&\geq  (1-q) \left| (1-q^d)\zeta^d+q^d \right|^{\lfloor \frac{\ell^{*}}{d}\rfloor} \exp \left(-C \cdot d  L \right),
\end{align*} 
for a constant $C > 0$ depending only on $q$. 
Using the bound $\ell^{*} < n$, as well as \lemref{BoundVariable} (where we drop the factor of $1-q^{d}$ in the exponent), we have that 
\[
 \sum\limits_{j=0}^{n-1}   \left| \mathbb{E} \left[ b^{1}_{j} \right] - \mathbb{E} \left[ b^{2}_{j} \right] \right| 
 \geq (1-q) \exp \left(-2 \pi^2 \cdot q^d nd/ L^2 \right)  \exp \left(-C \cdot d  L \right).
\]
Setting 
$L =\max\{ \left( 4 \pi^{2} n q^{d} / C \right)^{1/3}, 20\}$
and plugging into the display above, we find that there exists an index $j$ such that 
\begin{align}\eqlab{eq:1}
\left| \mathbb{E} \left[ b^{1}_{j} \right] - \mathbb{E} \left[ b^{2}_{j} \right] \right| 
\geq \frac{1}{n} \exp \left(- C' \cdot d (n q^d)^{1/3} \right)
\end{align} 
for some constant $C' > 0$ depending only on $q$. 
Therefore, we have shown that there is some index $j = j \left( X^{1}, X^{2} \right)$ where we expect the traces corresponding to $X^{1}$ and $X^{2}$ to differ significantly.

Suppose spider $X^1$ goes through the deletion channel and 
we observe $T$ samples, $S^1, \ldots, S^T$
where sample $S^t$ has labels $\{u^t_j\}_{j=0}^{n-1}$.
Let $\eta$ denote the right hand side of \Eqref{eq:1}.
We say that a spider $X^2$ is a \textit{better match} than $X^1$ for traces $\{S^t\}_{t \in [T]}$
if at the index $j=j(X^1,X^2)$, $X^2$ looks closer to the traces than $X^1$; that is, if 
\begin{align*}
\left | \frac{1}{T} \sum\limits_{t=1}^T u_j^t - \mathbb{E}\left[b^2_j\right] \right | \leq  \left | \frac{1}{T} \sum\limits_{t=1}^T u_j^t - \mathbb{E}\left[b^1_j\right]\right |.
\end{align*}
As before, the expectation is over the random labels $b^1$ and $b^2$.
A Chernoff bound implies that if the traces $\{S^t\}_{t \in [T]}$ came from spider $X^1$, then the probability that $X^2$ is a better match than $X^1$ is at most $ \exp (-T \eta^2/2)$.  
Repeating this for all pairs of binary labeled $(n,d)$-spiders, the algorithm outputs $X^*$, the $(n,d)$-spider which is a better match than all others 
(the best match), if such a spider exists. Otherwise, the algorithm outputs a random binary labeled $(n,d)$-spider. 

Lastly, we show that the algorithm correctly reconstructs an $(n,d)$-spider with high probability when $d \leq \log_{1/q} n$.
We bound from above the probability that the algorithm 
does not find that $X^1$ is the best match by a combination of a union bound and a Chernoff bound (as discussed above). 
The probabilities below are taken over the random traces $\{S^t\}_{t \in [T]}$:
\begin{align*}
\Pr [ X^* \neq X^1] &\leq \sum\limits_{X^{2} : X^2 \neq X^1} \Pr[X^2 \text{ is a better match than } X^1] \leq 2^n \cdot 
\exp \left (-T \eta^2/2 \right) \\
&= 2^n \exp \left(-\frac{T}{2n^2} \exp\left(-C \cdot d (n q^d)^{1/3} \right) \right)
\end{align*}
for a constant $C > 0$ depending only on $q$. 
This latter expression is at most $1/n$ if and only if 
\[
T \geq 2n^{2} \left( n \ln(2) + \ln(n) \right) \exp \left( C d \left( n q^{d} \right)^{1/3} \right).
\]
This holds if 
$T \geq \exp \left( c d \left( n q^{d} \right)^{1/3} \right)$ 
for a large enough constant $c$ depending only on~$q$. 
\end{proof}

\subsection{Proof of \lemref{RidLogFactor}}\seclab{LemmaContour}

We assume basic knowledge of subharmonic functions and harmonic measure.
 For background, we refer readers to any introductory complex analysis book (e.g.,~\cite{Ahlfors, HarmonicMeasure}). 
 For a more elementary (but slightly weaker) bound, see \lemref{SimpleWithLog} in \secref{spider_proofs}.

Let $\Omega \subset \mathbb{C}$ be a bounded, open region, and let $\partial \Omega$  denote its  boundary. 
The \emph{harmonic measure} of a subset $\gamma \subset \partial \Omega$ with respect to a point $w_0 \in \Omega$, denoted $\mu^{w_0}_{\Omega}(\gamma)$, is the probability that a Brownian motion starting at $w_0$ exits $\Omega$ through $\gamma$.
Let $f(w)$ denote an analytic function; 
we will choose $f = \widetilde{A}$, which is a polynomial and hence analytic.
Given $|f(w_{0})|$ at a point $w_{0} \in \Omega$ and a
condition on the growth of $|f|$ in $\Omega$,
we utilize harmonic measure to bound $|f|$ on $\partial \Omega$.
Specifically, we use that $\log|f|$ satisfies the \textit{sub-mean value property}: for all $w_{0} \in \Omega$ we have that 
\begin{equation}\Eqlab{w0bound}
\log|f(w_0)| \leq \int\limits_{\partial \Omega} \log|f(w)| d\mu^{w_0}_{ \Omega}(w).
\end{equation}

As in \Eqref{w0bound}, we will define a region of integration 
where the value of $\log \big|\widetilde{A}(w) \big|$ is controlled along the boundary,
and the boundary will contain $\gamma_L = \{e^{i \theta}: -\pi/L \leq \theta \leq \pi/L\}$. 
We need to separate this boundary into a few different pieces 
and use different techniques to upper bound $\log|\widetilde{A}(w)|$ on each curve. 
In fact, the methods of~\cite{HartungHP18} show a lower bound for $\sup_{\gamma_L}|f(w)|$ 
for an analytic function $f(w)$ satisfying the growth condition in~\lemref{UBA} (see below), 
by using~\Eqref{w0bound} and a particular choice of $w_{0}$. 
We show that $\widetilde{A}$ satisfies the growth condition specified in~\lemref{UBA}, 
then borrow techniques from~\cite{HartungHP18} to upper bound 
the right hand side of~\Eqref{w0bound}. 
However, we have to work more to find an appropriate point $w_0 \in \mathbb{D}$ in order to find a 
lower bound for the left hand side of~\Eqref{w0bound},
so that we can also show a lower bound for $\sup_{\gamma_L}|\widetilde{A}(w)|$.
We discuss this difficulty in~\remref{ForAllq}.

In what follows, open discs of radius $r$ centered at a point $z$ are denoted as $D_r(z)$.
The unit disc, $D_1(0)$, is an exception, denoted as $\mathbb{D}$. 
Recall that $L \geq 20$, and let
upper and lower case \emph{c}'s with tick marks denote constants depending only on $q$.

\begin{proof}[Proof of \lemref{RidLogFactor}]
First, we choose an appropriate point $w_0$ where 
we can lower bound $\big|\widetilde{A}(w_0)\big|$, as stated in \lemref{LBA0}.
\begin{lemma}\lemlab{LBA0}
Let $q < 0.7$ and for fixed integer $d>0$, 
let $w_0 := -q$ if $d$ is odd and $w_0 := q e^{i \cdot \pi(d-1)/d}$ if $d$ is even. 
Then there exists a constant $c>0$ depending only on $q$ such that $\big|\widetilde{A}(w_0)\big| \geq e^{-c\cdot d}$.
\end{lemma}
The calculation justifying the bound is standard and deferred to \secref{spider_proofs}, 
but the choices of $w_0$ 
are quite careful. 
The choices depend on the parity of $d$ 
so as to control the sign of $w_0^d$. 
With these choices 
we can relate $|q+(1-q)w_0|$ to $|q^d+(1-q^d)w_0^d|$, 
a motive which becomes clear in the proof of \lemref{LBA0}. 
Note that our inability to handle constant $q \geq 1/\sqrt{2}$
comes from  \lemref{LBA0}, as detailed in \remref{ForAllq}.

In the following we fix $w_{0}$ according to the specifications of \lemref{LBA0}. 
Next, we define a region $\Omega_L$ that contains $w_0$, 
and whose boundary we integrate over; see \figref{contour} for an illustration. 
For a translate $h_L \in \mathbb{R}^+$, let $D_L = D_{1/2}(1/2)+ h_L$, where $h_L$ is chosen so that $D_L \cap \partial \mathbb{D} = \gamma_L$. Observe that $L\geq20$ implies $h_L \leq 1/10$. 
We also define a rectangle $R \subset \mathbb{D}$ that has the following properties: 
$R$ contains $w_{0}$, 
$R$ has nonempty intersection with $D_{L}$, 
and $\partial R$ has bounded distance from $w_{0}$ and $\partial \mathbb{D}$. 
As we only consider $q$ bounded away from $1$ and $d \geq 20$, 
we may (and will) choose $R$ to be centered about the real axis, 
with height $1/5$ and with length extending from $-0.8$ to $1/2$. 
Our region of integration is then defined as 
$\Omega_{L} := \left( \mathbb{D} \cap D_{L} \right) \cup R$.

We partition the boundary of $\Omega_{L}$ into three parts by defining 
\begin{align*}
\gamma_L' := \left \{ w \in \partial \Omega_L \setminus \gamma_L : |w| > \frac12 + h_L \right  \}, 
\qquad 
\gamma_L'' :=  \left \{ w \in \partial \Omega_L : |w| \leq \frac12 + h_L \right \}.
\end{align*}
We thus have that 
$\partial \Omega_{L} = \gamma_{L} \cup \gamma_{L}' \cup \gamma_{L}''$; 
see \figref{contour}, where the different parts of the contour are colored differently.
\begin{figure}[t]
    \centering
    \includegraphics[angle=0, width=0.45 \textwidth]{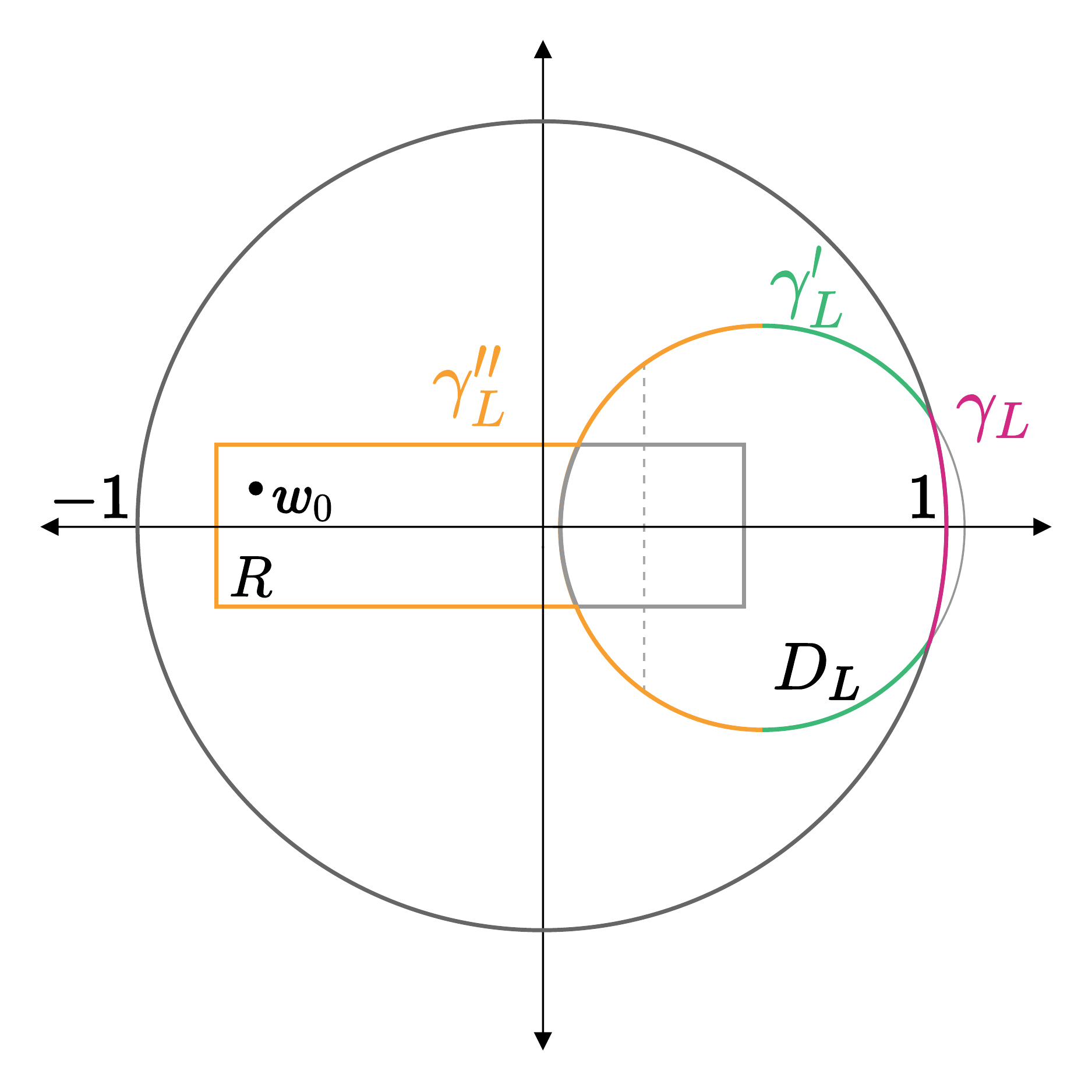} 
    \caption{The region $\Omega_L = \left( \mathbb{D} \cap D_{L} \right) \cup R$ with boundary $\partial \Omega_{L} = \gamma_{L} \cup \gamma_{L}' \cup \gamma_{L}''$.}
    \figlab{contour}
\end{figure}
Using the sub-mean value property, we see that
\begin{align*}
\log \left|\widetilde{A}(w_0) \right| &\leq \int_{\partial \Omega_L} \log \left|\widetilde{A}(w) \right| d \mu_{\Omega_L}^{w_0}(w)  \\ &=
\int_{\gamma_L} \log \left|\widetilde{A}(w) \right| d \mu_{\Omega_L}^{w_0}(w) + 
\int_{\gamma_L'} \log \left|\widetilde{A}(w) \right| d \mu_{\Omega_L}^{w_0}(w) + 
\int_{\gamma_L''} \log \left|\widetilde{A}(w) \right| d \mu_{\Omega_L}^{w_0}(w).
\end{align*}

Next, we upper bound each of the integrals.
Our upper bounds for $\gamma_L$ and $\gamma_L''$ are simple modulus length bounds. 
We start with the integral over $\gamma_{L}$. 
We know that the boundary $\partial \Omega_{L}$ has constant length, 
the curve $\gamma_{L}$ has length bounded by $c/L$ for some constant $c$, 
and $w_{0}$ is bounded away from $\gamma_{L}$. 
Therefore the probability that a Brownian motion starting at $w_{0}$ exits $\Omega_{L}$ through the arc $\gamma_{L}$ is at most $C/L$ for some constant $C$.
That is, $\mu_{\Omega_L}^{w_0}(\gamma_L) \leq C/L$, where 
we can choose $C$ to hold for all $w_0$ as $q$ and $d$ vary. 
Therefore we have that 
\begin{align}\eqlab{UB1}
\int_{\gamma_L} \log \left|\widetilde{A}(w) \right| d \mu_{\Omega_L}^{w_0}(w) 
\leq \mu_{\Omega_L}^{w_0}(\gamma_L) \cdot \log \sup\limits_{w \in \gamma_L} \left|\widetilde{A}(w) \right| 
\leq \frac{C}{L} \cdot \log \sup\limits_{w \in \gamma_L} \left|\widetilde{A}(w) \right|.
\end{align}

Somewhat more work is needed to show that
\begin{align}\eqlab{UB2}
    \int_{\gamma_L'} \log \left| \widetilde{A}(w) \right| d \mu_{\Omega_L}^{w_0}(w) \leq c' \quad \text{and} \quad
    \int_{\gamma_L''} \log \left| \widetilde{A}(w) \right| d \mu_{\Omega_L}^{w_0}(w) \leq c''.
\end{align}
To prove these bounds, we use \lemref{UBA}, a growth condition for the generating function.
\begin{lemma}\lemlab{UBA}
For all $w \in \mathbb{D}$ and all deletion probabilities $q \in (0,1)$, 
we have $\big|\widetilde{A}(w)\big| \leq \frac{1}{(1-q)(1-|w|)}$.
\end{lemma}
The proof of \lemref{UBA} is a triangle inequality calculation that we defer to \secref{spider_proofs}. 
We can directly apply \lemref{UBA} to bound the integral over $\gamma_{L}''$. 
By our choice of $R$ and the fact that $h_L \leq 1/10$, 
we have for all $w \in \gamma_{L}''$ that $\left| w \right| \leq 0.83$. 
Applying \lemref{UBA}, we see that 
$\big|\widetilde{A}(w)\big| \leq 10/(1-q)$ 
for all $w \in \gamma_{L}''$. 
Noting that harmonic measure is a probability measure and thus 
$\mu_{\Omega_L}^{w_0}(\gamma_{L}'') \leq 1$, 
we have that 
\[
 \int_{\gamma_{L}''} \log \left| \widetilde{A}(w) \right| d \mu_{\Omega_{L}}^{w_{0}}(w) 
 \leq \sup_{w \in \gamma_{L}''} \log \left| \widetilde{A}(w) \right| 
 \leq \log \frac{10}{1-q}.
\]

We turn now to the integral over $\gamma_{L}'$, 
where we cannot use modulus length bounds because as $|w|$ approaches 1, 
the factor $1/(1-|w|)$ becomes arbitrarily large. 
However, we can still use \lemref{UBA} to obtain that 
\begin{align*}
\int_{\gamma_L'} \log \left|\widetilde{A}(w)\right| d \mu_{\Omega_L}^{w_0}(w) 
\leq \log \left( \frac{1}{1-q} \right) + \int_{\gamma_L'} \log \left( \frac{1}{1-|w|} \right) d \mu_{\Omega_L}^{w_0}(w).
\end{align*} 
It remains to bound the integral on the right hand side of the display above. 
A Brownian motion in $\Omega_L$ starting at $w_0$ must hit the segment
$s_L =\{w \in \Omega_L : \mathrm{Re}(w) =\frac{1}{4}+h_L\}$ before it hits $\gamma_{L}'$, so 
\[
 \int_{\gamma_L'} \log \left( \frac{1}{1-|w|} \right) d \mu_{\Omega_L}^{w_0}(w) 
 \leq 
 \sup_{z \in s_{L}}
 \int_{\gamma_L'} \log \left( \frac{1}{1-|w|} \right) d \mu_{\Omega_L}^{z}(w).
\]

Note that $z$ and $\partial \Omega_L$ will 
not be arbitrarily close.
Considering Brownian motions starting at $z$, 
we will upper bound the probability that it exits $\Omega_L $ through $\gamma_L'$
by the probability that it exits $D_L $ through $\gamma_L'$. 
As $z$ is bounded away from $\partial \Omega_L$ and $\partial D_L$, the
measures $\mu^z_{D_L}$ and $\mu^{z}_{\Omega_L}$ are equivalent, 
meaning they have the same null sets.
Then by the Radon-Nikodym theorem, there exists a measurable function 
$f = d \mu^{z}_{\Omega_L} / d \mu^z_{D_L}$ such that for any measurable 
set $S$,
$\mu^{z}_{\Omega_L}(S)  = \int_S f(w) d \mu^z_{D_L}(w)$.
From the probabilistic definition of harmonic measure, 
observe that there exists a constant $c>0$ such that for any measurable $S \subset \gamma_L'$,
\begin{align*}
\mu_{\Omega_L}^{z}(S) 
\leq c \cdot \mu^z_{D_L}(S).
\end{align*}
Since $\mu^{z}_{\Omega_L} / \mu^z_{D_L}$ is bounded on $\gamma_L'$,
so is $f$ up to a set of measure 0
\footnote{If $f$ is unbounded on a set of positive measure, $A \subset \gamma_L'$,
then for all $C$, $f \geq C$ on $A$. Writing 
$ \mu^{z}_{\Omega_L}(A)  = \int_A f  d \mu^z_{D_L}  \geq C  \mu^z_{D_L}(A)$, 
we see that $\mu^{z}_{\Omega_L}(A) / \mu^z_{D_L}(A)$ is unbounded.}. 
The upper bound on $f$ almost everywhere is sufficient.
Then returning to our remaining integral over $\gamma_L'$, we obtain the upper bound
\begin{align*}
	\sup_{z \in s_{L}}
	\int_{\gamma_L'} \log \left( \frac{1}{1-|w|} \right) d \mu_{\Omega_L}^{z}(w)
	&=
	 \sup_{z \in s_{L}}\int_{\gamma_L ' } \log \left (\frac{1}{1-|w|} \right) f(w) d \mu_{D_L}^{z}(w) 
	 \\&\leq c \cdot  \sup_{z \in s_{L}} \int_{\gamma_L ' } \log \left (\frac{1}{1-|w|} \right) d \mu_{D_L}^z(w).
\end{align*}

We move to harmonic measure with respect to a disc, instead of $\Omega_L$, so that we can 
switch measures to integrate with respect to angles on the disc.
Specifically, 
we have an explicit form for the Radon-Nikodym derivative.
Letting $s$ denote the arc length measure and $D_r$ denote a disc of radius $r$ containing the point $z$,
$d \mu^z_{D_r} / d s$ at a point $\zeta \in \partial D_r$ is
 the Poisson kernel $P(z,\zeta)= \frac{r^2-|z|^2}{2 \pi \cdot r|z-\zeta|^2}$.
Note that $P(z,\zeta)$ is uniformly bounded above by a constant for all $\zeta \in \gamma_L'$,
as all $z \in s_L$ are bounded away from $\gamma_{L }'$, 
since $w \in \gamma_{L }'$ has $\mathrm{Re}(w) \geq 1/2$. 
This is a useful observation, as now we can integrate
with respect to the angle on $D_L$ between $w=(1/2+h_L) + e^{i \theta}/2 \in \gamma_L '$ 
and $x=e^{i \pi/L}=1/2+h_L +  e^{i \theta_0}/2$, while only gaining a constant 
factor depending on $q$ in the upper bound. 
Recall $x$ is the intersection point of $\partial D_L$ and $\partial \mathbb{D}$ in the first quadrant.
We use the following lemma to obtain a new bound.
There are several proofs for it using only elementary geometry,
and we include one in \secref{spider_proofs}.
\begin{lemma}\lemlab{BoundInt}
For translate $0<h_L \leq 1/10$ satisfying $  e^{i \pi/L} \in D_L \cap \mathbb{D}$,
consider the points $w=(1/2+h_L) + e^{i \theta}/2$ with $0 \leq \theta \leq \pi/2$,
and $x = e^{i \pi / L} = 1/2+h_L +  e^{i \theta_0}/2$ with $0 \leq \theta_0 \leq \theta$.
Then $1-|w| \geq \frac{1}{64}(\theta - \theta_0)^4$.
\end{lemma}
As our region of integration is symmetric,
it suffices to only show the inequality for $\gamma_L'$ in the first quadrant.
From the boundedness of the Poisson kernel and \lemref{BoundInt},

\begin{align*}
c \cdot \sup_{z \in s_L} \int_{\gamma_L ' } \log \left (\frac{1}{1-|w|} \right) d \mu_{D_L}^z(w)
 &\leq C \cdot  \int_{\theta_0}^{\pi/2} \left (\log64+ \log \left (\frac{1}{\theta-\theta_0} \right)^4 \right) d \theta \\
 &\leq c'' \cdot \int\limits_{\theta_0}^{\pi/2} \log \left (\frac{1}{\theta-\theta_0} \right) d\theta \leq  c'.
\end{align*}

Having proven the bounds on the integrals in \Eqref{UB1} and \Eqref{UB2}, 
we are now ready to conclude the proof of \lemref{RidLogFactor}. Combining these bounds with \lemref{LBA0} and the sub-mean value property inequality, 
we see that
\[
 - c d \leq \frac{C}{L} \log \sup\limits_{w \in \gamma_L} \left|\widetilde{A}(w) \right| + c' + c'',
\]
where all constants are positive and depend only on $q$. 
Rearranging, we now have the lower bound
$\sup_{w \in \gamma_L} \big|\widetilde{A}(w) \big| 
\geq \exp \left( - C' d L \right)$ 
for some constant $C'$ depending only on $q$. 
As $\gamma_L$ is a closed arc,
there exists $\zeta \in \gamma_L$ such that $|\widetilde{A}(\zeta)|= \sup_{w \in \gamma_L} \big|\widetilde{A}(w)\big|$, as desired. 
\end{proof}

\subsection{Bounds for Spiders from String Trace Reconstruction}\seclab{spiders-from-strings}

String reconstruction methods can be used as a black box for spiders. For depth $d \geq \log_{1/q} n$, 
this achieves the best known bound. However, for smaller depths, our algorithm is more efficient.

\subsubsection*{Large depth $(n,d)$-spiders}\seclab{LargeSpider}

\begin{proof}[Proof of \propref{SpiderCor}]
With probability $(1-q^d)^{n/d}$, a trace contains at least one non-root node from each of the $n/d$ paths in the spider. 		
When all paths are present, we can match paths of the trace to paths of the original spider and learn paths separately. 
Using only such traces, 
we are faced, 
for each path, 
with a string trace reconstruction problem with censoring (see \apndref{TR_censoring}), 
where the string length is $d$, the deletion probability is~$q$, and the censoring probability is 
$\gamma = 1 - \left( 1 - q^{d} \right)^{n/d - 1}$. 
\lemref{censored_TR} in \apndref{TR_censoring} (with $\eps = 1/2$) tells us that
\[
 T_{\gamma}^{\textbf{cens}} \left( d, \frac{1}{n^{2}} \right) 
 \leq 
 \frac{3/2}{\left( 1 - q^{d} \right)^{n/d}} \cdot T \left( d, \frac{1}{2n^{2}} \right) 
 \leq 
 2 \cdot T \left( d, \frac{1}{2n^{2}} \right),
\]
where the second inequality holds (for all $n$ large enough) because 
$\left( 1 - q^{d} \right)^{n/d} \to 1$ 
as $n \to \infty$ 
when $d \geq \log_{1/q} n$. 
That is, if we observe 
$2 \cdot T \left( d, \frac{1}{2n^{2}} \right)$ 
traces of the spider, 
then the bits along each specific path can be reconstructed with error probability at most $1/n^{2}$. 
Hence, by a union bound the bits along all paths can be reconstructed with error probability at most 
$(n/d) \times (1/n^{2}) \leq 1/n$. 
\end{proof}

We can extend \propref{SpiderCor} to the following result. We omit the proof, which follows the same outline and ideas as the proof of \propref{SpiderCor}.

\begin{proposition}\proplab{ImproveCor}
For $n$ large enough, $\a \geq 0$, and $d \geq \log_{1/q} n - \log_{1/q} \left( \log^{1+\a} n \right)$,
an $(n,d)$-spider 
can be reconstructed with 
$\exp \left(C \left( \log^{\a} n \right) \right) \cdot T \left(d, \frac{1}{2n^{2}} \right)$ traces with high probability,
where $C$ is a constant depending only on $q$. 
\end{proposition}

\subsubsection*{Small depth $(n,d)$-spiders}\seclab{SmallSpider}

When $d = c\log_{1/q} n$ with constant $0<c<1$,
the same reconstruction strategy still applies, but it does worse than our mean-based algorithm (which results in \thmref{MainSpider}).
In this regime of $d$, to ensure that with high probability we see even a single trace containing all $n/d$ paths, 
we must take $\exp \left( \Omega \left( n^{1-c} / \log n \right) \right)$ traces. 
It suffices to take 
$\exp \left( O \left(n^{1-c} / \log n \right) \right)\cdot T\left( d, 1/n^{2} \right)=\exp(\tilde{O}(n^{1-c}))$ traces 
to ensure that enough traces contain all $n/d$ paths. 
However, our mean-based algorithm resulting in \thmref{MainSpider} does better than this, 
requiring only $\exp( \widetilde{O} ( n^{(1-c)/3} ))$ traces to reconstruct.

We observe that the previous results on string trace reconstruction can also be used to derive the following proposition (in addition to \propref{SpiderCor} and \propref{ImproveCor}). The consequences are twofold: (i) when $d = O(1)$, then the trace complexity of spiders is asymptotically the same as strings, and (ii) our result in \thmref{MainSpider} offers an improvement when $d = \omega(1)$.

\begin{proposition}\proplab{RowReduction} For $d < \log_{1/q} n$, we can reconstruct an $(n,d)$-spider with high probability by using at most  $\exp \left (\left (\frac{C'  n }{d (1-q)^{2d}} \right )^{1/3} \right )$
	traces, for $C' >0$ depending on $q$.
\end{proposition}

We sketch the proof. A path in the spider of depth $d$ retains all of its nodes with probability $(1-q)^d$. Equivalently, some node is deleted with probability $q' = 1 - (1-q)^d$. For any trace, consider the modified channel that deletes any path entirely if it is missing at least one node. With this modification, every row of the spider behaves as if it were a string on $n/d$ bits in a channel with deletion probability $q'$. Opening up the proof of \thmref{StringMain}, for non-constant deletion probability $q' = 1 - (1-q)^d$, then gives the proposition.

\subsection{Additional Proofs and Remarks for Reconstructing Spiders}\seclab{spider_proofs}

\begin{remark}[Remark for \thmref{MainSpider}]\remlab{ForAllq}
	In the proof of \lemref{RidLogFactor}, which is needed to prove \thmref{MainSpider}, we are unable to handle general generating functions with deletion probability $1/\sqrt{2} < q < 1$. We require some anchor point, $w_0$, for which we can lower bound $|\widetilde{A}(w_0)|$ and a simple curve surrounding $w_0$ for which we can upper bound $|\widetilde{A}(w)|$ along that curve. For any fixed $|w|>1$, for $d=\Omega(1)$ we see that $|(1-q^d)w^d+q^d| >1$ for sufficiently large $n$. This results in terms on the order of $c^{n/d}$ in our generating function, for constant $c>1$. So our anchor point cannot lie outside of $\mathbb{D}$, and more specifically the surrounding curve cannot leave $\mathbb{D}$.
	
	Inside the unit disc, upper bounds on $|\widetilde{A}(w)|$ have a nice form due to \lemref{UBA}. It seems for any fixed point in $w_0 \in \mathbb{D}$, there is a factored generating function $\widetilde{A}(w)$ which is small at $w_0$, $|A(w_0)| =\Theta(q^d)$. However it is not clear whether for every factored generating function $\widetilde{A}(w)$ there is some $w_0 \in \mathbb{D}$, not tending to the boundary, such that $|\widetilde{A}(w_0)| >c$ for come constant $c$ depending only on $q$. Such arguments are common in complex analysis for families of analytic functions which are sequentially compact, but our family of generating functions does not satisfy this property.  
	\end{remark}

\begin{proof}[Proof of \lemref{GeneratingFunction}] 
We index the non-root nodes of the spider according to the DFS ordering described in \secref{spider-prelims}. 
We can uniquely write any $j \in \left\{ 0, 1, \ldots, n - 1 \right\}$ as 
$j = d \cdot s_{j} + r_{j}$ with $s_{j} \in \left\{ 0, 1, \ldots, n/d - 1 \right\}$ corresponding to a particular path of the spider and $r_{j} \in \left\{ 0, 1, \ldots, d - 1 \right\}$ describing where along this path node $j$ is. 
Consider two nodes, $j = d \cdot s_j + r_j$ and $\ell = d \cdot s_{\ell} + r_{\ell}$, with $j \geq \ell$. 
After passing $a$ through the deletion channel to get the trace $b$, 
$b_{\ell}$ comes from $a_{j}$ if and only if 
$a_{j}$ is retained, 
exactly $r_{\ell}$ of the first $r_j$ nodes in the path of $j$ are retained, 
and exactly $s_{\ell}$ of the first $s_j$ paths are retained. 
This leads to the following generating function:
\begin{align*}
    \mathbb{E} \left[\sum\limits_{\ell=0}^{n-1} b_{\ell} w^{\ell}\right]  &= (1-q) \sum\limits_{\ell = 0}^{n-1} w^{\ell}\sum\limits_{j=\ell}^{n-1} a_j \binom{r_j}{r_{\ell}} (1-q)^{r_{\ell}}q^{r_j-r_{\ell}} \binom{s_j}{s_\ell} q^{d(s_j-s_{\ell})}(1-q^d)^{s_{\ell}} \mathbf{1}_{\left\{ r_{\ell} \leq r_{j} \right\}}\\
&= (1-q) \sum\limits_{j = 0}^{n-1} a_j\sum\limits_{\ell=0}^{j}  \binom{r_j}{r_{\ell}} (1-q)^{r_{\ell}}q^{r_j-r_{\ell}} \binom{s_j}{s_{\ell}} q^{d(s_j-s_{\ell})}(1-q^d)^{s_{\ell}}w^\ell \mathbf{1}_{\left\{ r_{\ell} \leq r_{j} \right\}}\\
&= (1-q)\sum\limits_{s_j=0}^{n/d-1} \sum\limits_{r_j=0}^{d-1} a_{s_j d + r_j} \sum\limits_{s_{\ell}=0}^{s_{j}} \sum\limits_{r_{\ell}=0}^{r_j}
\binom{r_j}{r_{\ell}}(1-q)^{r_{\ell}}q^{r_j-r_{\ell}} \binom{s_j}{s_{\ell}}q^{d(s_j-s_{\ell})}(1-q^d)^{s_{\ell}}w^{s_{\ell}  d + r_{\ell}},
\end{align*}
where we used linearity of expectation and interchanged the order of summation. 
Observing that the sums are binomial expansions we have that 
\begin{align*}
    \mathbb{E} \left ( \sum\limits_{\ell=0}^{n-1} b_{\ell} w^{\ell} \right )
&= (1-q) \sum\limits_{s_j=0}^{n/d-1} \sum\limits_{r_j=0}^{d-1} a_{ds_j + r_j} (q+(1-q)w)^{r_j}(q^d+(1-q^d)w^d)^{s_j} \\
&= (1-q) \sum\limits_{j=0}^{n-1} a_{j} (q+(1-q)w)^{j \Mod d}~(q^d+(1-q^d)w^d)^{ \lfloor  \frac{j}{d}\rfloor},
\end{align*}
which proves the claim.
\end{proof}

\begin{proof}[Proof of \lemref{BoundVariable}]
Writing $w = \cos(\theta) + i \sin(\theta)$, we see that 
\begin{align*}
|(1-q^d)& w^d+q^d |^2\\ 
&= \left|(1-q^d)(\cos(\theta) + i \sin(\theta))^d+q^d \right|^2 = \left|(1-q^d)(\cos(d \theta) + i \sin(d \theta))+q^d \right|^2 \\
&= ((1-q^d)\cos(d \theta) + q^d)^2 + ((1-q^d)\sin(d \theta))^2 \\
&= (1-q^d)^2\cos^2(d \theta) + 2 q^d (1-q^d)\cos(d \theta) + q^{2d} + (1-q^d)^2\sin^2(d \theta) \\ 
&=(1-q^d)^2 + 2 q^d (1-q^d)\cos(d \theta) + q^{2d} = 1 - 2q^d + 2q^{2d} + 2q^d(1-q^d)\cos(d \theta)\\ 
&= 1-2q^d(1-q^d)(1-\cos(d \theta)).  
\end{align*}
Now using the fact that $1 - \cos(y) \leq y^{2} / 2$, 
as well as the inequality $1 - y \geq \exp(-4y)$ which holds for all $y \in [0,0.9]$ 
(in our case indeed $q^{d} (1 - q^{d}) d^{2} \theta^{2} \in [0,0.9]$ for all possible parameter values), 
we obtain that 
\[
 \left| (1-q^d)w^d+q^d \right|^2 
 = 1 - 2q^{d}(1-q^d)(1-\cos(d \theta))
 \geq\exp( - 4q^d (1-q^d) d^2\theta^2 ).
\]
Taking a square root of the last line shows $|(1-q^d)w^d+q^d| \geq \exp(-2 q^d(1-q^d) d^2\theta^2)$. 
Finally, the assumption that $w \in \gamma_{L}$ implies that $\theta^{2} \leq \pi^{2} / L^{2}$ and the claim follows. 
\end{proof}

\begin{proof}[Proof of \lemref{LBA0}]
We will consider the case of even and odd $d$ separately, starting with the cleaner case of when $d$ is odd. Recall that we assume that $d \geq 20$ and we choose $w_0 = -q$ when $d$ is odd and $w_0 = q e^{i\cdot\pi(d-1)/d}$ when $d$ is even. 
Let $\alpha :=|q+(1-q)w_0|$ and $\beta := |q^d+(1-q^d)w_0^d|$. 

When $d$ is odd, $\alpha = q^2$, and also $w_0^d =-q^d$, hence $\beta = q^{2d}$. 
When $d$ is even, $w_0$ is still chosen so that $w_0^d = -q^d$ and thus $\beta = q^{2d}$, 
as in the case when $d$ is odd. 
It is clear geometrically that $\alpha \geq q^2$, but we include the calculation as well:
\begin{align*}
\alpha &= \sqrt{(\mathrm{Re}(q+(1-q)w_0))^2 + (\mathrm{Im}(q+(1-q)w_0))^2}\\
&= \sqrt{(q+(1-q)q\cos(\pi (d-1)/d))^2 + ((1-q)q\sin(\pi (d-1)/d)))^2} \\
& = \sqrt{2q^2(1-q)(1+\cos(\pi (d-1)/d)) + q^4} \geq \sqrt{0 + q^4} =q^2,
\end{align*}
where we use that $\cos(\pi (d-1)/d) \geq -1$.
By our choice of $\beta$, we also see that $\alpha^d \geq q^{2d}=\beta$.

We are now ready to prove a lower bound on $\big|\widetilde{A}(w_0)\big|$ which holds for both $d$ even and $d$ odd. 
First, recalling the definition of $\widetilde{A}$ and the fact that $w_{0}^{d} = -q^{d}$, we have that 
\begin{align*}
 \widetilde{A} (w_{0}) 
 &= (1-q) \sum_{\ell = \ell^{*}}^{n-1} a_{\ell} \left( q + (1-q) w_{0} \right)^{\ell \Mod d} \left( q^{d} + (1-q^{d}) w_{0}^{d} \right)^{\left\lfloor \frac{\ell}{d} \right\rfloor - \left\lfloor \frac{\ell^{*}}{d} \right\rfloor} \\
 &= (1-q) \sum_{\ell = \ell^{*}}^{n-1} a_{\ell} \left( q + (1-q) w_{0} \right)^{\ell \Mod d} q^{2d \left( \left\lfloor \frac{\ell}{d} \right\rfloor - \left\lfloor \frac{\ell^{*}}{d} \right\rfloor \right)}.
\end{align*}
Now recall that $\left| a_{\ell^{*}} \right| = 1$ 
and thus the first term in the sum above (corresponding to index $\ell = \ell^{*}$) is, 
in absolute value, equal to 
$\alpha^{\ell^{*} \Mod d}$. 
Since $a_{\ell} \in \left\{ - 1, 0, 1 \right\}$ for all $\ell$, 
the rest of the sum above (adding terms corresponding to indices $\ell^{*}+1 \leq \ell \leq n-1$) is, 
in absolute value, at most 
\begin{align*}
 \sum_{\ell = \ell^{*}+1}^{n-1} \alpha^{\ell \Mod d} q^{2d \left( \left\lfloor \frac{\ell}{d} \right\rfloor - \left\lfloor \frac{\ell^{*}}{d} \right\rfloor \right)} 
 &\leq 
 \sum_{\ell = \ell^{*}+1}^{n-1} \alpha^{\ell \Mod d} \alpha^{d \left( \left\lfloor \frac{\ell}{d} \right\rfloor - \left\lfloor \frac{\ell^{*}}{d} \right\rfloor \right)} \\
 &= \alpha^{-d \left\lfloor \frac{\ell^{*}}{d} \right\rfloor} \sum_{\ell = \ell^{*} + 1}^{n-1} \alpha^{\ell} 
 \leq \alpha^{-d \left\lfloor \frac{\ell^{*}}{d} \right\rfloor} \frac{\alpha}{1-\alpha} \alpha^{\ell^{*}} 
 = \frac{\alpha}{1-\alpha} \alpha^{\ell^{*} \Mod d}.
\end{align*}
Putting these two bounds together we obtain that 
\[
 \left| \widetilde{A} (w_{0}) \right| 
 \geq (1-q) \left( \alpha^{\ell^{*} \Mod d} - \frac{\alpha}{1-\alpha} \alpha^{\ell^{*} \Mod d} \right) 
 = (1-q) \frac{1-2\alpha}{1-\alpha} \alpha^{\ell^{*} \Mod d}.
\]
When $\alpha$ is less than and bounded away from $1/2$, then this bound is at least a positive constant times 
$\alpha^{\ell^{*} \Mod d}$. 
Here, our choice of $d$ and $q$ becomes clear, as when $d \geq 20$ and $q \leq .7$ then $\alpha \leq .49$. 
Since $\ell^{*} \Mod d < d$, we have that 
$\alpha^{\ell^{*} \Mod d} \geq q^{2d}$ 
and the claim follows.
\end{proof}

\begin{proof}[Proof of \lemref{UBA}]
First, we show that $q^d + (1-q^d)|w|^d \leq (q+(1-q)|w|)^d$ for all $w \in \mathbb{D}$ and $q \in (0,1)$. 
This is because
\begin{align*}
    (q+(1-q)|w|)^d &= \sum\limits_{j=0}^d \binom{d}{j} q^j ((1-q)|w|)^{d-j} = q^d +\sum\limits_{j=0}^{d-1} \binom{d}{j} q^j ((1-q)|w|)^{d-j} \\
     &\geq q^d + |w|^d\sum\limits_{j=0}^{d-1} \binom{d}{j} q^j (1-q)^{d-j} = q^d + |w|^d(1-q^d),
\end{align*}
where we used the inequality 
$\left| w \right|^{-j} \geq 1$ 
which holds when $\left| w \right| \leq 1$ and $j \geq 0$. 
Combining this inequality with the triangle inequality, 
we can show the desired upper bound for $\big|\widetilde{A}(w)\big|$: 
\begin{align*}
 \left| \widetilde{A}(w) \right| 
 &\leq \sum_{\ell = \ell^{*}}^{n-1} \left| a_{\ell} \right| \left| q + (1-q) w \right|^{\ell \Mod d} \left| q^{d} + (1-q^{d}) w^{d} \right|^{\left\lfloor \frac{\ell}{d} \right\rfloor - \left\lfloor \frac{\ell^{*}}{d} \right\rfloor} \\
 &\leq \sum_{\ell = \ell^{*}}^{n-1} \left( q + (1-q) \left| w \right| \right)^{\ell \Mod d} \left( q^{d} + (1-q^{d}) \left|w\right|^{d} \right)^{\left\lfloor \frac{\ell}{d} \right\rfloor - \left\lfloor \frac{\ell^{*}}{d} \right\rfloor} \\
 &\leq \sum_{\ell = \ell^{*}}^{n-1} \left( q + (1-q) \left| w \right| \right)^{\ell \Mod d + d \left( \left\lfloor \frac{\ell}{d} \right\rfloor - \left\lfloor \frac{\ell^{*}}{d} \right\rfloor \right)}\\ 
 &= \left( q + (1-q) \left| w \right| \right)^{ - d \left\lfloor \frac{\ell^{*}}{d} \right\rfloor} \sum_{\ell = \ell^{*}}^{n-1} \left( q + (1-q) \left| w \right| \right)^{\ell} \\
 &\leq \left( q + (1-q) \left| w \right| \right)^{ - d \left\lfloor \frac{\ell^{*}}{d} \right\rfloor} \frac{\left( q + (1-q) \left| w \right| \right)^{\ell^{*}}}{1 - \left( q + (1-q) \left| w \right| \right)} 
 \\ & \leq \frac{1}{1 - \left( q + (1-q) \left| w \right| \right)} 
 = \frac{1}{(1-q)(1-|w|)},
\end{align*}
where we used that $q + (1-q)|w| < 1$ and $\ell^{*} - d \lfloor \ell^{*}/d \rfloor \geq 0$. 
Note that the same upper bound holds for $|A(w)|$ as well, since $|A(w)| \leq \big|\widetilde{A}(w)\big|$ for all $w \in \mathbb{D}$.
\end{proof}

\begin{proof}[Proof of \lemref{BoundInt}]

For the setup of this proof, it may be helpful to refer to \figref{contour}.
$x = e^{i \pi/L}$ lies on the disc $D_L$, and so
we can write also write $x$ as $x = 1/2 + h_L + e^{i \theta_0}/2 $.
On the other hand, letting $\varepsilon = 1-|w|$ we see that 
$ |1/2 + h_L + (1/2 + \varepsilon)e^{i \theta}|=1$. 
We will assume that $\theta_0$ and $\theta$ are in the first quadrant, 
and we could obtain the same result when they are both in the fourth quadrant by symmetry.
Set the two moduli equal to each other:
\begin{equation}\eqlab{moduli}
 |1/2 + h_L + 1/2 e^{i \theta_0}|^2 = |1/2 + h_L + (1/2 + \varepsilon)e^{i \theta}|^2
\end{equation}
Computing the left hand side of \Eqref{moduli}:
\begin{align*}
|1/2 + h_L + 1/2 e^{i \theta_0}|^2 & = \left ( 1/2 + h_L + 1/2 \cos(\theta_0) \right)^2 + (1/2\sin(\theta_0))^2\\ 
&= (1/2+h_L)^2+1/4 \cos^2(\theta_0) + 1/4 \sin^2 \theta_0 + (1/2+h_L)\cos(\theta_0) \\
&= (1/2+h_L)^2+1/4 + (1/2+h_L)\cos(\theta_0) .
\end{align*}
Computing the right hand side of \Eqref{moduli}:
\begin{align*}
	|1/2& + h_L + (1/2 + \varepsilon)e^{i \theta}|^2\\ \qquad & = 
	\left ( 1/2 + h_L + (1/2 + \varepsilon) \cos(\theta) \right)^2 + ((1/2+\varepsilon)\sin(\theta))^2 \\ 
	&= \left ( 1/2 + h_L \right)^2 + (1/2 + \varepsilon)^2\cos^2(\theta) 
	+(1/2 + \varepsilon)^2\sin^2(\theta) + 2(1/2+h_L)(1/2+\varepsilon)\cos(\theta)\\ 
	&= \left ( 1/2 + h_L \right)^2 + (1/2 + \varepsilon)^2 + 2(1/2+h_L)(1/2+\varepsilon)\cos(\theta).  
\end{align*}
Setting the simplified terms equal:
\begin{align*}
(1/2+h_L)^2+1/4 + (1/2+h_L)\cos(\theta_0)&=
\left ( 1/2 + h_L \right)^2 + (1/2 + \varepsilon)^2 + 2(1/2+h_L)(1/2+\varepsilon)\cos(\theta)\\
1/4 + (1/2+h_L)\cos(\theta_0)&=
1/4 + \varepsilon^2 + \varepsilon + 2(1/2+h_L)(1/2+\varepsilon)\cos(\theta)\\		
(1/2+h_L)(\cos(\theta_0)-\cos(\theta))&=
\varepsilon^2 + \varepsilon +2(1/2+h_L) \cdot \varepsilon\cos(\theta) \leq 4 \varepsilon.
\end{align*}
Recall that $ 0 < h_L <1/10$ and $ 0 < \theta_0 <\theta  \leq \pi/2 $. 
Then using standard identities, 
\begin{align*}
	(1/2+h_L)(\cos(\theta_0)-\cos(\theta)) &\leq 4 \varepsilon \\ 
	2(1/2+h_L) \cdot \sin((\theta-\theta_0)/2) \sin((\theta_0+\theta)/2)) &\leq 4 \varepsilon.
\end{align*}
Using the fact that $x^2/4 \leq \sin(x/2)$ for $0 \leq x \leq \pi/2$
and $\theta \geq \theta_0$, we see that
\begin{align*}
	 \frac{(\theta-\theta_0)^2}{4}\cdot \frac{(\theta_0+\theta)^2}{4} &\leq 4 \varepsilon\\
	 (\theta-\theta_0)^4&\leq 64 \varepsilon = 64(1-|w|).
\end{align*} 
\end{proof}

The following lemma and its proof are analogous to Lemma 3.1 in \cite{NazarovPeres17-WorseCase}.
\begin{lemma}\lemlab{SimpleWithLog}
Let $0 < q < 1/2$ be a constant. We have that 
\begin{align*}
    \sup_{w \in \gamma_L} \left|\widetilde{A}(w) \right| \geq \exp(-c \cdot d L) \cdot n^{-L},
\end{align*}
where $c$ is a constant depending only on $q$. 
\end{lemma}
\begin{proof}
Let $\lambda := \sup_{w\in\gamma_{L}} \left| \widetilde{A}(w) \right|$ to simplify notation. 
Define the following analytic function on $\mathbb{D}$:
\[
F(w) := \prod\limits_{j=0}^{L-1} \widetilde{A} \left(w \cdot e^{2 \pi i j/L} \right).
\]
Note that $F(w)$ is entire, as it is the product of polynomials. 
We bound $\sup_{w \in \partial \mathbb{D}}|F(w)|$ from above and below. 
For the upper bound, we use $\lambda$ for one of the factors, 
and for the other $L-1$ factors, we  use the following trivial bound. 
For $|w| \leq 1$, the moduli of both terms in the factored generating function are at most $1$, since $\left|w^d(1-q^d)+q^d \right| \leq q^d+(1-q^d)|w|^d \leq 1$, and so for $w \in \partial \mathbb{D}$ we have that $\big| \widetilde{A}(w) \big| \leq n$. 
Putting these together, we obtain that 
$|F(w)| \leq n^{L-1} \lambda$ 
for all $w \in \partial \mathbb{D}$.

To obtain a lower bound for $\sup_{w \in \partial \mathbb{D}}|F(w)|$ we use the maximum principle. Observe that $|F(0)| = \big|\widetilde{A}(0) \big|^{L}$. Since $F$ is analytic in $\mathbb{D}$, by the maximum modulus principle it must achieve modulus at least $\big|\widetilde{A}(0)\big|^{L}$ on $\partial \mathbb{D}$. Combining the upper and lower bounds on $ \sup_{w \in \partial \mathbb{D}}|F(w)|$, we see that $\big| \widetilde{A}(0) \big|^L \leq n^{L-1} \lambda$ and hence 
$\lambda \geq \big| \widetilde{A}(0) \big|^L n^{-L}$. 

It remains to lower bound $\big|\widetilde{A}(0)\big|$. 
From the definition of $\widetilde{A}$ we have that 
\[
\left| \widetilde{A}(0) \right| 
= (1-q) q^{-d \left\lfloor \frac{\ell^{*}}{d} \right\rfloor} \left| \sum_{\ell = \ell^{*}}^{n-1} a_{\ell} q^{\ell} \right|.
\] 
Recall from the definition of $\ell^{*}$ that 
$\left| a_{\ell^{*}} \right| = 1$ 
and that $a_{\ell} \in \left\{ -1, 0, 1 \right\}$ for all $\ell$. This implies that 
\[
\left| \sum_{\ell = \ell^{*}}^{n-1} a_{\ell} q^{\ell} \right| 
\geq 
q^{\ell^{*}} - \sum_{\ell = \ell^{*} + 1}^{n-1} q^{\ell} 
= 
q^{\ell^{*}} \left( 1 - \frac{q - q^{n-\ell^{*}}}{1-q} \right) 
\geq 
\frac{1-2q}{1-q} q^{\ell^{*}}.
\]
Our assumption that $q < 1/2$ implies that this lower bound is positive. 
Putting the previous two displays together and noting that $\ell^{*} - d \lfloor \ell^{*} / d \rfloor \leq d$, we have that 
\[
\left| \widetilde{A}(0) \right| 
\geq (1-2q) q^{\ell^{*} - d \left\lfloor \frac{\ell^{*}}{d} \right\rfloor} 
\geq (1-2q) q^{d}.
\]
Thus we have that 
$\lambda \geq (1-2q)^{L} q^{dL} n^{-L}$, 
which proves the claim with constant 
$c = - \log(q-2q^{2})$. 
\end{proof}

\section{Conclusion}\seclab{conclusion}

We introduced the problem of {tree trace reconstruction} and demonstrated, for multiple classes of trees, that we can utilize the structure of trees 
to develop more efficient algorithms than the current state-of-the-art for string trace reconstruction. 
We provided new algorithms for reconstructing complete $k$-ary trees and spiders in two different deletion models. For sufficiently small degree or large depth, we showed that a polynomial number of traces suffice to reconstruct worst-case trees.

\subsection{Future Directions} 

 \begin{enumerate}
 	\item \textbf{Improved bounds.} Can our existing sample complexity bounds be improved? Our results leave open several questions for complete $k$-ary trees and spiders. Of particular interest are 
(1) the TED model for complete $k$-ary trees with $\omega(1) \leq k \leq c \log^{2} n$ 
and 
(2) spiders with depth $d = c \log_{1/q} n$, $c < 1$; 
can we reconstruct with $\poly(n)$ traces in these cases? 
 	\item \textbf{General trees.} 
	We believe our results can extended to more general trees.
	In general, we do not know if the trace complexity can be bounded simply in terms of the number of nodes,
	 the depth, and the min/max degree of the tree. What other tree structure must we take into account for tight bounds?
 	\item \textbf{Lower bounds.} Lower bounds have recently been proven for string trace reconstruction~\cite{Chase19,HL18}. When can analogous bounds be proven for trees? For example, is it possible to reconstruct worst-case or average-case complete binary trees with $\mathrm{polylog}(n)$ traces?
 	\item \textbf{Insertions and substitutions.} We have focused on deletion channels, but insertions and substitutions are well-defined and relevant for tree edit distance applications. Similar to previous work, it would be worthwhile to understand the trace complexity for these edits. 
	\item \textbf{Applications.} Can insights from tree trace reconstruction be helpful in applications, for instance in computational biology? 
	In particular, DNA sequencing and synthesis techniques are rapidly evolving, and the future statistical error correction techniques will likely be different from the ones used currently. For instance, Anavy~et~al.~\cite{anavy2019data} recently demonstrated a new DNA storage method using composite DNA letters. Similarly, future DNA synthesis techniques may use physical constraints to enforce structure on the written bases; this could take the form of a two-dimensional array or a tree as we study. 
 \end{enumerate}

\subsection{Acknowledgments}
We thank Nina Holden for helpful discussions relating to \lemref{RidLogFactor}, 
and Bichlien Nguyen and Karin Strauss for pointing us to connections on branched DNA and recent work in this area. 
We also thank Alyshia Olsen for designing and creating the figures. 
Finally, we thank Tatiana Brailovskaya and an anonymous referee for their careful reading of the paper and their numerous helpful questions and suggestions that helped improve the paper. 

\newpage
\bibliographystyle{alpha}
\bibliography{tree-tr-paper}

\newpage
\appendix

\section{Appendix}\apndlab{TR_censoring}

\begin{remark}[Root node is fixed]\remlab{root}

	Our models implicitly enforce the property that the tree traces are connected (not forests). 
	This is consistent with the string case, because traces are never disjoint subsequences. 
	Moreover, in the TED model, having connected traces justifies the assumption that the root is never deleted.
	 If the root were deleted with probability $q$, then the preservation of the root could be achieved by sampling 
	 $O(1/q)$ times more traces and only keeping those that are connected. 
	 In the Left-Propagation model, our algorithms either already reconstruct the root node as written, 
	 or could be easily altered to learn the root. 
	 This may be least obvious for spiders, but here we consider the root to be in the first path and claim 
	that since the paths have monotonically decreasing length, our proof still holds.
	
	It is easy to imagine other models where this assumption would not work,
	 such as (i) allowing disconnected traces, (ii) deleting edges or subtrees, or (iii) sampling random subgraphs
	  or graph minors. We leave such investigations as future work.

\end{remark}

\subsection*{Trace reconstruction with censoring}

We analyze here a variant of string trace reconstruction where each trace is independently ``censored'' with some probability and instead of the actual trace we receive an empty string. 
In other words, we have to reconstruct the original string from a random sample of the traces. 
Here we reduce this problem to the original string trace reconstruction problem. 

Let $\mathbf{x} \in \left\{ 0, 1 \right\}^{n}$ denote the original string that we aim to reconstruct. 
Let $\emptyset$ denote the empty string, 
let $\mathcal{S}_{\geq 1} := \bigcup_{k \geq 1} \left\{ 0, 1 \right\}^{k}$ denote the set of binary sequences of finite positive length, and 
let $\mathcal{S} := \emptyset \cup \mathcal{S}_{\geq 1}$ denote the union of this set with the empty string. 
Let $P_{\mathbf{x},q}$ denote the probability measure on $\mathcal{S}$ that we obtain by passing $\mathbf{x}$ through the deletion channel with deletion probability~$q$. 
That is, if $\mathbf{Y}$ is a random (potentially empty) string that is obtained by passing $\mathbf{x}$ through the deletion channel, 
then $P_{\mathbf{x},q}( \mathbf{y}) = \mathbf{Pr} \left( \mathbf{Y} = \mathbf{y} \right)$ for every $\mathbf{y} \in \mathcal{S}$. 
Now, conditionally on $\mathbf{Y}$, define $\mathbf{Z}$ as follows: 
with probability $\gamma$, let $\mathbf{Z} = \emptyset$, 
and with probability $1- \gamma$, let $\mathbf{Z} = \mathbf{Y}$. 
That is, 
$\mathbf{Pr} \left( \mathbf{Z} = \emptyset \, \middle| \, \mathbf{Y} = \mathbf{y} \right) 
= \gamma 
= 1 - \mathbf{Pr} \left( \mathbf{Z} = \mathbf{y} \, \middle| \, \mathbf{Y} = \mathbf{y} \right)$. 
We call $\mathbf{Z}$ a \textit{censored trace} 
and $\gamma$ the \textit{censoring probability}.  
Let $P_{x,q,\gamma}$ denote the distribution of $\mathbf{Z}$ on $\mathcal{S}$. 

While in the string trace reconstruction problem we aim to reconstruct $\mathbf{x}$ 
from i.i.d.\ traces $\mathbf{Y}_{1}, \ldots, \mathbf{Y}_{k}$ from $P_{\mathbf{x},q}$, 
in the trace reconstruction problem with censoring we aim to reconstruct $\mathbf{x}$ 
from i.i.d.\ censored traces $\mathbf{Z}_{1}, \ldots, \mathbf{Z}_{k}$ from $P_{\mathbf{x},q,\gamma}$. 
Recall that $T(n,\delta)$ denotes the minimum number of traces needed to reconstruct a worst-case string on $n$ bits with probability at least $1-\delta$, where the dependence on the deletion probability $q$ is left implicit.

\begin{lemma}\lemlab{censored_TR}
Let 
$T_{\gamma}^{\textbf{cens}} \left( n, \delta \right)$ 
denote the minimum number of traces 
needed to reconstruct a worst-case string on $n$ bits with probability at least $1 - \delta$ 
from i.i.d.\ censored traces from the deletion channel with deletion probability $q$ and censoring probability $\gamma$. 
If $\liminf_{n \to \infty} \frac{1}{n} \log \delta > - \infty$, then 
for every $\eps > 0$ we have that 
$T_{\gamma}^{\textbf{cens}} \left( n, \delta \right) 
\leq \frac{1+\eps}{\left( 1 - q^{n} \right) \left( 1 - \gamma \right)} T \left( n, (1-\eps) \delta \right)$ 
for large enough $n$. 
\end{lemma}

\begin{proof}
First, note that empty strings contain no information, so reconstruction only uses nonempty traces. 
Next, observe that conditionally on outputting a nonempty trace, the deletion channel and the deletion channel with censoring have the same distribution. 
That is, if $\mathbf{x}$ is the original string, 
$\mathbf{Y}$ is a trace obtained by passing $\mathbf{x}$ through the deletion channel, 
and $\mathbf{Z}$ is a censored trace, then 
$\mathbf{Pr} \left( \mathbf{Y} = \mathbf{y} \, \middle| \, \mathbf{Y} \neq \emptyset \right) 
= \mathbf{Pr} \left( \mathbf{Z} = \mathbf{y} \, \middle| \, \mathbf{Z} \neq \emptyset \right)$ 
for every $\mathbf{y} \in \mathcal{S}_{\geq 1}$. 

Suppose that we have 
\[T := \frac{1+\eps}{\left( 1 - q^{n} \right) \left( 1 - \gamma \right)} T \left( n, (1-\eps)\delta \right)\]
censored traces and 
let $T'$ denote the number of these censored traces that are not empty. 
By our first two observations we have that if $T' \geq T(n,(1-\eps)\delta)$, then we can reconstruct the original string (no matter what it was) with probability at least $1 - (1-\eps)\delta$. 
By construction we have that 
$T' \sim \mathrm{Bin}\left( T, \left( 1 - q^{n} \right) \left( 1 - \gamma \right) \right)$ 
and thus $\mathbb{E} T' = (1+\eps) T(n,(1-\eps)\delta)$. 
Hence by a Chernoff bound we have that 
$\mathbf{Pr} \left( T' < T(n,(1-\eps)\delta) \right) \leq \exp \left( - \frac{\eps^{2}}{2(1+\eps)^{2}} \mathbb{E} T' \right) = \exp \left( - \frac{\eps^{2}}{2(1+\eps)} T(n,(1-\eps)\delta) \right)$. 
Since $T(n,(1-\eps)\delta) = \widetilde{\Omega} \left( n^{1.25} \right)$ (see~\cite{HL18}) and 
$\liminf_{n \to \infty} \frac{1}{n} \log \delta > - \infty$, 
it follows that 
\[\exp \left( - \frac{\eps^{2}}{2(1+\eps)} T(n,(1-\eps)\delta) \right) \leq \eps \delta\] for large enough $n$. 
The probability that we cannot reconstruct the original string is at most~$\delta$. 
\end{proof}
Note that this result is essentially optimal up to constant factors (depending on $q$, $\gamma$, and $\delta$). 
Note also that the range of $\delta$ for which the statement holds can be relaxed (this sufficient condition was chosen for its simplicity).

\end{document}

%% file: preamble.tex
\usepackage{subcaption}
\usepackage{color}
\usepackage{hyperref}
\usepackage[dvipsnames,svgnames,x11names,hyperref, enumerate]{xcolor}
\hypersetup{
	colorlinks,
	citecolor=Sepia,
	filecolor=red,
	linkcolor=[RGB]{0,0,215},
	urlcolor=Blue
}
\usepackage{graphicx}
\usepackage{amsmath, amssymb, amsthm} 
\usepackage{url}            
\usepackage{booktabs}       
\usepackage{amsfonts}       
\usepackage{algorithm}
\usepackage{algorithmicx}
\usepackage{algpseudocode}
\newcommand{\nth}[1]{\ensuremath{#1^{\,\mathrm{th}}}}
\newcommand{\avgT}{\widehat T}
\newcommand{\Mod}[1]{\ (\mathrm{mod}\ #1)}
\DeclareMathOperator*{\argmin}{arg\,min} 
\newtheorem{theorem}{Theorem}

\newtheorem{lemma}[theorem]{Lemma}

\newtheorem{claim}[theorem]{Claim}
\newtheorem{proposition}[theorem]{Proposition}

\newtheorem{definition}{Definition}

\newtheorem{remark}{Remark}
\def\Pr{\mathop{\mathbf{Pr}}\limits}
\renewcommand{\epsilon}{\varepsilon}

\newcommand{\eps}{\varepsilon}

\renewcommand{\leq}{\leqslant}
\renewcommand{\geq}{\geqslant}
\renewcommand{\a}{\alpha}

\newcommand{\bit}{\{0,1\}}

\newcommand{\A}{\mathcal{A}}

\newcommand{\calI}{\mathcal{I}}
\newcommand{\calJ}{\mathcal{J}}

\newcommand{\calE}{{\cal E}}

\newcommand{\poly}{\mathrm{poly}}

\newcommand{\pth}[1]{\left(#1\right)}%

\newcommand{\wt}{\widetilde}%
\newcommand{\remove}[1]{}

\renewcommand{\algref}[1]{\HLink{Algorithm}{Algorithm:#1}}

\newcommand{\HLinkShort}[2]{\hyperref[#2]{#1\ref*{#2}}}
\newcommand{\HLink}[2]{\hyperref[#2]{#1~\ref*{#2}}}
\newcommand{\HLinkPage}[2]{\hyperref[#2]{#1~\ref*{#2}%
		$_\text{p\pageref{#2}}$}}
\newcommand{\HLinkPageOnly}[1]{\hyperref[#1]{Page~\refpage*{#1}%
		$_\text{p\pageref{#1}}$}}

\newcommand{\HLinkSuffix}[3]{\hyperref[#2]{#1\ref*{#2}{#3}}}
\newcommand{\HLinkPageSuffix}[3]{\hyperref[#2]{#1\ref*{#2}%
		#3$_\text{p\pageref{#2}}$}}
	
\newcommand{\figlab}[1]{\label{fig:#1}}
\newcommand{\figref}[1]{\HLink{Figure}{fig:#1}}

\newcommand{\seclab}[1]{\label{section:#1}}
\newcommand{\secref}[1]{\HLink{Section}{section:#1}}

\providecommand{\deflab}[1]{\label{def:#1}}

\newcommand{\apndlab}[1]{\label{apnd:#1}}
\newcommand{\apndref}[1]{\HLink{Appendix}{apnd:#1}}

\newcommand{\clmlab}[1]{\label{claim:#1}}
\newcommand{\clmref}[1]{\HLink{Claim}{claim:#1}}

\newcommand{\remlab}[1]{\label{rem:#1}}
\newcommand{\remref}[1]{\HLink{Remark}{rem:#1}}%

\newcommand{\lemlab}[1]{\label{lemma:#1}}
\newcommand{\lemref}[1]{\HLink{Lemma}{lemma:#1}}%

\newcommand{\proplab}[1]{\label{prop:#1}}

\newcommand{\propref}[1]{\HLink{Proposition}{prop:#1}}

\newcommand{\alglab}[1]{\label{Algorithm:#1}}%

\newcommand{\thmlab}[1]{{\label{theo:#1}}}
\newcommand{\thmref}[1]{\HLink{Theorem}{theo:#1}}

\providecommand{\eqlab}[1]{}%
\renewcommand{\eqlab}[1]{\label{equation:#1}}
\newcommand{\Eqlab}[1]  {\label{equation:#1}}
\newcommand{\Eqref}[1]{\HLinkSuffix{Eq.~(}{equation:#1}{)}}